\documentclass[preprint,nocopyrightspace]{sigplanconf}

\usepackage{amsmath}

\usepackage{stmaryrd}
\usepackage{enumerate}
\usepackage{xspace}
\usepackage{amsmath} 
\usepackage{amssymb} 
\usepackage{mathpartir}
\usepackage{amsthm}
\usepackage{hyperref}

\usepackage{proof}
\usepackage{color}

\begin{document}

\setlength{\pdfpageheight}{\paperheight}
\setlength{\pdfpagewidth}{\paperwidth}

\conferenceinfo{CONF 'yy}{Month d--d, 20yy, City, ST, Country} 
\copyrightyear{20yy} 
\copyrightdata{978-1-nnnn-nnnn-n/yy/mm} 
\doi{nnnnnnn.nnnnnnn}



                                  
\newcommand{\mlab}[2]{\ensuremath{\{\!\!\{#1\}\!\!\}^{#2}}}
\newcommand{\glab}[2]{\ensuremath{\{\!\!\{#1\}\!\!\}_{#2}}}
\newcommand{\rv}[1]{\mathcal{#1}} 
\newcommand{\recc}{\mathbf{corec}\,}
\newcommand{\closetc}[1]{\ensuremath{\mathcal{C}_{#1}}}
\newcommand{\pceq}{\simeq_{c}}
\newcommand{\eqsw}{\ensuremath{\simeq_{\mathtt{sw}}}}
\newcommand{\encp}[1]{\llbracket #1 \rrbracket}
\newcommand{\proj}[2]{\ensuremath{#1\!\!\upharpoonright\!{#2}}}
\newcommand{\sproj}[2]{\ensuremath{#1\!\wr{#2}}}
\newcommand{\ether}[1]{\ensuremath{\mathsf{M}\encp{#1}}}
\newcommand{\retheraux}[3]{\ensuremath{\mathsf{\mathbf{M}}^\mu\encp{#1}^{#3}_{#2}}}
\newcommand{\rether}[2]{\ensuremath{\mathsf{\mathbf{M}}^\mu\encp{#1}_{#2}}}
\newcommand{\raether}[3]{\ensuremath{\mathcal{M}^\mu\encp{#1}_{#3}}}
\newcommand{\etherb}[2]{\ensuremath{\mathtt{M}\encp{#1}_{#2}}}
\newcommand{\SWF}{SWF\xspace}
\newcommand{\MWF}{MWF\xspace}
\newcommand{\lt}[1]{\ensuremath{\langle\!\langle #1\rangle\!\rangle}}
\newcommand{\clt}[3]{\ensuremath{\langle\!| #1|\!\rangle^{#2}_{#3}}}
\newcommand{\gt}[1]{\ensuremath{(\!| #1|\!)}}
\newcommand{\partp}[1]{\ensuremath{\mathsf{part}(#1)}}
\newcommand{\rpart}[1]{\ensuremath{\pi(#1)}}
\newcommand{\lend}{\ensuremath{\mathsf{end}}\xspace}
\newcommand{\gend}{\ensuremath{\mathtt{end}}\xspace}
\newcommand{\nnum}{\nonumber}
\newcommand{\gto}[2]{\ensuremath{\mathtt{#1}\!\twoheadrightarrow\!\mathtt{#2}{:}}}
\newcommand{\igto}[2]{\ensuremath{\mathtt{#1}\!\leadsto\!\mathtt{#2}{:}\,}}
\newcommand{\mywith}[2]{\ensuremath{\with\!\{#1\}_{#2}}}
\newcommand{\myoplus}[2]{\ensuremath{\oplus\!\{#1\}_{#2}}}
\newcommand{\mytau}[2]{\ensuremath{\boldsymbol{\tau}\{#1\}_{#2}}}
\newcommand{\mycase}[3]{\ensuremath{#1\triangleright\!\{#2\}_{#3}}}
\newcommand{\mycasecol}[3]{\ensuremath{\jblue{#1\,\triangleright}\{#2\}_{#3}}}
\newcommand{\mycaseu}[3]{\ensuremath{\underline{#1\triangleright}\{#2\}_{#3}}}
\newcommand{\mycasebig}[3]{\ensuremath{#1\triangleright\!\big\{#2\big\}_{#3}}}
\newcommand{\mycasebigcol}[3]{\ensuremath{\jblue{#1\,\triangleright}\big\{#2\big\}_{#3}}}
\newcommand{\mycaseb}[3]{\ensuremath{#1\triangleright\!\{#2\}_{}}}
\newcommand{\mysel}[2]{\ensuremath{#1\triangleleft\!#2}}
\newcommand{\lb}[1]{\ensuremath{\mathtt{#1}}}
\newcommand{\pt}[1]{\ensuremath{\mathtt{#1}}}
\newcommand{\merg}{\bowtie}
\newcommand{\fuse}{\sqcup}
\newcommand{\subt}{\preceq^\fuse}
\newcommand{\subts}{\preceq^\oplus}

\newtheorem{theorem}{Theorem}[section]
\newtheorem{remark}[theorem]{Remark}
\newtheorem{nota}[theorem]{Notation}
\newtheorem{proposition}[theorem]{Proposition}
\newtheorem{corollary}[theorem]{Corollary}

\newtheorem{example}[theorem]{Example}
\newtheorem{lemma}[theorem]{Lemma}
\newtheorem{definition}[theorem]{Definition}

\newcommand{\Names}{\Lambda}
\newcommand{\fn}[1]{\mbox{\it fn}(#1)}
\newcommand{\bn}[1]{\mbox{\it bn}(#1)}
\newcommand{\fv}[1]{\mbox{\it fv}(#1)}

\def\sep{  ~\mbox{\large{$\mid$}}~}
\newcommand{\caseopof}[1]{\textbf{(Case of #1)}}

\newcommand{\linkr}[2]{[#1\!\leftrightarrow\!#2]}
\newcommand{\pvar}[1]{\mathcal{#1}}
\newcommand{\para}{\mathord{\;\boldsymbol{|}\;}}
\newcommand{\zero}{{\bf 0}}
\newcommand{\nub}{{\boldsymbol{\nu}}}
\newcommand{\taub}{{\boldsymbol{\tau}}}
\newcommand{\mub}{{\boldsymbol{\mu}}}
\newcommand{\send}[2]{\mathord{#1\!\left<#2\right>}}
\newcommand{\inp}[3]{#1(#2).#3}
\newcommand{\out}[1]{ \ov{#1}\,}
\newcommand{\ov}[1]{\overline{#1}}
\newcommand{\name}[1]{\mbox{{{(#1)}}}}
\newcommand{\cbox}[1]{\vspace{1ex}\centerline{#1}\vspace{1ex}} 
\newcommand{\til}[1]{\widetilde{#1}}

\newcommand{\eqcc}{\ensuremath{\simeq_{c}}\xspace}
\newcommand{\typconp}{\eqcc}
\newcommand{\trelind}[4]{\ensuremath{#1 \vdash #3 \tybis   #4\,{::}\,#2}}
\newcommand{\tybis}{\ensuremath{\approx}\xspace}

\newcommand{\lol}[0]{\lolli} 
\newcommand{\one}{{\bf 1}}
\newcommand{\lolli}{\mathord{\multimap}}
\newcommand{\tensor}{\otimes}
\newcommand{\with}{\mathbin{\binampersand}}

\newcommand{\bang}{\mathbf{!}}

\newcommand{\lft}[1]{{{#1}\mathsf{L}}}
\newcommand{\rgt}[1]{{{#1}\mathsf{R}}}
\newcommand{\cut}{\mathsf{cut}}
\newcommand{\cpy}{\mathsf{copy}}
\newcommand{\cutbang}{{\mathsf{cut}^\bang}}
\newcommand{\D}{\Delta}
\newcommand{\G}{\Gamma}

\newcommand{\cpar}{{\mathsf{comp}}}

\newcommand{\outp}[2]{\ov{#1}(#2)}
\newcommand{\about}[1]{\widehat{#1}}
\newcommand{\outa}[2]{\ov{#1}(#2)}

\def\subst#1#2{\{\raisebox{.5ex}{\small$#1$}\! / \mbox{\small$#2$}\}}
\newcommand{\red}{\rightarrow}

\newcommand{\aand}{{,\;}}

\newcommand{\tra}[1]{\xrightarrow{#1}}
\newcommand{\wtra}[1]{\stackrel{#1}{\Longrightarrow}}
\newcommand{\labelset}{\ensuremath{\lambda}}
\newcommand{\gtlabelset}{\ensuremath{\sigma}}

\newtheorem{myfact}[theorem]{Fact}
\newtheorem{convention}[theorem]{Convention}


\conferenceinfo{CONF'15}{XXX} 
\copyrightyear{2014} 
\copyrightdata{[to be supplied]} 

\titlebanner{DRAFT}        
\preprintfooter{DRAFT}   

\title{
A Typeful Characterization of Multiparty Structured Conversations Based on 
Binary Sessions
}

\authorinfo{Lu\'{i}s Caires}
           {Universidade Nova de Lisboa}
           {lcaires@fct.unl.pt}
\authorinfo{Jorge A. P\'{e}rez}
   {University of Groningen}
   {j.a.perez@rug.nl}

\maketitle

\begin{abstract}
Relating the specification of the global communication behavior of a distributed system and  
the specifications of the local 
communication behavior of each of its nodes/peers (e.g., to check if
the former is realizable by the latter under some safety and/or liveness conditions) is a challenging
problem addressed in many relevant scenarios. 
In the context of networked software services, 
a widespread programming language-based approach relies on global specifications defined by session types or behavioral contracts.  
Static type checking can then be used to ensure that components follow the prescribed interaction protocols.
In the case of 
session types, developments  have  been mostly framed within quite different type theories for 
either   \emph{binary} (two-party) or \emph{multiparty} ($n$-party) protocols.
Unfortunately, the precise relationship between
analysis techniques for 
multiparty and binary protocols is yet to be understood. 

In this work, we bridge this previously open gap in a principled way: we show that the  analysis of multiparty protocols can also be developed within a
much simpler type theory for binary protocols, ensuring protocol fidelity and deadlock-freedom.
We present characterization theorems which provide new insights on the relation between
two existing, yet very differently motivated, session type systems---one based on linear logic,
the other based on automata theory---and suggest useful type-based verification techniques for
 multiparty systems relying on reductions to the binary case.
\end{abstract}




\section{Introduction}

Relating the global specification of a distributed system and the set of components that implement such a specification is a problem found in many relevant scenarios. 
For instance, 
the analysis of security protocol narrations relies on formal correspondences between a global protocol specification and implementations for the each roles/principals in the protocol (see, e.g.,~\cite{DBLP:conf/esorics/McCarthyK08}).
Also, the formal connection between  Message Sequence Charts (MSCs) and  Communicating Finite State Machines (CFSMs) has been thoroughly studied (see, e.g.,~\cite{DBLP:conf/acsd/GenestM05}). A recent survey by Castagna et al.~\cite[\S7]{CastagnaDezaniPadovani12} contrasts these two scenarios.

Establishing relationships between global and local specifications is also important in
the analysis of networked software services.
In this context, the emphasis is in the analysis of message-passing, communication-oriented programs, which often feature advanced forms of concurrency, distribution, and trustworthiness. Within programming language-based techniques, 
notable approaches include \emph{interface contracts} (see, e.g.,~\cite{CastagnaGesbertPadovani08}) and 
\emph{behavioral types}~\cite{DBLP:conf/concur/Honda93,DBLP:conf/tgc/BonelliC07,DBLP:conf/esop/HondaVK98,DBLP:journals/tcs/IgarashiK04,DBLP:conf/popl/HondaYC08,DBLP:journals/tcs/CairesV10,DBLP:conf/popl/CarboneM13}. 
Our focus is in the latter, often defined on top of core programming calculi (usually, some dialect of the $\pi$-calculus~\cite{MilnerPW92a}) which specify the inherently interactive nature of communication-based systems.
By classifying behaviors (rather than values), behavioral types define abstract descriptions of structured protocols, and 
enforce disciplined exchanges of values and communication channels. 
Advantages of the behavioral types approach include simplicity and flexibility; successful integrations of behavioral types into functional and object-oriented languages~\cite{DBLP:journals/eatcs/Vasconcelos11} offer compelling evidence of these benefits. A variety of  frameworks based on behavioral types has been put forward, revealing a rather rich landscape of models and languages in which communication is delineated by types. 
In particular, frameworks 
based on \emph{session types}~\cite{DBLP:conf/concur/Honda93,DBLP:conf/esop/HondaVK98}
have received much attention in  academic and applied contexts.
In these models, multiparty conversations are organized as concurrent \emph{sessions}, which define structured dialogues. 
Unfortunately, the diversity of underlying models and techniques makes
formally relating independently defined typed frameworks a challenging task.
This limitation 
makes it difficult to objectively compare the expressiveness and significance of seemingly related behavioral type theories. 
Also, it 
hinders the much desirable transfer of reasoning/verification techniques 
between different typed frameworks---a notable example being techniques for ensuring deadlock-free protocols (see below).

In this paper, we identify formal relationships between 
two  distinct typed frameworks for structured communications.
Precisely, we establish natural and useful bridges between typed models for \emph{binary} and \emph{multiparty} session communications~\cite{DBLP:conf/esop/HondaVK98,DBLP:conf/popl/HondaYC08}: by relying on a theory of binary session types rooted on linear logic~\cite{CairesP10,Toninho2014}, we elucidate new deep foundational connections between both frameworks. 
Our results not only reveal new logically motivated justifications for key concepts
in models of typeful specifications of global and local behaviors;
they also 
 enable the principled transfer of useful reasoning techniques. 
As we argue next, this is rather significant for session types, 
as (well-understood) techniques in the binary setting are usually hard to generalize to multiparty sessions.

In binary communications~\cite{DBLP:conf/esop/HondaVK98} protocols involve exactly two partners, each one abstracted by a behavioral type; 
correct interactions  depend on 
 \emph{compatibility}, which intuitively means that when one partner 
 performs an action (e.g., send), the other performs a complementary one (e.g., receive).
In the 
multiparty setting, protocols may involve an arbitrary number of 
partners.
There is a global specification 
(or \emph{choreography}) to which all  partners, from their local perspectives, should adhere.
In multiparty session types~\cite{DBLP:conf/popl/HondaYC08,DBLP:conf/popl/DenielouY11,DBLP:conf/icalp/DenielouY13},
these two visions are described by
a \emph{global type} and \emph{local types}, respectively;  a \emph{projection function}  
formally relates the two. 
The expressiveness jump from binary to multiparty communications is  significant.
Extensive research has shown that type systems for multiparty communication have a much more involved theory than binary  ones.
In session types, e.g., 
this shows up in basic concepts such as 
compatibility:
while binary compatibility can be simply characterized as \emph{type duality}~\cite{DBLP:conf/esop/HondaVK98},
a formal characterization of multiparty compatibility was given only recently~\cite{DBLP:conf/icalp/DenielouY13}.
Also, some analysis techniques, such as those for  deadlock-freedom, 
are a difficult issue in the multiparty setting, and certainly harder than in the binary case.
The question that arises
is then: 
could multiparty session types be 
reduced, in some precise sense, into  binary ones?
As discussed in~\cite{DBLP:conf/popl/HondaYC08}, 
such a reduction entails decomposing
a single global specification into several binary fragments. 
Defining 
such a decomposition is far from trivial, for it should satisfy at least two  requirements.
First, it must preserve crucial \emph{sequencing information} among multiparty exchanges.
Second, 
the resulting collection of binary interactions should not exhibit 
\emph{undesirable behaviors}, such as, e.g., 
synchronization errors, 
unspecified protocol steps,
deadlocks, 
and unproductive, non-terminating reductions.

This paper answers the above question in the affirmative.
We exhibit a tight correspondence between:
\begin{enumerate}[$\bullet$]
\item a standard theory of multiparty session types
(as first formulated 
by Honda, Yoshida, and Carbone~\cite{DBLP:conf/popl/HondaYC08}
and recently characterized via communicating automata 
by Deni{\'e}lou and Yoshida~\cite{DBLP:conf/icalp/DenielouY13}), and
\item the theory of deadlock-free binary session types proposed by Caires and Pfenning in~\cite{CairesP10},
which interprets linear logic propositions as session types, in the 
style of Curry-Howard. 
\end{enumerate}

\noindent We  briefly motivate our approach.
Deni{\'e}lou and Yoshida have
recently identified the set of global types that admit a sound and complete characterization of
multiparty compatibility in terms of communicating automata~\cite{DBLP:conf/icalp/DenielouY13}.
This suggests whether 
multiparty communication 
could be related to binary communication. 
An initial observation is that 
a communication from  $\pt{p}$ to  $\pt{q}$
can be \emph{decoupled} into two simpler steps: 
a \emph{send} action from $\pt{p}$ to an intermediate entity, 
followed by a step in which the 
entity
\emph{forwards} the message to $\pt{q}$.
Our approach 
generalizes 
this observation:
given a multiparty conversation 
$G$ (a global type),  we extract its semantics in terms of a \emph{medium process}, denoted \ether{G}, 
an intermediate entity in all protocol exchanges (\S\,\ref{s:medium}).
Process \ether{G}  may uniformly capture all the sequencing information prescribed by $G$. 
The next step 
for a full characterization is determining the conditions under which \ether{G}
may be well-typed in a theory for binary session types,
with respect to  the local types for $G$ (i.e., its associated projections).
A key  ingredient in our characterization is the 
theory for binary session types introduced in~\cite{CairesP10} and extended in~\cite{Toninho2014}.
Due to their logical foundations, 
well-typed processes are naturally 
type preserving; this entails  
 \emph{fidelity} (i.e., protocols are respected) and \emph{safety} (i.e., absence of runtime communication errors).
Moreover, well-typed processes are 
\emph{deadlock-free} (i.e., processes do not get stuck) and  \emph{compositionally non-diverging} (i.e., infinite observable behavior is allowed, but unproductive, infinite internal behavior is ruled out). 
Particularly relevant for our approach is deadlock-freedom, not 
directly ensured by alternative type systems.
Our core developments rely on tight relationships between global and binary session types: 
\begin{enumerate}[a)]
\item For any  global type $G$ which is \emph{well-formed} (in the sense of Deni{\'e}lou and Yoshida~\cite{DBLP:conf/icalp/DenielouY13}), a medium process $\ether{G}$ can be constructed such that $\ether{G}$ is well-typed in the binary session type system of~\cite{CairesP10}, under a typing environment
in which participants are assigned binary types that correspond to the expected 
 projections of $G$ onto all participants.
 \item Conversely, for any $G$ 
such that 
$\ether{G}$ is well-typed in the type system of~\cite{CairesP10,Toninho2014}, under a typing environment assigning some binary session types for participants, such binary types 
correspond, in a very precise sense, to the projections of $G$.
\end{enumerate}

\noindent Notice that (a) 
immediately provides an alternative proof of global progress/deadlock-freedom for global systems as in \cite{DBLP:conf/popl/HondaYC08,DBLP:conf/icalp/DenielouY13}. 
Moreover, given the uniform definition of medium processes, 
our results immediately apply to more expressive systems, in particular supporting name passing, session delegation, 
and parallel composition of global types, all of these features being beyond the scope of~\cite{DBLP:conf/icalp/DenielouY13}.
It is also worth highlighting that, 
unlike, e.g., \cite{DBLP:conf/icalp/DenielouY13},
our medium characterization of multiparty session types s
bridges all the way from global types to actual 
processes implementing the system.
This allows us to explore known properties of the underlying typed processes, in particular, behavioral equivalences~\cite{DBLP:conf/esop/PerezCPT12} (see below), to reason about (and justify) properties of multiparty systems.

\paragraph{Contributions.}
This paper offers the following contributions.

\begin{enumerate}[1.]
\item We offer 
an analysis 
of multiparty session types using
a theory of 
binary session types, ensuring 
fidelity and deadlock-freedom. 
We give a two-way correspondence relating well-formed global types 
with typability of its associated medium  on a logically motivated theory of binary sessions.
This results holds both for global types without recursion (but with parallel composition, cf. Theorems~\ref{l:ltypesmedp} and \ref{l:medltypes}) and for global types with recursion (and without parallel composition, exactly the same type structure studied in~\cite{DBLP:conf/icalp/DenielouY13}, cf. Theorem~\ref{l:ltypesmedprec} and~\ref{l:medltypesrec}).

\item 
For global types without recursion (but with parallel composition), we  show how known typed behavioral equivalences for binary sessions~\cite{DBLP:conf/esop/PerezCPT12} may be used to justify behavioral transformations of global types, expressed at the level of mediums (Theorem~\ref{p:swapmeds}). This
provides a deep semantic justification of useful structural identities on global types, such as those capturing parallelism 
via commutation or interleaving of causally independent  communications~\cite{DBLP:conf/popl/CarboneM13}.


\item For global types with recursion (and without parallel composition),  we prove an operational correspondence result relating the behavior 
of a global type with the observable behavior of the composition of (a)~its medium process (instrumented in a natural way) and (b)~arbitrary 
implementations for local participants~(Theorem~\ref{th:opcorr}).
This confirms that mediums faithfully mirror global specifications.
That is, going through the medium does not introduce extra sequentiality in protocols, as mediums are essentially transparent from an operational standpoint.

\item We show how an approach based on mediums allows to effectively transfer techniques from binary sessions to multiparty session types. We show how to enrich global type specifications with an expressive \emph{join primitive}. Also, we describe how to analyze global types with \emph{parametric polymorphism} based on already existing theories in the binary setting (\S\,\ref{s:exts}). The latter is remarkable, for we do not know of multiparty session type theories supporting parametric polymorphism.

  \end{enumerate}

\noindent Our results not only 
clarify 
the relationship between multiparty and binary session types. They also give further evidence on the fundamental character
of the notions involved (projection functions, type primitives, multiparty compatibility),
since they can be independently and precisely explained within the separate worlds of
communicating automata and linear logic.

Next we illustrate our approach and key results by means of an example.
\S\,\ref{s:prelim} collects definitions on multiparty and binary sessions.
Our main results are reported in \S\,\ref{s:tmedium}.
Further illustration is given in \S\,\ref{s:exx}. 
In \S\,\ref{s:exts} we describe extensions to our approach, while 
\S\,\ref{s:relwork} discusses related works.
\S\,\ref{s:concl} collects some concluding remarks.

\emph{An accompanying appendix gives details of omitted definitions and proofs.}

\section{Our Approach and Main Results, By Example}\label{s:example}

\newcommand{\angg}[5]{\gto{#1}{#2}\{\lb{#3}\langle \mathsf{#4}\rangle.#5\}\!}
\newcommand{\mangg}[5]{\gto{#1}{#2}\big\{\lb{#3}\langle \mathsf{#4}\rangle.#5\big\}_{}}
\newcommand{\anggb}[5]{\gto{#1}{#2}\{\lb{#3}\langle \mathsf{#4}\rangle. ~~~~~~~~  \\ ~~ #5\}\!}
\newcommand{\anggd}[8]{\gto{#1}{#2}\{\lb{#3}\langle \mathsf{#4}\rangle.#5 \, , \, \lb{#6}\langle \mathsf{#7}\rangle.#8 \}\!}

We now elaborate further on our contributions by 
illustrating how our approach bridges all the way from global types (choreographies) to actual $\pi$-calculus processes
which realize multiparty scenarios and inherit key properties from binary session typing, most notably, protocol fidelity and freedom from deadlocks. 

We illustrate the multiparty session types approach by an example.
In this paper, we consider the following syntax of global types:
\[
\begin{array}{rll}
 G ::= & \gto{p}{q}\{\lb{l}_i\langle U_i\rangle.G_i\}_{i \in I} &  \text{communication from $\pt{p}$ to $\pt{q}$} \\
\sep & G_1 \para G_2 & \text{composition of global types}\\
\sep & \mu\rv{X}.\,G  \sep \rv{X} & \text{recursive global type} \\
\sep& \gend & \text{terminated global type}
\end{array}
\]
Consider global type $G_c$ below, which abstracts a  
variant of the \emph{commit protocol} in~\cite{DBLP:conf/icalp/DenielouY13}.
The protocol involves three participants:  \pt{A}, \pt{B}, and \pt{C}. First, 
\pt{A} orders to \pt{B} to \emph{act} or to \emph{quit}. 
In the first case, \pt{B} sends a \emph{signal} to \pt{C} and subsequently \pt{A} sends to \pt{C}
a \emph{commit} order. In the second case, \pt{B} orders to \pt{C} to \emph{save}, then \pt{A} orders \pt{C} to \emph{finish}:
\begin{align*}
G_c = \gto{A}{B}\{  \lb{~act}\langle \mathsf{int} \rangle.\gto{B}{C}\{\lb{sig}\langle \mathsf{str}\rangle.
\gto{A}{C}\{\lb{comm}\langle \one 
\rangle.\gend\} \}  , \\
  \lb{quit}\langle \mathsf{int} \rangle.\gto{B}{C}\{\lb{save}\langle \one \rangle.
 \gto{A}{C}\{\lb{fini}\langle \one \rangle.\gend\} \} ~\} 
\end{align*}
Given a global specification such as $G_c$, 
to derive code for each participant in the choreography, one first extracts a set of so-called \emph{local types}. 
Formally, these local types are obtained using a \emph{projection function} (cf. Def.~\ref{d:proj}). 
This way, e.g.,
the local behavior for  $\pt{A}$ and $\pt{C}$,
denoted $\proj{G_c}{\pt{A}}$ and $\proj{G_c}{\pt{C}}$, 
 is given by the local types:
$$
\!\!\!\!\!\!\begin{array}{rl}
\proj{G_c}{\pt{A}}  & = \mathtt{A}!\big\{\lb{act}\langle \mathsf{int}\rangle.  \mathtt{A}! \{\lb{comm}\langle \one 
\rangle.\lend \} ,  \\
&  \qquad \quad \lb{~quit}\langle \mathsf{int}\rangle.  \mathtt{B}! \{\lb{sig}\langle \mathsf{str} \rangle.\lend \}\, \big\} \\
\proj{G_c}{\pt{C}} & =  \mathtt{B}?\big\{\lb{sig}\langle \mathsf{str}\rangle.  \mathtt{A}? \{\lb{comm}\langle \one 
\rangle.\lend \} ,  \\
&  \qquad \quad \lb{~save}\langle \one \rangle.  \mathtt{A}?\{\lb{fini}\langle \one  \rangle.\lend \}\, \big\} 
\end{array}
$$
The local type $\proj{G_c}{\pt{B}}$ is obtained similarly.
Not all global types are meaningful: following~\cite{DBLP:conf/icalp/DenielouY13},
we focus on 
\emph{well-formed global types}---global types which have a well-defined local type for each participant.
For well-formed global types, one may  independently obtain  implementations for each participant; then, using 
associated type systems~\cite{DBLP:conf/popl/HondaYC08}, it can be ensured that such 
implementations are type-safe and
realize the intended global scenario.
Using complementary techniques one may guarantee other advanced properties, such as deadlock-freedom in interleaved sessions (cf.~\cite{DBLP:conf/coordination/CoppoDPY13}).
Also, using the correspondence with communicating automata~\cite{DBLP:conf/icalp/DenielouY13}, safety and liveness guarantees carry over to choreographies. 

In this paper, we propose analyzing global types  by means of an existing type system for binary sessions, which 
defines a Curry-Howard isomorphism between intuitionistic linear logic propositions and session types~\cite{CairesP10,Toninho2014}.
The analysis that we propose is thus endowed with a deep foundational significance.
The main conceptual device in our approach is the  \emph{medium process} of a global type (or simply \emph{medium}). Extracted following the syntax of the global type, a medium  process is intended to interact with all the implementations which conform to the local types obtained via projection. Before describing the medium for $G_c$ above, let us give some intuitions on our process syntax---a standard $\pi$-calculus with $n$-ary labeled choice (cf. Def.~\ref{d:procs}). 
We write $a(v)$ and $\outp{b}{w}$ to denote, respectively, 
 prefixes for \emph{input} (along name $a$, with $v$ being a placeholder)
and \emph{bound output} (along name $b$, with $w$ being a freshly generated name).
Labeled choice is specified using
processes $\mycase{a}{\lb{label_1}:P_1, \cdots, \lb{label_n}:P_n}{}$ (branching)
and $\mysel{a}{\lb{label_j}};Q$ (selection).
The interaction of input and output prefixes (resp. branching and selection constructs) defines a \emph{reduction}.
Process $\linkr{u}{y}$ equates names $u$ and $y$, while 
$(\nub x)P$ and 
$P \para Q$ denote the usual restriction and parallel composition constructs.

The medium  of a global type is a process in which every directed labeled communication is captured by its decomposition into simpler prefixes. This way, e.g. the medium of $G_c$, denoted $\ether{G_c}$, is as in Fig.~\ref{f:exam},
\begin{figure}
$$
\begin{array}{ll}
\mycasebig{a} { & \!\!\!\lb{act} : a(v).\mysel{b}{\lb{act}};\outp{b}{w}.(\linkr{w}{v} \para \\
		& \quad ~~ \mycaseb{b}{\lb{sig} : b(n).\mysel{c}{\lb{sig}};\outp{c}{m}.(\linkr{n}{m} \para \\
		& \quad \quad  \mycaseb{a}{\lb{comm}: a(u).\mysel{c}{\lb{comm}};\outp{c}{y}.(\linkr{u}{y} \para \zero) 
 		}{} )     }{}\, )~~,  \\ 
		& \!\!\!\lb{quit} : a(v).\mysel{b}{\lb{quit}};\outp{b}{w}.(\linkr{w}{v} \para \\
		& \quad ~~  \mycaseb{b}{\lb{save} : b(n).\mysel{c}{\lb{save}};\outp{c}{m}.(\linkr{n}{m} \para \\
		& \quad \quad   \mycaseb{a}{\lb{fini}: a(u).\mysel{c}{\lb{fini}};\outp{c}{y}.(\linkr{u}{y} \para \zero) 
 		}{}  )     }{} \, ) }{} \\ 
\end{array}
$$
\caption{Medium process $\ether{G_c}$ for global type $G_c$\label{f:exam}}
\vspace{-4mm}
\end{figure}
\noindent
where we have used names $a, b$, and $c$ to indicate interactions associated to \pt{A}, \pt{B}, and \pt{C}, respectively.
With this in mind, we may now informally explain  how the medium $\ether{G_c}$ gives semantics to the global type $G_c$. Consider the first labeled communication in $G_c$, in which $\pt{A}$ sends a value to $\pt{B}$ by selecting a label $\lb{act}$ or $\lb{quit}$.
We assume that process implementations for $\pt{A}$ and $\pt{B}$ are available on names $a$ and $b$, respectively. 
The implementation for $\pt{A}$ should first select a label (say $\lb{act}$) and then output a message for $\pt{B}$ (say, a name $y$ denoting a reference to an integer).
 Accordingly, the first two actions of $\ether{G_c}$ are on name $a$: it first 
 commits to the labeled alternative $\lb{act}$ and then it receives $y$.
This completes the involvement of $\pt{A}$ in the communication to $\pt{B}$.
Subsequently, the medium acts on name $b$, first selecting label $\lb{act}$ in the implementation of 
$\pt{B}$ and then sending a fresh name $w$. Then, names $w$ and $y$ (i.e., the one received from $\pt{A}$) are ``linked'' together,  and 
the medium  for the continuation of the global protocol is spawned.

Mediums give a simple characterization of global types; intuitively, they define the code for ``gluing together''  the behavior of all local participants. However, a medium by itself does not relate  
a global type and the local types obtained by projection.
Giving a logically motivated justification for this connection is the central contribution of this paper.
To this end,
we rely on the theory of binary session types developed in~\cite{CairesP10,Toninho2014}, which  connects intuitionistic linear logic propositions and session types.
This way, e.g., session types 
$!A; B$ and $?A;B$, introduced in \cite{DBLP:conf/esop/HondaVK98} to abstract input and output, 
are represented in~\cite{CairesP10,Toninho2014} by the tensor $A\otimes B$ and the linear arrow $A\lol B$.
More precisely, we have:
$$
\begin{array}{rll}
A,B,C  ::= & \one  & \text{terminated binary session}\\
  \sep & A \lol B  \, \sep A \otimes B & \text{input, output}\\ 
       \sep &   A \with B  \sep   A \oplus B & \text{branching and selection} \\
   \sep & ! A     ~~\sep   \nu \rv{X}.A  \sep \rv{X} &  \text{replicated and  coinductive type }
\end{array}
$$
\noindent 
We use the expected extension of the binary operators $\with$ and $\oplus$ to the $n$-ary case. 
Type assignments are of the form
$x{:}A$, where $x$ is a name and $A$ a session type. Given a process $P$,
the typing judgment 
$$\Gamma ; \Delta \vdash_\eta P :: z{:}C$$
asserts that $P$ provides a behavior described by session type $C$ at
channel/name $z$, building on \emph{linear} session behaviors declared in type environment $\Delta$ and \emph{unrestricted} (or \emph{persistent}) behaviors declared in type environment $\Gamma$. 
The map $\eta$ relates (corecursive) type variables to typing contexts~\cite{Toninho2014}.
Thus, judgments in~\cite{CairesP10,Toninho2014} specify both 
the behavior that a process \emph{offers} (or \emph{implements}) and 
the  (unrestricted, linear) behaviors that it \emph{requires} to do so.
As discussed above,  
well-typed processes in this system are naturally type-preserving and deadlock-free.
They are also non-diverging. 

Our characterization results concern well-typedness of mediums 
in the logic-based theory of binary sessions.
In fact, Theorems~\ref{l:ltypesmedp} and~\ref{l:medltypes} 
ensure that the medium process $\ether{G_c}$ given above is typable as follows:
\begin{equation}\label{eq:preview}
\cdot \,;\, a{:}\lt{\proj{G_c}{\pt{A}}},\,  b{:}\lt{\proj{G_c}{\pt{B}}},\,  c{:}\lt{\proj{G_c}{\pt{C}}} \vdash_\eta \ether{G_c} :: z : \one
\end{equation}
where 
`$\cdot$' stands for the empty (shared) environment and
$\lt{\cdot}$ relates local types and binary session types (cf. Def.~\ref{d:loclogt}).
Well-typedness of $\ether{G_c}$ as captured by~\eqref{eq:preview} has  significant consequences: 
\begin{enumerate}[a.]
\item It formalizes the  dependence of the medium on behaviors with local types defined by projection:~\eqref{eq:preview} says that
$\ether{G_c}$ requires exactly behaviors 
$\lt{\proj{G_c}{\pt{A}}}$,  $\lt{\proj{G_c}{\pt{B}}}$, and $\lt{\proj{G_c}{\pt{C}}}$, which should be available, as linear resources, along names $a$, $b$, and $c$.

\item The judgement also ensures that the medium does not add extraneous behaviors: since 
the offered behavior of $\ether{G_c}$ (in the right-hand side) is $\one$, we are sure that it  acts  as a faithful mediator among the behaviors described in the left-hand side. 

\item By well-typedness, \ether{G_c} 
inherits 
type preservation, deadlock-freedom, and non-divergence as ensured by the type system in~\cite{CairesP10,Toninho2014}.
This not only means that $\ether{G_c}$ in isolation behaves as intended. 
Consider processes $P_\pt{A}$, $Q_\pt{B}$, and $R_\pt{c}$ which 
implement
$\lt{\proj{G_c}{\pt{A}}}$,  $\lt{\proj{G_c}{\pt{B}}}$, and $\lt{\proj{G_c}{\pt{C}}}$
in appropriate names:
\begin{eqnarray*}
\cdot \, ; \cdot \vdash_\eta  P_\pt{A} :: a{:}\lt{\proj{G_c}{\pt{A}}}
\quad
\cdot \, ; \cdot \vdash_\eta  Q_\pt{B} :: b{:}\lt{\proj{G_c}{\pt{B}}}
\\
\cdot \, ; \cdot \vdash_\eta  R_\pt{C} :: c{:}\lt{\proj{G_c}{\pt{C}}}
\end{eqnarray*}
where, for  simplicity, we have assumed no linear/shared dependencies.
These processes can be constructed independently from (and unaware of) $\ether{G_c}$, and type-checked in the system of~\cite{CairesP10,Toninho2014}, inheriting all key properties. 
Moreover, 
the composition of $\ether{G_c}$ with such processes (i.e., a \emph{system} realizing $G_c$)
$$
(\nub c)((\nub b)((\nub a)(\ether{G_c} \para P_\pt{A}) \para Q_\pt{B}) \para R_\pt{C})
$$
is also well-typed, and therefore
type-preserving, deadlock-free, and non-diverging. Deadlock-freedom can be seen to be crucial in ensuring  faithful sets of interactions between the local implementations $P_\pt{A}$, $Q_\pt{B}$, $R_\pt{C}$ and the medium $\ether{G_c}$.
\end{enumerate}
\noindent To further support point (b) above, we define the \emph{annotated medium} of $G_c$, denoted $\raether{G_c}{k}{k}$ (cf. Def.~\ref{d:raether}).
This process extends \ether{G_c} with visible actions on name $k$ which mimic those in $G_c$. Our main  results extend smoothly to annotated mediums, and we may derive the following:
 \begin{equation}\label{eq:preview2}
\cdot \, ; a{:}\lt{\proj{G_c}{\pt{A}}},  b{:}\lt{\proj{G_c}{\pt{B}}},  c{:}\lt{\proj{G_c}{\pt{C}}} \vdash_\eta \raether{G_c}{k}{k} :: k{:} \gt{G_c}
\end{equation}
where \gt{G_c} denotes a session type which captures the sequentiality of $G_c$ (cf. Def.~\ref{d:glolog}). Building upon~\eqref{eq:preview2}, 
we may state a rather strong result of operational correspondence relating global types and annotated mediums, which is given by Theorem~\ref{th:opcorr}. Roughly speaking, such a result identifies the conditions under which each step of type $G_c$ (as formalized by a simple LTS for global types) can be matched by a labeled transition of  process \raether{G_c}{k}{k}.

\section{Preliminaries: Multiparty and Binary Sessions}\label{s:prelim}
Here we collect key definitions for multiparty session types, as
presented in~\cite{DBLP:conf/popl/HondaYC08,DBLP:conf/icalp/DenielouY13} (\S\,\ref{ss:multip}).
We also summarize the key concepts of the logically motivated theory of binary sessions,
as first introduced in~\cite{CairesP10} and extended with coinductive session types in~\cite{Toninho2014} (\S\,\ref{ss:btypes}).

\subsection{Multiparty Session Types}\label{ss:multip}
Our syntax of global and local types subsumes constructs from the original presentation~\cite{DBLP:conf/popl/HondaYC08}
and from recent formulations~\cite{DBLP:conf/icalp/DenielouY13}.
With respect to~\cite{DBLP:conf/popl/HondaYC08}, we 
consider labeled communication, recursion, and
retain parallel composition, which 
enables  compositional reasoning over global specifications.
With respect to~\cite{DBLP:conf/icalp/DenielouY13},
we consider value/session passing 
in branching (cf. $U$ below) and add parallel composition.
Below, \emph{participants} and \emph{labels}
are ranged over by $\pt{p},  \pt{q}, \pt{r}, \ldots$.
and $\lb{l}_1, \lb{l}_2, \ldots$, respectively. 
\begin{definition}[Global and Local Types]\label{d:gltypes}
Define global and local types as
{
\begin{eqnarray*}
&G  ::= & \!\!\!\!\gend \sep \gto{p}{q}\{\lb{l}_i\langle U_i\rangle.G_i\}_{i \in I} \sep G_1 \para G_2 \sep \mu\rv{X}.\,G \sep \rv{X}\\ 
&U  ::=  & \!\!\!\!\mathsf{bool} \sep \mathsf{nat} \sep \mathsf{str} \sep\ldots \sep T\\
&T  ::= & \!\!\!\!\lend \sep \mathtt{p}?\{\lb{l}_i\langle U_i\rangle.T_i\}_{i \in I} \sep \mathtt{p}!\{\lb{l}_i\langle U_i\rangle.T_i\}_{i \in I}  \sep \mu\rv{X}.\,T \sep \rv{X}
\end{eqnarray*}
}
The set of participants of $G$, denoted \partp{G}, is defined as:
$\partp{\gto{p}{q}\{\lb{l}_i\langle U_i\rangle.G_i\}_{i \in I}} = \{\mathtt{p} , \mathtt{q}\} \cup \bigcup_{i \in I} \partp{G_i}$, \\ ~ $\partp{G_1 \para G_2} = \partp{G_1 } \cup \partp{G_2}$, ~$\partp{\mu\rv{X}.\,G} = \partp{G}$, ~$\partp{\gend} = \partp{\rv{X}} =\emptyset$.
We sometimes write $\pt{r} \in G$ to abbreviate $\pt{r} \in \partp{G}$.
\end{definition}
\noindent The global type 
$ \gto{p}{q}\{\lb{l}_i\langle U_i\rangle.G_i\}_{i \in I}$ specifies that, by choosing label $\lb{l}_i$,
 \pt{p} may send to 
 \pt{q} a message of type $U_i$, with subsequent behavior $G_i$.
 As in~\cite{DBLP:conf/popl/HondaYC08,DBLP:conf/icalp/DenielouY13}, we decree $\pt{p} \neq \pt{q}$, 
 so reflexive interactions are disallowed. Also, set $I$ is finite and labels are assumed pairwise different.
  The global type $ G_1 \para G_2$, introduced in~\cite{DBLP:conf/popl/HondaYC08}, 
  allows the concurrent execution of $G_1$ and $G_2$.
  Global type $\mu \rv{X}.G$ defines recurring conversation structures. 
As in~\cite{DBLP:conf/popl/HondaYC08,DBLP:conf/icalp/DenielouY13}, we 
restrict to global recursive types in which type variables $\rv{X}$ are all \emph{guarded}, i.e., 
they may only occur under branchings.
  Global type $\gend$ denotes the completed choreography. 
  At the local level, branching types $ \mathtt{p}?\{\lb{l}_i\langle U_i\rangle.T_i\}_{i \in I}$
  and selection types  $\mathtt{p}!\{\lb{l}_i\langle U_i\rangle.T_i\}_{i \in I}$ have expected readings.
  The terminated local type is denoted $\lend$.


We now define \emph{projection} and \emph{well-formedness} for global types.
We consider 
\emph{merge-based} projection as first proposed in 
\cite{DBLP:journals/corr/abs-1208-6483}
and used in \cite{DBLP:conf/icalp/DenielouY13}.
The definition below adds flexibility to the one in~\cite{DBLP:conf/popl/HondaYC08} by 
relying on a \emph{merge} operator on local types; we give a   
simpler presentation of the definition in~\cite[\S\,3]{DBLP:conf/icalp/DenielouY13},
considering messages $U$.


\begin{definition}[Merge]\label{d:mymerg}
We define $\fuse$ 
as the commutative partial operator 
on base and local types such that: 
\vspace{-1mm}
\begin{enumerate}[1.]
\item $\mathsf{bool} \fuse \mathsf{bool} =  \mathsf{bool}$ (and analogously for other base types)
\item $\lend \fuse \lend = \lend$
\item $\pt{p}!\{\lb{l}_i \langle U_i\rangle.T_i\}_{i \in I} \fuse \pt{p}!\{\lb{l}_i \langle U_i\rangle.T_i\}_{i \in I} = \pt{p}!\{\lb{l}_i \langle U_i\rangle.T_i\}_{i \in I}$
\item $\pt{p}?\{\lb{l}_k\langle U_k\rangle.T_k\}_{k \in K} \fuse \mathtt{p}?\{\lb{l}'_j\langle U'_j\rangle.T'_j\}_{j \in J} =$ \\ $ 
\pt{p}?(\{\lb{l}_k\langle U_k\rangle.T_k\}_{k \in K} \cup \{\lb{l}'_j\langle U'_j\rangle.T'_j\}_{j \in J})$ if 
for all $k \in K, j \in J$, $(\lb{l}_k = \lb{l}'_j)$ implies that both 
$U_k \fuse U'_j$ and $T_k \fuse T'_j$ are defined.
\item $\rv{X} \fuse \rv{X} = \rv{X}$ and $\mu \rv{X}.G_1 \fuse \mu \rv{X}.G_2 = \mu \rv{X}.(G_1 \fuse G_2)$.
\end{enumerate}
and is undefined otherwise.
\end{definition}

\noindent Intuitively, 
for $U_1 \fuse U_2$ to be defined there are two options:
(a) $U_1$ and $U_2$ are both identical base, terminated or selection types;
 (b) $U_1$ and $U_2$ are both branching types, but not necessarily identical:
they may offer different options but with the condition that the
behavior in labels occurring in both $U_1$ and $U_2$ must be the same, up to
$\fuse$.

\begin{definition}[Projection~\cite{DBLP:journals/corr/abs-1208-6483,DBLP:conf/icalp/DenielouY13}]\label{d:proj}
Let $G$ be a global type.
The \emph{(merge-based) projection} of $G$ under participant $\pt{r}$, 
denoted \proj{G}{\pt{r}}, is defined as follows:
\begin{enumerate}[$\bullet$]
\item $\proj{\gend}{\pt{r}}  =  \lend$

\item $\proj{\gto{p}{q}\{\lb{l}_i\langle U_i\rangle.G_i\}_{i \in I}}{\pt{r}}  = $ \\
$\begin{cases}
\pt{p}!\{\lb{l}_i \langle U_i\rangle.\proj{G_i}{\pt{r}}\}_{i \in I} & \text{if $\pt{r} = \pt{p}$} \\
\pt{p}?\{\lb{l}_i \langle U_i\rangle.\proj{G_i}{\pt{r}}\}_{i \in I} & \text{if $\pt{r} = \pt{q}$} \\
\fuse_{i \in I} \,\proj{G_i}{\pt{r}} & \text{otherwise, with $\fuse$ as in Def.~\ref{d:mymerg}} 
\end{cases} $

\item $\proj{(G_1 \para G_2) }{\pt{r}}  = 
\begin{cases} \proj{G_i}{\pt{r}} & \text{if $\pt{r} \in G_i$ and $\pt{r} \not\in G_j$, $i \neq j \in \{1,2\}$ } \\
														 \lend & \text{if  $\pt{r} \not\in G_1$ and $\pt{r} \not\in G_2$}
											\end{cases} $

\item $\proj{ \rv{X}}{\pt{r}}  =  \rv{X}$

\item $\proj{\mu \rv{X}.G}{\pt{r}}  = 
\begin{cases}
\mu \rv{X}.\proj{G}{\pt{r}}  & \text{if $ \proj{G}{\pt{r}}  \neq \rv{X}$} \\
\gend & \text{otherwise} 
\end{cases} $
\end{enumerate}
When a side condition does not hold, the map is undefined.
\end{definition}



\begin{definition}[Well-Formed Global Types~\cite{DBLP:conf/icalp/DenielouY13}]\label{d:wfltypes}
We say global type $G$ is \emph{well-formed  (WF, in the following)} if for all $\pt{r} \in G$, 
the projection $\proj{G}{\pt{r}}$ is defined.
\end{definition}


\newcommand{\spt}{\ensuremath{\bowtie}}

\subsection{Binary Session Types Based on Linear Logic}\label{ss:btypes}
In this paper we build upon the theory of binary session types of~\cite{CairesP10,Toninho2014},
based on a Curry-Howard interpretation of 
session types as linear logic propositions.
In the remainder, we assume no special background on linear logic; 
we refer to~\cite{CairesP10,Toninho2014} for further details.

\paragraph{The Process Model: Syntax and Semantics.}
We define a 
synchronous $\pi$-calculus~\cite{sangiorgi-walker:book} with replication, 
forwarding, $n$-ary labeled choice, and co-recursive definitions/variables (cf. \cite{Toninho2014}).
As for global types, we use $\lb{l}_1, \lb{l}_2, \ldots$ to range over \emph{labels}.

\begin{definition}[Processes]\label{d:procs}
Given an infinite set $\Names$ of {\em names} $(x,y,z,u,v)$, 
the set of {\em processes} $(P,Q,R)$ is  defined by
$$ 
\begin{array}{rllllllllllllllllllllllllll}
   P  ::= & \mathbf{0} &    \sep & P \para Q   &  \sep &  (\nub y)P \\
         \sep & \out{x}y.P  &  \sep & x(y).P   &  \sep &   \bang x(y).P &    \\
   \sep &  \mysel{x}{\lb{l}_i};P &  \sep &   \mycase{x}{\lb{l}_i : P_i}{i \in I}  & \sep  &\linkr{x}{y} \\
   \sep &  \rv{X}(\widetilde{c}) & \sep & (\recc \rv{X}(\widetilde{y}).P)\,\widetilde{c}
\end{array}
$$
\end{definition}
\noindent The operators $\zero$ (inaction), $P\para Q$ (parallel composition),
and $(\nub y)P$ (name restriction) are standard.
We then have $\out{x}y.P $ (send   $y$ on $x$ and
proceed as $P$), $\inp{x}{y}P $ (receive a   $z$ on $x$ and proceed
as $P$ with parameter $y$ replaced by $z$), and $ \bang
x(y).P$ which denotes replicated (persistent) input. 
The forwarding construct $\linkr{x}{y}$ equates names $x$ and $y$; 
it is a primitive representation of a copycat process.
The operators
$\mysel{x}{\,\lb{l}};P$ and $\mycase{x}{\lb{l}_i : P_i}{i \in I}$,
 define labeled choice  as in~\cite{DBLP:conf/esop/HondaVK98}.
 Given a sequence of names  
 $\widetilde{c}$, constructs 
 $(\recc \rv{X}(\widetilde{y}).P)\,\widetilde{c}$ and $\rv{X}(\widetilde{c})$ represent 
 co-recursive definitions and co-recursive variables, respectively.
We write $\outp{x}{y}$ to stand for $(\nub y)\out{x} y$.

In restriction $(\nub y)P$ and input $x(y).P$ the distinguished
occurrence of name $y$ is binding, with scope  $P$.
The set of
{\em free names} of a process $P$ is denoted $\fn{P}$. 
A process is {\em closed} if it does not
contain free occurrences of names.  We identify processes
up to consistent renaming of bound names, writing $\equiv_{\alpha}$
for this congruence.
We write
$P\subst{x}{y}$ for the 
capture-avoiding
substitution of  $x$ for  $y$ in $P$.
\emph{Structural congruence} ($P \equiv Q$) expresses basic identities on the structure of
processes. It is defined as 
  the least congruence relation on processes such that
  {
  $$
\begin{array}{c}
  P \para \zero  \equiv P \quad 
  P \para Q \equiv Q \para P  \quad 
  P \para (Q \para R) \equiv (P \para Q) \para R \\ 
    (\nub x)(\nub y)P \equiv (\nub y)(\nub x)P \quad
       (\nub x)\zero \equiv \zero \quad 
   P \equiv_{\alpha} Q \Rightarrow P \equiv Q \\
   \linkr{x}{y} \equiv \linkr{y}{x}  \quad 
     x\not\in\fn{P} \Rightarrow P \para (\nub x)Q \equiv (\nub x)(P \para Q) 
\end{array}$$
}
%
\emph{Reduction} 
specifies the computations a process performs on its own.
Closed  under structural congruence, 
  it is 
  the binary relation on processes defined by
  the rules in Fig.~\ref{fig:reduc}.
  
\begin{figure}[!t]
$$
\begin{array}{c}
\send{x}{y}.Q\para \inp{x}{z}{P} \red Q\para P\subst{y}{z} \\
 \send{x}{y}.Q\para \bang \inp{x}{z}{P} \red  Q\para P\subst{y}{z} \para \bang \inp{x}{z}{P} \\
(\nub x)(\linkr{x}{y} \para P) \red  P\subst{y}{x} \\
Q\red Q'\Rightarrow P \para Q \red P \para Q'  \\
P\red Q\Rightarrow(\nub y)P\red(\nub y)Q \\
P\equiv P'\aand P' \red  Q'\aand Q'\equiv Q\Rightarrow P\red Q \\
\mysel{x}{\lb{l}_j};P \para \mycase{x}{\lb{l}_i : Q_i}{i \in I}
\red  P \para Q_j ~~(j \in I)~ \\
(\recc \rv{X}(\widetilde{y}).P)\,\widetilde{c} \red P\subst{\widetilde{c}}{\widetilde{y}}\subst{\recc \rv{X}(\widetilde{y}).P\,}{\rv{X}}
\end{array}
$$
\caption{Process Reduction\label{fig:reduc}}
\vspace{-4mm}
\end{figure}

\paragraph{Session Types as Linear Logic Propositions.}
The theory of binary session types of~\cite{CairesP10} connects
session types as linear logic propositions. 
This correspondence has been extended in~\cite{Toninho2014} with coinductive session types.
Main properties derived from typing, important for our work and absent from other binary session type theories are \emph{global progress} and \emph{compositional non-divergence}.

%
%


\begin{definition}[Binary Types] \label{d:types}Types $(A,B,C)$ are given by 

\cbox{$
\begin{array}{rcl}
A,B & ::= & \one  \sep  ! A    \sep    A \otimes B 
  \sep  A \lol B   \sep  \rv{X} \sep   \nu \rv{X}.A \\
   &  \sep &   \mywith{\lb{l}_i : A_i}{i \in I}  \sep   \myoplus{\lb{l}_i : A_i}{i \in I}
\end{array}
$}
\end{definition}
Types are assigned to   names; assignment $x{:}A$ 
enforces the use of $x$ according to discipline $A$.
As hinted at above, we use $A\otimes B$ 
(resp. $A \lolli B$) 
to type  a name that performs an output (resp. an input) 
to its partner, sending (resp. receiving) a name of type $A$, 
and then behaves as
type $B$.  
We generalize the type syntax in~\cite{CairesP10} 
by considering $n$-ary  
offer $\with$ and choice $\oplus$.
Given a finite index set $I$, 
 $\mywith{\lb{l}_i{:}A_i}{i \in I}$ types a name that offers 
 a choice between an $l_i$. 
Dually, $\myoplus{\lb{l}_i{:}A_i}{i \in I}$ types the selection of one of the $\lb{l}_i$. 
Type $\bang A$ types a shared  channel, 
to be
used by a server for spawning an arbitrary number of new sessions
(possibly none), each one conforming to type $A$.
We use $\nu \rv{X}.A$ to type coinductive sessions, here required to type the medium processes of global types with recursion.
Coinductive types are required to have strictly positive occurrences of the type variable.
Also, coinductive types without session behavior before the occurrence of the type variable (e.g., $\nu \rv{X}.\rv{X}$) are excluded.
Finally, 
type
$\one$ means that the session terminated, no further interaction will
take place on it; names of type $\one$ may be passed
around in sessions, as opaque values.  

\begin{figure*}[t]
{\footnotesize
$$
\begin{array}{c}
\inferrule[\name{T$\mathsf{id}$}]{\mathstrut}{\Gamma; x{:}A \vdash_\eta \linkr{x}{z} :: z{:}A}
\qquad
\inferrule[\name{T$\lft\one$}]{\Gamma;\Delta \vdash_\eta P :: z{:}C}
{\Gamma;\Delta, x{:}\mathbf{1} \vdash_\eta P :: z{:}C}
\qquad
\inferrule[\name{T$\rgt\one$}]{\mathstrut}{\Gamma ; \cdot \vdash_\eta \zero :: x{:} \mathbf{1}}
\qquad
\inferrule[\name{T$\cut$}]{\Gamma; \Delta\vdash_\eta P :: x{:}A\;\;\; \Gamma; \Delta', x{:}A \vdash_\eta Q::z{:}C}
{\Gamma; \Delta, \Delta' \vdash_\eta (\nub x)(P\para Q):: z{:}C}
\vspace{0.2cm}
\\
\inferrule*[left=\name{T$\lft\lolli$}]{\Gamma; \Delta \vdash_\eta P :: y{:}A\;\;\;\; \Gamma; \Delta', x{:}B\vdash_\eta Q:: z{:}C}
{\Gamma; \Delta, \Delta', x{:} A\lol B \vdash_\eta  \outp{x}{y}. (P\para Q):: z{:}C}
\qquad
\inferrule*[left=\name{T$\lft\otimes$}]{\Gamma ; \Delta, y{:}A, x{:}B \vdash_\eta P :: z{:}C}
{\Gamma; \Delta, x{:} A\otimes B \vdash_\eta x(y).P :: z{:}C}
\vspace{0.2cm}
\\
\inferrule*[left=\name{T$\lft\oplus$}]{\Gamma; \Delta, x{:}A_1 \vdash_\eta P_1:: z{:}C \quad \cdots \quad \Gamma;  \Delta, x{:}A_k \vdash_\eta P_k :: z{:}C \quad I = \{1, \ldots, k\}}
{\Gamma; \Delta, x{:}\myoplus{\lb{l}_i : A_i}{i \in I} \vdash_\eta \mycase{x}{\lb{l}_i : P_i}{i \in I}:: z{:}C}
\vspace{0.2cm}
\\
\inferrule[\name{T$\rgt\oplus_1$}]{\Gamma; \Delta \vdash_\eta P:: x{:}A}
{\Gamma; \Delta \vdash_\eta \mysel{x}{\lb{l}_i};P::x{:}\myoplus{\lb{l}_i : A}{\{i\}}}
\quad
\inferrule[\name{T$\rgt\oplus_2$}]{\Gamma; \Delta \vdash_\eta P:: x{:}\myoplus{\lb{l}_i : A_i}{i \in I} \quad k \not\in I}
{\Gamma; \Delta \vdash_\eta P:: x{:}\myoplus{\lb{l}_j : A_j}{j \in I\cup\{k\}} }
\vspace{0.2cm}
\quad
\inferrule[\name{T$\lft\with_1$}]{\Gamma; \Delta, x{:}A \vdash_\eta P:: z{:}C}
{\Gamma; \Delta, x{:}\mywith{\lb{l}_i : A}{\{i\}}  \vdash_\eta \mysel{x}{\lb{l}_i};P:: z{:}C}
\\
\inferrule[\name{T$\rgt\with$}]{\Gamma; \Delta \vdash_\eta P_1:: x{:}A_1 ~~ \cdots ~~\Gamma;  \Delta \vdash_\eta P_k ::  x{:}A_k \quad I = \{1, \ldots, k\}}
{\Gamma; \Delta \vdash_\eta \mycase{x}{\lb{l}_i : P_i}{i \in I}:: x{:}\mywith{\lb{l}_i : A_i}{i \in I} }
\quad
\inferrule[\name{T$\lft\with_2$}]{\Gamma; \Delta, x{:}\with\{\lb{l}_i : A_i\}_{i \in I} \vdash_\eta P:: z{:}C \quad k \not\in I}
{\Gamma; \Delta, x{:}\mywith{\lb{l}_j : A_j}{j \in I\cup\{k\}} \vdash_\eta P:: z{:}C}
\vspace{0.2cm}
\\
\inferrule[$(\lft\nu)$]
{\G ; \D , c{:}A\subst{\nu \rv{X}.A}{\rv{X}} \vdash_\eta Q :: d{:}D}
{\G ; \D ,c{:}\nu \rv{X}.A \vdash_\eta Q :: d{:}D}
\qquad
\inferrule[$(\rgt\nu)$]
{\G ; \D \vdash_{\eta'} P :: c{:}A \qquad \eta' =
  \eta[\rv{X}(\widetilde{y}) \mapsto  \G ; \D \vdash_\eta c{:}\rv{Y}] }
{ \G ; \D \vdash_\eta 
  (\recc\;\rv{X}(\widetilde{y}).P\subst{\widetilde{y}}{\widetilde{z}})\;\widetilde{z}
  :: c{:}\nu \rv{Y}.A}
\qquad
\inferrule[\name{var}]
{ \eta(\rv{X}(\widetilde{y})) =  \G ; \D \vdash_\eta d{:}\rv{Y}
    \qquad \rho = \subst{\widetilde{z}}{\widetilde{y}}}
{ \rho(\G) ; \rho(\D) \vdash_\eta \rv{X}(\widetilde{z}) :: \rho(d){:}\rv{Y}}
\end{array}
$$
}
\caption{\label{fig:type-systemii}The Type System for Binary Sessions: Selected Rules.}
\vspace{-4mm}
\end{figure*}

A \emph{type environment} collects type assignments of the form
$x{:}A$, where $x$ is a name and $A$ a type, the names being pairwise disjoint.  
We consider two typing environments, subject to different structural properties: a
\emph{linear} part $\Delta$ and an \emph{unrestricted} part $\Gamma$,
where weakening and contraction principles hold for $\Gamma$ but not
for $\Delta$.  

A \emph{type judgment} 
is of the form $\Gamma ; \Delta \vdash_\eta P :: z{:}C$.
It asserts that 
$P$ provides behavior $C$ at
channel $z$, building on ``services'' declared in $\Gamma;\Delta$. 
Recall that $\eta$ denotes a map from (corecursive) type variables to typing contexts.
The domains of $\G,\D$ and $z{:}C$ are required to be pairwise disjoint.
As $\pi$-calculus terms are considered up to structural
congruence, typability is closed under $\equiv$ by definition. 
As a simple example of type judgment, 
a client $Q$ that relies on external
services and does not provide any is typed as $\Gamma ; \Delta \vdash Q::z{:}\one$.
Empty environments 
$\G, \D$ are   denoted by `$\,\cdot\,$'.
Also, we sometimes abbreviate $\G ; \D \vdash_\eta P::z{:}\one$ 
as $\G ; \D \vdash P$. 

Fig.~\ref{fig:type-systemii} presents selected rules of the type system; see~\cite{CairesP10,Toninho2014} for a full account.
Due to the logic correspondence,
we have 
right~($\mathsf{R}$) and left~($\mathsf{L}$) rules. 
The former 
detail how a process can implement the behavior described by the given connective; the latter 
explain how a process may use of a session of a given type.
Given these intuitions, 
the interpretation of the various rules should be clear. 
Rule~\name{T$\mathsf{id}$} defines identity in terms of forwarding.
Rule~(T$\cut$) define typed composition, restricting the scope of involved processes.
Based on rules~\name{T$\cut$} and \name{T$\lft\one$}, 
a rule 
for \emph{independent parallel composition}, enabling the  composition of 
$\Gamma; \Delta_1 \vdash P :: z {:}\one$ (with $z \not\in\fn{P}$, cf. rule \name{T$\rgt\one$}) and $\Gamma; \Delta_2 \vdash Q :: x{:}A$
into $\Gamma; \Delta_1,\Delta_2 \vdash P \para Q :: x{:}A$
is derivable. 
Implementing a session with type $\mywith{\lb{l}_i : A_i}{i \in I}$
amounts to offering a choice between $n$ sessions with type $A_i$ (cf. rule~\name{T$\rgt\with$}).
Using a session of type $\mywith{\lb{l}_i : A_i}{i \in I}$ 
on name $x$
entails selecting one of the alternatives, using a prefix $\mysel{x}{\,\lb{l}_j}$
(cf. rules~\name{T$\lft\with_1$} and \name{T$\lft\with_2$}). 
The  interpretation for the $n$-ary additive disjunction $\myoplus{\lb{l}_i : A_i}{i \in I}$ is  dual.

We now recall some main results for well-typed processes.
For any $P$, define
$live(P)$ iff $P \equiv (\nub \til{n})(\pi.Q \para R)$,
for some names $\til{n}$, 
a process $R$,  and a
\emph{non-replicated} guarded process $\pi.Q$.
Also, we write $P \Downarrow$, if there is no infinite reduction path from process $P$.


\begin{theorem}[Properties of Well-Typed Session Processes, \cite{Toninho2014}] Suppose $\G ; \D \vdash_\eta P :: z{:}A$.
\label{th:prop}
\begin{enumerate}[1.]
\item Type Preservation: If $P \red Q$ then $\G ; \D \vdash_\eta Q :: z{:}A$.
\item Progress: If $live(P)$ then there is $Q$ with $P \red Q$ or one of the following holds: 
\begin{enumerate}[(a)]
\item  $\D = \D',y{:}B$, for some $\D'$ and $y{:}B$. \\
There exists $\G ; \D''  \vdash_\eta R :: y{:}B$ s.t. 
$(\nub y)(R \mid P) \red Q$.
\item  Exists $\G ; z{:}A , \D' \vdash_\eta R :: w{:}C$ s.t. $(\nub z)(P \!\mid\! R) \tra{} Q$.
\item $\G = \G' , u{:}B$, for some $\G'$ and $u{:}B$. \\ 
There exists $\G ; \cdot \vdash_\eta R :: x{:}B$ s.t.  $(\nub u)(\bang u(x).R \mid P) \red Q$.
\end{enumerate}
\item Non-Divergence: $P \Downarrow$.
\end{enumerate}
\end{theorem}

In particular, Theorem~\ref{th:prop}(2), key in our developments, implies that our type discipline ensures freedom from deadlocks.

\section{Relating Multiparty Protocols and Binary Session Typed Processes}\label{s:tmedium}

\newcommand{\fgtypes}{\ensuremath{\mathcal{G}^{\text{fin}}}\xspace}
\newcommand{\rgtypes}{\ensuremath{\mathcal{G}^{\mu}}\xspace}

Our typeful characterization of multiparty conversations  as binary session types relies on the 
\emph{medium process} 
of a 
global type. Mediums provide a simple conceptual device for 
analyzing global types using the logically motivated binary session types of \cite{CairesP10,Toninho2014}.
In fact, as 
the medium
takes part in all message exchanges between local participants,  it uniformly and cleanly
captures the  sequencing behavior  stipulated by the global type. 

For  technical convenience, we divide the presentation of our main results into two representative sub-languages of global types:
\begin{enumerate}[$\bullet$]
\item We shall write \fgtypes to denote the class of global types generated by the syntax in Definition~\ref{d:gltypes}
\emph{without recursion}.
\item We shall write \rgtypes to denote the class of global types generated by the syntax in Definition~\ref{d:gltypes}
\emph{without the composition operator}.
\end{enumerate}
Global types in \fgtypes describe \emph{finite choreographies}. Isolating this class is useful to illustrate the simplicity of our approach; in particular, to illustrate the fact that it is fully orthogonal from infinite behaviors induced by recursion. Investigating  global types in \rgtypes is insightful: this is exactly the class of global types for which Deni{\'e}lou and Yoshida discovered the sound and complete characterization as communicating automata~\cite{DBLP:conf/icalp/DenielouY13}. In \S\,\ref{ss:parrec} we discuss further the tension between composition and recursion in process characterizations of global types.

\subsection{Medium Processes}\label{s:medium}
Based on the distinction between \fgtypes and \rgtypes,
we now introduce different definitions of medium processes. 
They all realize the simple concept motivated above and
 provide the basis for developing our technical results:
\begin{enumerate}[$\bullet$]
\item For $G \in \fgtypes$, we define \emph{finite mediums} $\ether{G}$ (Def.~\ref{d:ether}) and establish characterization results connecting well-typed mediums and well-formed global types (Theorems~\ref{l:ltypesmedp} and \ref{l:medltypes}). Based on finite mediums we also offer a behavioral characterization of \emph{swapping} in global types (Theorem~\ref{p:swapmeds}).

\item For $G \in \rgtypes$, we define \emph{recursive mediums} $\rether{G}{k}$ (Def.~\ref{d:rether}) to extend the characterization results to global types featuring infinite behavior (Theorem~\ref{l:ltypesmedprec} and~\ref{l:medltypesrec}). Then, 
we define \emph{annotated mediums} $\raether{G}{}{k}$ (Definition~\ref{d:raether}).
This notion, a slight variation of Def.~\ref{d:rether}, allows us to precisely relate actions of the global type and the observable behavior of its associated annotated medium (Theorem~\ref{th:opcorr}).
\end{enumerate}

\noindent We now proceed to define each of these different representations of the behavior of a global type.

\begin{definition}[Finite Mediums]\label{d:ether}
Let $G \in \fgtypes$ be a 
global type. 
The \emph{finite medium process} of $G$, denoted \ether{G}, is  defined inductively as follows:
\begin{enumerate}[$\bullet$]
\item $\ether{\gend}  =  \mathbf{0}$ 
\item $\ether{\gto{p}{q}\{\lb{l}_i\langle U_i\rangle.G_i\}_{i \in I}}   = $\\
$ 
\mycasebig{c_\pt{p}}{\lb{l}_i : c_\pt{p}(u).\mysel{c_\pt{q}}{\lb{l}_i};\outp{c_\pt{q}}{v}.( \linkr{u}{v} \para \ether{G_i} )}{i \in I} $
\item $\ether{G_1 \para G_2}  =  \ether{G_1} \para \ether{G_2}$
\end{enumerate}
\end{definition}

We now introduce recursive mediums, which represent recursive global types using co-recursive processes.
For technical reasons related to typability, we find it useful to ``signal'' when the global type ends and recurses. 
For simplicity, this signal is represented as a selection prefix (of a carried label) on a fresh name $k$:

\begin{definition}[Recursive Mediums]\label{d:rether}
Let $G \in \rgtypes$ be a
global type. 
Also, let 
$k$
be 
and a name assumed distinct from any other name. 
The \emph{recursive medium of $G$ 
with respect to label $\lb{l}$},
denoted \retheraux{G}{k}{\lb{l}}, 
is  defined  as follows:
\begin{enumerate}[$\bullet$]
\item $\retheraux{\gend}{k}{\lb{l}}  =  \mysel{k}{\lb{l}};\mathbf{0}$ 
\item $\retheraux{\gto{p}{q}\{\lb{l}_i\langle U_i\rangle.G_i\}_{i \in I}}{k}{\lb{l}}   = $ \\ $\mycasebig{c_\pt{p}}{\lb{l}_i : c_\pt{p}(u).\mysel{c_\pt{q}}{\,\lb{l}_i};\outp{c_\pt{q}}{v}.( \linkr{u}{v} \para \retheraux{G_i}{k}{\lb{l}_i} )}{i \in I} $
\item $\retheraux{\mu \rv{X}.G}{k}{\lb{l}}  =  (\recc \rv{X}(\til{z}).\retheraux{G}{k}{\lb{l}})\, \til{c}$ 
\item $\retheraux{\rv{X}}{k}{\lb{l}}   = \mysel{k}{\lb{l}};\rv{X}(\til{z})$
\end{enumerate}
Let $\lb{lb}$ be a label not in $G$.
The \emph{recursive medium of $G$}, denoted $\rether{G}{k}$,  is defined as $\retheraux{G}{k}{\lb{lb}}$.
\end{definition}

\begin{example}\label{ex:reccomm}
Finite mediums have already been  illustrated in \S\ref{s:example}. To illustrate recursive mediums, 
consider the following variant of the commit protocol in \S\ref{s:example}.
This is the running example in~\cite{DBLP:conf/icalp/DenielouY13}, here extended with base types:
$$
\begin{array}{rl}
G_{rc} = \mu \rv{X}.\,\gto{A}{B}\{ & \lb{~act}\langle \mathsf{int} \rangle. 
 \gto{B}{C}\{\lb{sig}\langle \mathsf{str}\rangle. \\
& \qquad \gto{A}{C}\{\lb{comm}\langle \one 
\rangle.\rv{X}\} \} \, , \\
 & \lb{quit}\langle \mathsf{int} \rangle.\gto{B}{C}\{\lb{save}\langle \one \rangle. \\
 & \qquad \gto{A}{C}\{\lb{fini}\langle \one \rangle.\gend\} \} ~~\} 
\end{array}
$$
Then, associating participants $\pt{A}$, $\pt{B}$, and $\pt{C}$ in $G_{rc}$ to names 
$a$, $b$, and $c$, respectively, process $\rether{G_{rc}}{k}$ is as in Figure~\ref{f:reccomm}.
\begin{figure*}[t]
$$
\begin{array}{ll}
\rether{G_{rc}}{k} = 
\recc\rv{X}(\til{z}).\,\big(\mycasebig{a} { 	& \lb{act} : a(v).\mysel{b}{\lb{act}};\outp{b}{w}.(\linkr{w}{v} \para 
		  \mycaseb{b}{\lb{sig} : b(n).\mysel{c}{\lb{sig}};\outp{c}{m}. \\
		& \qquad \qquad  \qquad (\linkr{n}{m} \para \mycaseb{a}{\lb{comm}: a(u).\mysel{c}{\lb{comm}};\outp{c}{y}.(\linkr{u}{y} \para \mysel{k}{\lb{comm}};\rv{X}(\til{z})\,) 
 		}{}  
		 )     }{} 
		 )~~,  \\ 
		& \lb{quit} : a(v).\mysel{b}{\lb{quit}};\outp{b}{w}.(\linkr{w}{v} \para 
		 \mycaseb{b}{\lb{save} : b(n).\mysel{c}{\lb{save}};\outp{c}{m}.\\
		& \qquad \qquad  \qquad (\linkr{n}{m} \para \mycaseb{a}{\lb{fini}:  a(u).\mysel{c}{\lb{fini}};\outp{c}{y}.(\linkr{u}{y} \para \mysel{k}{\lb{fini}};\zero \,) 
 		}{}  
		 )     }{} 
		 \, ) }{} \big)\,\til{c}\\ 
\end{array}
$$
\caption{Recursive medium process for the commit protocol, as in Example \ref{ex:reccomm}.\label{f:reccomm}}
\vspace{-4mm}
\end{figure*}
\end{example}

We now introduce the third class of mediums, \emph{annotated mediums}.
As in Definition~\ref{d:rether}, also in this case it is convenient to consider an additional fresh session $k$.
However, rather than signaling termination/recursion, in this case we use $k$ to 
emit an observable signal on $k$ for each action of $G$.
To this end, below we use the following notational conventions.
First, we write $k.P$ to stand for process $k(x).P$ whenever $x$ is not relevant in $P$.
Also, $\about{k}.P$ stands for the process $\outp{k}{v}.(\zero \para P)$
in which name $v$ 
is unimportant.

\begin{definition}[Annotated Mediums]\label{d:raether}
Let $G \in \rgtypes$ be a 
global type.  
Also, let $k$ be a fresh name.
The \emph{annotated medium} of $G$ with respect to $k$, 
denoted \raether{G}{\lb{l}}{k}, is  defined inductively as follows:
\begin{enumerate}[$\bullet$]
\item $\raether{\gend}{\lb{l}'}{k}  =  \mathbf{0}$ 

\item $\raether{\gto{p}{q}\{\lb{l}_i\langle U_i\rangle.G_i\}_{i \in I}}{\lb{l}}{k}    = $ \\
 $\qquad \quad \mycasebig{c_\pt{p}}{\lb{l}_i :~ \mysel{k}{\lb{l}_i}; c_\pt{p}(u).\about{k}.$\\
 $~~~~~~~~~~~\big(\mysel{c_\pt{q}}{\lb{l}_i};\mycase{k}{\lb{l}_i : \outp{c_\pt{q}}{v}.( \linkr{u}{v} \para k.\raether{G_i}{\lb{l}_i}{k} )}{\{i\}}\, \big)}{i \in I} $
\item $\raether{\mu \rv{X}.G}{\lb{l}}{k}  =  (\recc \rv{X}(\til{z}).\raether{G}{\lb{l}}{k})\, \til{c}$ 
\item $\raether{\rv{X}}{\lb{l}}{k}   = \rv{X}(\til{z})$
\end{enumerate}
\end{definition}

The key case is 
$\raether{\gto{p}{q}\{\lb{l}_i\langle U_i\rangle.G_i\}_{i \in I}}{\lb{l}}{k}$.
First, 
the selection of  label $\lb{l}_i$  by $\pt{p}$  is followed by a selection of $\mysel{k}{\,\lb{l}_i}$; 
then, the output from $\pt{p}$, 
captured by the medium by the input $c_\pt{p}(u)$, 
 is followed by an output on $k$; 
 subsequently, 
 the selection on $c_\pt{q}$ of label $\lb{l}_i$ is followed 
 by a branching in $k$ on label $\lb{l}_i$; 
 finally, the output $\outp{c_\pt{q}}{v}$ is signaled by an input on $k$, which prefixes the execution of the continuation $\raether{G_i}{\lb{l}_i}{k}$. This way, actions on $k$ induce a fine-grained correspondence with the behavior of  $G$. 

We assume the following name convention for mediums:
the actions of every participant $\pt{p}$ in $G$ 
are described in $\ether{G}{}$ 
by prefixes on name $c_{\pt{p}}$ (similarly for $\rether{G}{k}$ and $\raether{G}{}{k}$).
Since in labeled communications
$G = \gto{p}{q}\{\lb{l}_i\langle U_i\rangle.G_i\}_{i \in I}$
we always assume $\pt{p} \neq \pt{q}$, 
in \ether{G}{}  we will have $c_{\pt{p}} \neq c_{\pt{q}}$
(similarly for $\rether{G}{k}$ and $\raether{G}{k}{k}$).
Due to this convention we have:

\begin{myfact}\label{fc:name}
Let $G$ and \ether{G}{} be a global type and its medium, resp.
For each 
$\mathtt{r}_j \in \partp{G}$ there is a name $c_j \in \fn{\ether{G}{}}$.
(And analogously for $\rether{G}{k}$ and $\raether{G}{}{k}$.)
\end{myfact}

We stress that from the standpoint of the protocol participants, the existence of the medium is inessential:
the local implementations may be constructed (and type-checked) exactly as prescribed by their projected local types, unaware of the medium and its internal structure. 
In contrast, the medium 
\emph{depends} on well-behaved participants as stipulated by the global type. 
In the following we will formalize these intuitions,
using the theory of binary session types described in \S\ref{ss:btypes}.
We will then be able to formally define the dependence of well-typed mediums on  local participants which are 
well-typed with respect to projections of the given global type.

We first present characterization results for  global types in $\fgtypes$~(\S\,\ref{ss:frela}).
We extend these results to global types in \rgtypes (\S\,\ref{ss:rela}). 

\subsection{Relating Well-Formed Global Types and Typed Mediums: The Finite Case}\label{ss:frela}
We formally relate a global type $G \in \fgtypes$, its associated medium $\ether{G}$, 
and its corresponding  local types $\proj{G}{\pt{p}_1}, \ldots, \proj{G}{\pt{p}_n}$.
We first introduce some useful auxiliary notions. 
\emph{Compositional typings} are a class of type judgments which in line with 
the name convention for mediums (cf. Fact~\ref{fc:name}).
Below, we sometimes write $\G; \D \vdash_\eta \ether{G}$ instead of $\G; \D \vdash_\eta \ether{G} :: z {:}\one$, when $z \not\in\fn{\ether{G}}$.



\begin{definition}[Compositional Typing]\label{d:compty}
Let $G$ be a global type. 
We say that judgement 
$\G; \D \vdash \ether{G} ::z{:}C$
is a \emph{compositional typing} 
for $\ether{G}$
if:  (i)~it is a valid typing derivation;
(ii)~$\D = c_1{:}A_1, \ldots, c_n{:}A_n$;
(iii)~for all 
$\pt{r}_i \in G$
 there is a $c_i{:}A_i \in \D$;  
 (iv)~$C = \one$. 
 In case only conditions (i)--(iii) hold, we say that the judgment is a \emph{left-compositional typing} for $\ether{G}$.
\end{definition}

Intuitively, compositional typings 
formalize the intuitions hinted at the end of \S\,\ref{s:medium}.
These typed interfaces formalize 
the fact that the medium does not offer
any behaviors of its own (cf. the right-hand side  $z{:}\one$) while depending
on behaviors which should be available on its free names (cf. the condition on left-hand side typing $\D$).

The main difference between local types and binary session types is that
the latter do not mention participants.
Below, we use $B$ to range over base types ($\mathsf{bool}, \mathsf{nat}, \ldots$) in 
Definition~\ref{d:gltypes}.
\begin{definition}[Local Types and Binary Types]\label{d:loclogt}
Mapping \lt{\cdot} from local types $T$ (cf.~Def.~\ref{d:gltypes}) into binary types $A$ (cf. Def.~\ref{d:types}) is 
inductively defined as:
$$
\begin{array}{rcl}
\lt{\lend} = \lt{B} &=&  \one  \\
\lt{\mathtt{p}!\{\lb{l}_i\langle U_i\rangle.T_i\}_{i \in I}}  & = &  \myoplus{\lb{l}_i : \lt{U_i} \otimes \lt{T_i}}{i \in I} \\
 \lt{\mathtt{p}?\{\lb{l}_i\langle U_i\rangle.T_i\}_{i \in I}} & = &   \mywith{\lb{l}_i : \lt{U_i} \lolli \, \lt{T_i}}{i \in I} 
\end{array}
$$
\end{definition}

Given a global type $G$,
we now 
formally relate 
process $\ether{G}$ (typed with a compositional typing) and binary session types representing 
the projections of $G$.


\subsubsection{Characterization Results}
We now present the key correspondence results between global types and well-typed finite mediums
(Theorems~\ref{l:ltypesmedp} and \ref{l:medltypes}).
The first direction of the characterization says that
 well-formedness of global types (Def.~\ref{d:wfltypes})
suffice to ensure compositional typings for mediums with (linear logic based) binary session types:

\begin{theorem}[From Well-Formedness To Typed Mediums]\label{l:ltypesmedp}
Let $G \in \fgtypes$  be a global type, with~$\partp{G} = \{\pt{p}_1, \ldots, \pt{p}_n\}$.
If $G$ is 
WF
(Def.~\ref{d:wfltypes}) then $$\G; c_1{:}\lt{\proj{G}{\pt{p}_1}}, \ldots, c_n{:}\lt{\proj{G}{\pt{p}_n}} \vdash \ether{G}$$ is a compositional typing
for \ether{G}, for some $\G$.
\end{theorem}

\begin{proof}
By a structural induction on $G$; see Appendix~\ref{app:ltypesmedp}.
In case 
$$G = \proj{\gto{p}{q}\{\lb{l}_i\langle U_i\rangle.G_i\}_{i \in I}}{\pt{r}}$$ with $\{\pt{r}\}\#\{\pt{p}, \pt{q}\}$, 
the flexibility given by $\fuse$ (Def.~\ref{d:mymerg})
results into $\with$ types in the left-hand side typing for \ether{G} which may not be identical.
To derive the desired compositional typing,
we 
use 
rule~\name{T$\lft\with_2$} so as to silently add/remove labeled options in left $\with$ types
until achieving identical typings (as required to use rule~\name{T$\lft\oplus$}).
The case $G = G_1 \para G_2$ uses independent parallel composition.
\end{proof}

\noindent The following theorem
states the converse of Theorem~\ref{l:ltypesmedp}:
it says that compositional typings for mediums induce global types which are WF.
We require the following auxiliary definition, which relies on the merge operator given in Definition~\ref{d:mymerg}.
\begin{definition}\label{d:subt}
Given local types $T_1, T_2$, we write 
$T_1 \subt T_2$ if there exists a local type $T'$ such that $T_1 \fuse T' = T_2$.
\end{definition}


\begin{theorem}[From Well-Typedness To WF Global Types]\label{l:medltypes}
Let $G \in \fgtypes$ be a 
global type. 
If 
$\G; c_1{:}A_1, \ldots, c_ n{:}A_ n \vdash \ether{G}$
is a compositional typing for \ether{G}
then $\exists 
T_1, \ldots, T_ n$ 
s.t. 
$\proj{G}{\mathtt{r}_j} \subt T_j$ and 
$\lt{T_j} = A_j$, 
for all $\pt{r}_j \in G$.
\end{theorem}


\noindent The proof of 
Theorem~\ref{l:medltypes}
is by structural induction on $G$; 
see Appendix~\ref{app:medltypes}.
Observe how notation $\subt$ allows us to handle the occurrence of labeled alternatives which may be silently introduced by rule~\name{T$\lft\with_2$}.

%

Remarkably, our results tightly and formally connect global types (in $\fgtypes$), local types, and projection (on the multiparty approach)
and medium processes and  deadlock-free binary session types (rooted in linear logic).
Our results provide 
an independent deep justification, through purely logical arguments, to 
 the forms of projection proposed in the literature.
We do not know of works in which the semantics of global type projection
is compared/assessed based on different foundations; this also seems an interesting contribution of our  approach to multiparty protocol analysis.

\begin{remark}\label{rem:swf}
Theorems~\ref{l:ltypesmedp} and~\ref{l:medltypes} concern well-formed global types as in~\cite{DBLP:journals/corr/abs-1208-6483,DBLP:conf/icalp/DenielouY13}.
The theorems hold also when global types are well-formed as in~\cite{DBLP:conf/popl/HondaYC08}; we call those types \emph{simply well-formed (or SWF)}. In the analog of Theorem~\ref{l:ltypesmedp} for SWF global types, the proof is simpler as 
projectibility in ~\cite{DBLP:conf/popl/HondaYC08}  ensures identical behavior in all branches.
See Appendix~\ref{app:sec:swf}.
\end{remark}

\subsubsection{A Behavioral Characterization of Global Swapping}
The \emph{swapping relation} over global types was proposed in~\cite{DBLP:conf/popl/CarboneM13}
as a way of enabling behavior-preserving transformations among causally independent communications.
Such transformations may represent optimizations, 
in which parallelism is increased while preserving the overall intended semantics. 
We now show a characterization of swapping on global types in terms of a typed behavioral equivalence on mediums.

\begin{definition}[Swapping for Global Types]\label{d:gswap}
We define \emph{swapping}, denoted \eqsw, as the smallest congruence on global types which satisfies the rules 
in Fig.~\ref{fig:swap}.
\end{definition}
\begin{figure}[t!]
{\small
 $$
\begin{array}{c}
\inferrule*[left=\name{sw1}]{\{\mathtt{p_1},\mathtt{q_1}\} \# \{\mathtt{p_2},\mathtt{q_2}\}}
{
\begin{array}{c}
\gto{p_1}{q_1}\big\{\lb{l}_i\langle U_i\rangle.\gto{p_2}{q_2}\{\lb{l}'_j\langle U'_j\rangle.G_{ij}\}_{j \in J}\big\}_{i \in I} \vspace{-0.05cm} \\
\eqsw \qquad \vspace{-0.00cm}\\
\gto{p_2}{q_2}\big\{\lb{l}'_j\langle U'_j\rangle.\gto{p_1}{q_1}\{\lb{l}_i\langle U_i\rangle.G_{ij}\}_{i \in I}\big\}_{j \in J}
\end{array}}
\vspace{0.2cm}
\\
\inferrule[\name{sw2}]{\{\mathtt{p},\mathtt{q}\} \# \partp{G_1} \qquad \forall i, j \in I. G^1_i = G^1_j}
{
\begin{array}{c}
\gto{p}{q}\{\lb{l}_i\langle U_i\rangle.(G^1_i \para G^2_i)\}_{i \in I} 
\eqsw 
G^1_1 \para \gto{p}{q}\big\{\lb{l}_i\langle U_i\rangle.G^2_i\}_{i \in I}
\end{array}
}
\end{array}
$$
}
\caption{Swapping  on global types (cf. Def.~\ref{d:gswap}).  $A\#B$ denotes that sets $A, B$ are disjoint. The symmetric of~\name{\textsc{sw2}} is omitted. \label{fig:swap}}
\vspace{-3mm}
\end{figure}

To characterize swapping, we briefly discuss \emph{proof conversions}
and
\emph{typed behavioral equivalences} for logic-based binary session types, following~\cite{DBLP:conf/esop/PerezCPT12}.
The correspondence in~\cite{CairesP10,Toninho2014}
is 
realized by relating \emph{proof conversions} in linear logic with appropriate behavioral 
equivalences in the process setting.
Most conversions correspond to either reductions or structural congruences at the level of processes.
There is  a group of 
\emph{commuting conversions} which actually induce a behavioral congruence on typed processes, denoted $\eqcc$.
Process equalities justified by $\eqcc$ include, e.g., the following ones:
{
$$
\begin{array}{c}
(\nub x)(P \para y(z).Q)  \typconp   y(z).(\nub x) (P \para Q)  \\
(\nub x)(P \para \outp{y}{z}.(Q \para R))   \typconp   \outp{y}{z}.(Q \para (\nub x)(P \para R)) \\
(\nub x)(P\para \mysel{y}{\lb{l}_i}; Q)  \typconp  \mysel{y}{\lb{l}_i};(\nub x)(P \para  Q)  
\end{array}
$$
}
Processes 
equated by $\eqcc$
are syntactically very different and yet they are associated to the session typed (contextual) behaviour.
These equalities reflect a natural 
typed behavioral equivalence over session-typed processes, called \emph{typed context bisimilarity}~\cite{DBLP:conf/esop/PerezCPT12}.
 Roughly, 
 typed processes $\G; \D \vdash P ::x{:}A$ and $\G; \D \vdash Q ::x{:}A$ are 
 typed context bisimilar, denoted \trelind{\Gamma; \Delta}{x{:}A}{P}{Q} if, once composed with their requirements (as described by $\G$ and $\D$), they perform the same actions on $x$ (following $A$).
 Typed context bisimilarity is a congruence on well-typed processes.

\begin{theorem}[\cite{DBLP:conf/esop/PerezCPT12}]\label{t:soundcc}
Let $P, Q$ be well-typed processes. \\
If  
$\Gamma; \Delta \vdash P \eqcc Q :: z{:}C$
then \trelind{\Gamma; \Delta}{z{:}C}{P}{Q}.
\end{theorem}

It turns out that swapping in global types (Def.~\ref{d:gswap}) can also be directly justified
from crisper, more primitive notions, based on the correspondence established by Theorems~\ref{l:ltypesmedp} and~\ref{l:medltypes}.
Indeed, by formalizing the behavior of a global type in terms of its medium 
we may reduce transformations at the level of global types to
sound transformations at the level of  processes. 

Theorem~\ref{p:swapmeds} below 
gives a strong connection between swapping on global types ($\eqsw$) 
with typed context bisimilarity (\tybis), as 
motivated above (and defined by P\'{e}rez et al. in~\cite{DBLP:conf/esop/PerezCPT12}).
Thanks to the theorem, the sequentiality of mediums can be relaxed in the case of causally independent communications formalized by swapping.

\begin{theorem}\label{p:swapmeds}
Let $G_1 \in \fgtypes$ be a global type, such that $\ether{G_1}$ has 
a compositional typing $\G; \D \vdash \ether{G_1} $, for some $\G, \D$. \\
If $G_1 \eqsw G_2$ then  $\G;  \D \vdash \ether{G_1} \tybis \ether{G_2}  $.
\end{theorem}

\begin{proof}[Proof (Sketch)]
The proof proceeds by induction on the definition of $\eqsw$ (Def.~\ref{d:gswap}). 
To relate swapping with typed context bisimilarity we rely on 
$\eqcc$.
We first show that
\begin{equation}
\text{If $G_1 \eqsw G_2$ then $\G; \D \vdash \ether{G_1} \eqcc \ether{G_2}$}  \label{pr:swpa}
\end{equation}
To establish \eqref{pr:swpa}, we exploit 
the relation between participant identities in global types and names in associated mediums (Fact~\ref{fc:name}): this allows to infer that
disjointness conditions for swapping rules imply name distinctions,  which in turn enables type-preserving transformations via $\eqcc$.
The needed   transformations rely on equalities 
detailed by P\'{e}rez et al. in~\cite{DBLP:conf/esop/PerezCPT12}.
The thesis follows by combining \eqref{pr:swpa} with Theorem~\ref{t:soundcc}. 
See Appendix~\ref{app:swapmeds} for details.
\end{proof}

\noindent The converse of Theorem~\ref{p:swapmeds} does not hold in general: 
given $\ether{G}$, 
the existence of a   $P'$ such that $\ether{G} \eqcc P'$ does not necessarily imply 
the existence of a $G'$ such that $G \eqsw G'$ and $P' = \ether{G'}$.
For instance, consider the global type 
$$G_1 = \gto{p}{q}\big\{\lb{l}_i\langle U_i\rangle.\gto{r}{p}\{\lb{l}'_j\langle U'_j\rangle.G_{ij}\}_{j \in J}\big\}_{i \in I}$$
It cannot be swapped and yet 
prefixes for \pt{q} and \pt{r} 
in $\ether{G_1}$ could be commuted.
In general, mediums are a fine-grained representation of global types:
as a single communication in $G$ is implemented in $\ether{G}$ using several prefixes,
swapping of a type $G$ occurs only  when \emph{all} involved prefixes in $\ether{G}$ can be commuted  via \eqcc.
We stress that commutations induced by \eqcc are always type-preserving; hence,
typing for $\ether{G}$ is invariant under swapping.

\subsection{Results for Well-Formed Global Types With Recursion}\label{ss:rela}
In this sub-section we 
consider the language of global types
with recursion and without parallel and 
extend the characterization results in \S\,\ref{ss:frela}. We also present an operational correspondence result.

\subsubsection{Characterization Results for Recursive Mediums} 
We require the following (expected) extension to 
mapping $\lt{\cdot}$, given in
Definition~\ref{d:loclogt}:
\begin{eqnarray*}
\lt{\rv{X}} & = & \rv{X} \\
\lt{\mu\rv{X}.T} & = & \nu\rv{X}.\lt{T}
\end{eqnarray*}

To state the analogous  of Theorems~\ref{l:ltypesmedp} and~\ref{l:medltypes}
for co-recursive mediums, we need to close the local projections of a global type.
The following definition defines a closure for such recursion variables, using mapping $\eta$.
Given $ \eta =  \eta'[\rv{X}(\widetilde{y}) \mapsto  \G ; \D \vdash k{:}\rv{Y}]$
with $c_i{:}A_i \in \D$
we write $\eta(\rv{X})(c_i)$ to denote the type $A_i$.
We write $\fv{G}$ to denote the set of free recursion variables in $G$.

\begin{definition}[Closure for Local Types]
Let $G$, $\mathcal{P}$, and $\eta$  be a global type, a set of participants, and a mapping from process variables to typing contexts, respectively.
Also, let \lt{\cdot} be the map of Def.~\ref{d:loclogt}, extended as above.
We define:
\begin{itemize}
\item $\clt{\proj{G}{\pt{p}_i}}{\eta}{\mathcal{P}} =$ \\
$
\begin{cases}
\lt{\proj{G}{\pt{p}_i}}\subst{\eta(\rv{X})(c_i)}{\rv{X}} & \text{if $\fv{\proj{G}{\pt{p}_i}} = \{\rv{X}\}$ and $\pt{p}_i \in \mathcal{P}$} \\
\eta(\rv{X})(c_i) & \text{\text{if $\fv{\proj{G}{\pt{p}_i}} = \{\rv{X}\}$ and $\pt{p}_i \not\in \mathcal{P}$}} \\
\lt{\proj{G}{\pt{p}_i}} & \text{if $\fv{\proj{G}{\pt{p}_i}} = \emptyset$}
\end{cases}
$
\end{itemize}
\end{definition}

Concerning typing for mediums, 
the main consequence of adding recursion is that we no longer have $\one$ at the right-hand side typing (cf. Def.~\ref{d:compty}).
Intuitively, this is because we can never fully consume a recursive behavior, which is essentially infinite. 
If recursion is required in the left-hand side typing  then
some recursive behavior must show up in the right-hand side along name $k$.
(Notice that by Def.~\ref{d:proj}, all the local projections of a recursive global type will be also recursive.)

We now extend 
Theorems~\ref{l:ltypesmedp} and~\ref{l:medltypes}
to  global types in \rgtypes. 
As before, the first direction of the characterization says that
the conditions that merge-based well-formedness induces on global types
suffice to ensure compositional typings for mediums (cf. Definition~\ref{d:compty}):

\begin{theorem}[From Well-Formedness To Typed Mediums]\label{l:ltypesmedprec}
Let $G \in \rgtypes$ be a global type with~$\partp{G} = \mathcal{P} = \{\pt{p}_1, \ldots, \pt{p}_n\}$.
If $G$ is WF
(Def.~\ref{d:wfltypes}) then 
 $$\G; c_1{:}\clt{\proj{G}{\pt{p}_1}}{\eta}{\mathcal{P}}, \ldots, c_n{:}\clt{\proj{G}{\pt{p}_n}}{\eta}{\mathcal{P}} \vdash_\eta \rether{G}{k} :: k{:} A$$
 is a left compositional typing
for \rether{G}{k} 
for some $\G, \eta, A$.
\end{theorem}


We now state the converse of Theorem~\ref{l:ltypesmedprec}.
It says that well-typed mediums induce global types which are well-formed.
That is,  the sequential structure of mediums can be precisely captured by binary session types
which have a corresponding local type.

\begin{theorem}[From Well-Typedness To WF Global Types]\label{l:medltypesrec}
Let $G \in \rgtypes$ be a 
global type, 
with~$\partp{G} = \mathcal{P} = \{\pt{p}_1, \ldots, \pt{p}_n\}$. 
If 
$$\G; c_1{:}A_1, \ldots, c_ n{:}A_ n \vdash_\eta \rether{G}{k} ::k{:}B$$
is a 
left compositional typing for \rether{G}{k}
then $\exists 
T_1, \ldots, T_ n$ 
s.t. 
$\proj{G}{\mathtt{p}_j} \subt T_j$ and 
$\clt{T_j}{\eta}{\mathcal{P}} = A_j$
for all $\pt{p}_j \in G$.
\end{theorem}


\subsubsection{Operational Correspondence via Annotated Mediums}\label{s:opcorr}
The results already presented focus on the static semantics of multiparty and binary systems, and are
already key to justify essential properties such as preservation of global deadlock.
We now 
move on to dynamic semantics, and
establish the expected precise operational correspondence result between a global type  and its medium process (Theorem~\ref{th:opcorr}).
To this end, we rely on the annotated mediums of Definition~\ref{d:raether}: 
given a global type $G \in \rgtypes$, 
its annotated medium $\raether{G}{}{k}$ 
contains  
an independent session $k$ which signals the 
behavior of $G$.
In typing, the observable behavior on $k$ 
will appear on the right-hand side typing.
 

The following definition relates the behavior of a global type
and that of $k$.

\begin{definition}[Global Types and Binary Session Types]\label{d:glolog}
Let $\rpart{\cdot}$ denote a mapping from participants to binary session types.
The mapping \gt{\cdot} from a global types $G \in \rgtypes$ 
into binary session types $A$ (cf. Def.~\ref{d:types}) is 
inductively 
defined as:
\begin{enumerate}[$\bullet$]
\item $\gt{\gend}   =    \one$
\item $\gt{\gto{p}{q}\{\lb{l}_i\langle U_i\rangle.G_i\}_{i \in I}}   =$ \\
$~~~~~~~~~~~~~~~ \myoplus{\lb{l}_i : \rpart{\pt{p}} \otimes \mywith{\lb{l}_i :\rpart{\pt{q}} \lolli \gt{G_i}}{\{i\}}}{i \in I}$ 
\item $\gt{\rv{X}}  =  \rv{X}$ 
\item $\gt{\mu\rv{X}.G}  =  \nu\rv{X}.\gt{G}$
\end{enumerate}
\end{definition}

For simplicity, we shall assume $\rpart{\pt{p}} = \one$, for every $\pt{p}$.
We may recast Theorem~\ref{l:ltypesmedprec} above for annotated mediums as follows.

\begin{theorem}[From Well-Formedness To Typed Annotated Mediums]\label{l:ltypesmedpan}
Let $G \in \rgtypes$ be a 
global type with~$\partp{G} = \mathcal{P} = \{\pt{p}_1, \ldots, \pt{p}_n\}$.
If $G$ is WF
(Def.~\ref{d:wfltypes}) then 
judgment 
$$\G; c_1{:}\clt{\proj{G}{\pt{p}_1}}{\eta}{\mathcal{P}}, \ldots, c_n{:}\clt{\proj{G}{\pt{p}_n}}{\eta}{\mathcal{P}} \vdash_\eta \raether{G}{k}{k} :: k{:} \gt{G}$$
is well-typed,  for some $\G$.
\end{theorem}

The proof of Theorem~\ref{l:ltypesmedpan} extends the one for Theorem~\ref{l:ltypesmedprec} by considering
session $k$, which is causally independent from all other sessions of the medium. 
An analogous of 
Theorem~\ref{l:medltypesrec} holds for annotated mediums.
Because of the silent character of rule~\name{T$\rgt\oplus_2$}, we require some additional notation.
Below we write $A_1 \subts A_2$ iff 
either $A_1 = A_2$ or 
$A_1 = \myoplus{\lb{l}_i:A_i}{i \in I}$
and $A_2 = \myoplus{\lb{l}_j:A_j}{j \in I \cup J}$, for some $J$.

\begin{theorem}[From Well-Typed Annotated Mediums To WF Global Types]\label{l:medltypespan}
Let $G \in \rgtypes$ be a 
global type. 
If the following judgment is well-typed
$$\G; c_1{:}A_1, \ldots, c_ n{:}A_ n \vdash \raether{G}{k}{k} :: k:A_0$$
then 
$\gt{G} \subts A_0$ and 
$\exists T_1, \ldots, T_ n$ 
s.t. 
$\proj{G}{\mathtt{r}_j} \subt T_j$ and 
$\clt{T_j}{\eta}{\mathcal{P}} = A_j$
for all $\pt{r}_j \in G$.
\end{theorem}

\noindent The operational correspondence 
between global types and annotated mediums
is given by Theorem~\ref{th:opcorr} below.

To state operational correspondence, we consider the set of \emph{multiparty systems} of a global type.
Intuitively, 
given a global type $G$,
a particular multiparty system is 
obtained by the composition of well-typed implementations of the local behaviors stipulated by $G$ (with no linear/shared dependencies)
with the annotated medium \raether{G}{k}{k}, which provides the ``glue code'' for connecting them all.
We write $\mathcal{S}^k(G)$ to denote the set of all multiparty systems of $G$; hence, by construction, any $P \in \mathcal{S}^k(G)$ is a particular implementation of the multiparty conversations specified by $G$.
More formally, we require the following auxiliary definition.
  
  \begin{definition}[Closure]\label{d:elrauxc}
Let 
$\Delta = \{x_{j}{:}A_{j}\}_{j \in J}$
be 
a 
linear 
typing environment. 
We define the set of processes \closetc{\Delta} as:
$$
 \closetc{\Delta}  \stackrel{\textrm{def}}{=} \Big\{\prod_{j \in J} Q_{j} ~~\sep~ \cdot ; \cdot \vdash Q_{j} :: x_{j}{:}A_{j}\Big\}
$$
\end{definition}

Using closures, we can now define systems:

\begin{definition}[System]\label{d:systems}
Let $G \in \rgtypes$ be a WF global type, such that $\partp{G} = \mathcal{P} = \{\pt{p}_1, \ldots, \pt{p}_n\}$.
Also, let 
$$\D = c_1{:}\clt{\proj{G}{\pt{p}_1}}{\eta}{\mathcal{P}}, \ldots, c_n{:}\clt{\proj{G}{\pt{p}_n}}{\eta}{\mathcal{P}}$$
be an environment such that 
$\G;  \D \vdash \raether{G}{k}{k} :: k{:} \gt{G}$, for some  $\G$.
The set of multiparty systems of $G$, written $\mathcal{S}^k(G)$, is  defined as:
$$\big\{(\nub c_{\pt{p}_1}, \ldots, c_{\pt{p}_n})(Q \para \raether{G}{k}{k})  ~~\sep~ Q \in \closetc{\Delta }\big\}
$$
\end{definition}
By construction,
processes in $\mathcal{S}^k(G)$ will only have observable behavior on name $k$.
We rely on labeled transition systems (LTSs) for processes and global types;  they are
denoted $P \tra{\,\labelset\,} P'$ and $G \tra{\,\gtlabelset\,} G'$, respectively.
While the former  is standard for session $\pi$-calculi (see, e.g.,~\cite{CairesP10}), the latter   results by extending the LTS in~\cite{DBLP:conf/icalp/DenielouY13} with intermediate states.
Given a name $k$ and a participant   $\pt{p}$, 
two auxiliary mappings, denoted  $\mlab{\cdot}{k}$ and $\glab{\cdot}{\pt{p}}$,
tightly relate global type labels $\gtlabelset$ to process labels $\labelset$.
(See Appendix \ref{app:lts} for details in the LTS for global types and the auxiliary mappings.)
We  have:

\begin{theorem}[Global Types and  Mediums: Operational Correspondence]\label{th:opcorr}
Let $G \in \rgtypes$ be a WF global type and $P$ any process in $\mathcal{S}^k(G)$. We have:
\begin{enumerate}[(a)]
\item If $G \tra{\,\gtlabelset\,} G'$ then 
there exist 
$\labelset, P'$ s.t. $P \wtra{\,\labelset\,} P'$,
$\labelset = \mlab{\gtlabelset}{k}$, 
 and $P' \in \mathcal{S}^k(G')$.
 \item If there is some $P_0$ s.t. $P \wtra{} P_0 \tra{\,\labelset\,} P'$ with
 $\labelset \neq \tau$ then there exist $\gtlabelset,  G'$ s.t. $G \tra{\,\gtlabelset\,} G'$, 
 $subj(\gtlabelset) = \pt{p}$,
 $ \gtlabelset = \glab{\labelset}{\pt{p}} $, 
 and $P' \in \mathcal{S}^k(G')$.
\end{enumerate}
\end{theorem}

\noindent 
By means of this 
 operational correspondence between 
a global type $G$ and its process implementations
(as captured by set 
$\mathcal{S}^k(G)$), we confirm that (annotated) medium processes faithfully mirror the communicating behavior of the given global type.

\subsection{The Tension Between Global Types for Composition and Recursion\label{ss:parrec}}
Distinguishing between \fgtypes and \rgtypes has a conceptual justification, as discussed in \S\,\ref{s:medium}.
There is also a more technical motivation for this distinction, related to typability. The reason why finite mediums (Def.~\ref{d:ether})
can support the composition of global types (as originally proposed in~\cite{DBLP:conf/popl/HondaYC08}) is the following: as finite mediums do not have behavior on their own (i.e., compositional typings ensure that their right-hand side typing is  $\one$) they are amenable to independent parallel composition, as supported by binary session types. This kind of parallel composition neatly coincides with well-formedness for composition of global types in~\cite{DBLP:conf/popl/HondaYC08}, which allow composition the global types with disjoint senders and receivers (cf. Def.~\ref{d:proj}). We find it remarkable that finite mediums are able to cleanly justify natural requirements for multiparty session types. 

Independent parallel composition is no longer possible when we move to co-inductive types, which are needed to type the mediums for global types in $\rgtypes$ (cf. Defs.~\ref{d:rether} and~\ref{d:raether}). In fact, when the right-hand side typing is different from $\one$ we are not able to type the composition of independent processes. Nevertheless, slightly less general forms of composition of global types are still possible. 
For instance, in Def.~\ref{d:raether}, 
we could have combined 
non-annotated and annotated mediums as in, e.g.,
$\ether{G_i} \para \raether{G_j}{}{k}$. 
This resulting ``hybrid annotated medium''  is typable, with type $k{:} \gt{G_j}$ in the right-hand side.
We find these forms of composition useful for modular reasoning on multiparty systems.

\section{Sharing in Finite Multiparty Conversations}\label{s:exx}
Here we further illustrate reasoning about global types in $\fgtypes$ exploiting the properties given in~\S\,\ref{ss:frela}.
In particular, we show that the absence of recursive types does not necessarily preclude specifying and reasoning about non-trivial forms of replication and sharing. 

As an example, let us consider a variant of the 
the two-buyer protocol in~\cite{DBLP:conf/popl/HondaYC08}, in which two buyers ($\pt{B_1}$ and $\pt{B_2}$) 
coordinate to buy an item from a seller ($\pt{S}$). 
The three-party interaction is given by the following global type:
\begin{align*}
G_{\text{BS}} = \mangg{B_1}{S}{send}{str}{\angg{S}{B_1}{rep}{int}{\anggb{S}{B_2}{rep}{int}{\angg{B_1}{B_2}{shr}{int}{\anggd{B_2}{S}{ok}{\one}{\gend}{quit}{\one}{\gend}}}}} 
\end{align*}

We omit the (easy) definition of process \ether{G_{\text{BS}}}, and proceed to examine its properties.
Relying on Theorems~\ref{l:ltypesmedp} and~\ref{l:medltypes}, 
we have the compositional typing:
\begin{equation}
\G;\, c_1{:} \mathsf{B1},   c_2{:} \mathsf{S}, c_3{:} \mathsf{B2} \vdash_\eta \ether{G_{\text{BS}}} ::-{:}\one
\label{ex:int}
\end{equation}
for some $\G$ and with
$\mathsf{B1} = \lt{\proj{G_{\text{BS}}}{\mathtt{B_1}}}$, 
$\mathsf{S} = \lt{\proj{G_{\text{BS}}}{\mathtt{S}}}$, and
$\mathsf{B2} = \lt{\proj{G_{\text{BS}}}{\mathtt{B_2}}}$.
To implement the protocol, one may simply compose 
\ether{G_{\text{BS}}} with type compatible 
processes 
$\cdot; \cdot \vdash \mathit{Buy1} :: c_1{:}\mathsf{B1}$, 
~$\cdot; \cdot \vdash \mathit{Sel} :: c_2{:}\mathsf{S}$, and 
~$\cdot; \cdot \vdash \mathit{Buy2} :: c_3{:}\mathsf{B2}$:
\begin{eqnarray}
\G ; \cdot \vdash_\eta (\nub c_1)(\mathit{Buy1} \!\para\!(\nub c_2)(\mathit{Sel} \para (\nub c_3)(\mathit{Buy2} \!\para\! \ether{G_{\text{BS}}}))) 
\label{ex:comp}
\end{eqnarray}
The binary session types in \S\,\ref{ss:btypes}
allows us to infer that  
the multiparty system defined by~\eqref{ex:comp}
adheres to the declared projected types,  is lock-free, and non-diverging. 
Just as we inherit strong properties for  $\mathit{Buy1}$, $\mathit{Sel}$, and $\mathit{Buy2}$ above, 
we may inherit the same properties for 
more interesting system configurations. 
In particular, local implementations which appeal to replication and sharing, 
admit also precise analyses thanks to the characterizations in~\S\,\ref{ss:frela}. 
Let us consider a setting in which the processes to be composed with the medium must be 
invoked from a replicated service (a source of generic process definitions).
We may  have:
\begin{eqnarray*}
\nnum \cdot; \cdot  \vdash_\eta \, !u_1(w).\mathit{Buy1}_w :: u_1{:}\,!\mathsf{B1} & ~~ &
\nnum \cdot; \cdot  \vdash_\eta \, !u_2(w).\mathit{Sel}_w :: u_2{:}\,!\mathsf{S} \\  
\nnum \cdot; \cdot   \vdash_\eta  \, !u_3(w).\mathit{Buy2}_w :: u_3{:}\,!\mathsf{B2} &&
\end{eqnarray*}
and the following ``initiator processes'' would spawn a copy of the medium's requirements, instantiated at appropriate names:
\begin{align}
\nnum \cdot; u_1{:}\,!\mathsf{B1} & \vdash_\eta \outp{u_1}{x}.\linkr{x}{c_1} :: c_1{:}\mathsf{B1} \\
\nnum \cdot; u_2{:}\,!\mathsf{S} & \vdash_\eta \outp{u_2}{x}.\linkr{x}{c_2} :: c_2{:}\mathsf{S} \\
\nnum \cdot; u_3{:}\,!\mathsf{B2} & \vdash_\eta \outp{u_3}{x}.\linkr{x}{c_3} :: c_3{:}\mathsf{B2} 
\end{align}
Let us write 
$\mathit{RBuy1}$, $\mathit{RBuy2}$, and $\mathit{RSel}$ 
to denote the composition of replicated definitions and initiators above.
Intuitively, they represent the ``remote'' variants
of $\mathit{Buy1}$, $\mathit{Buy2}$, and $\mathit{RSel}$, respectively.
We may then define the multiparty system:
\begin{eqnarray*}
\G ; \cdot \vdash_\eta (\nub c_1)(\mathit{RBuy1} \!\para\!(\nub c_2)(\mathit{RSel} \!\para\! (\nub c_3)(\mathit{RBuy2} \!\para\! \ether{G_{\text{BS}}}))) 
\end{eqnarray*}
which, with a concise specification,
 improves~\eqref{ex:comp}
with concurrent invocation/instantiation of replicated service definitions.
As \eqref{ex:comp}, the revised composition above is correct, lock-free, and terminating.

Rather than appealing to initiators, a scheme in which the medium invokes and instantiates 
services directly is also  expressible in our framework, in a type consistent way.
Using \eqref{ex:int}, and assuming $\G =  u_1{:}\mathsf{B1},   u_2{:}\mathsf{S}, u_3{:}\mathsf{B2}$, 
we may derive:
\begin{equation}
\!\!\!\!\!\!\! \G ;  \cdot \vdash_\eta \outp{u_1}{c_1}.\outp{u_2}{c_2}.\outp{u_3}{c_3}.\ether{G_{\text{BS}}} 
\label{ex:egint}
\end{equation}
Hence, prior to engage in the mediation behavior for $G_{\text{BS}}$, the medium first spawns a copy of the required services.
We may relate the guarded process in \eqref{ex:egint} with the multicast session request construct in
multiparty session processes~\cite{DBLP:conf/popl/HondaYC08}. 
Observe that~\eqref{ex:egint} cleanly distinguishes between session initiation and actual communication behavior:
the distinction is given at the level of processes (cf. output prefixes on $u_1, u_2$, and $u_3$) 
but also at the level of typed interfaces. 

The service invocation \eqref{ex:egint} may be regarded as ``eager'':
all required services must be sequentially invoked prior to executing the protocol. 
We may also obtain, in a type-consistent manner, a medium process implementing a ``lazy'' invocation strategy that
spawns services only when necessary.
For the sake of example, consider process $\mathit{Eager_{\text{BS}}} \triangleq \outp{u_3}{c_3}.\ether{G_{\text{BS}}}$
in which only the invocation on $u_3$ is blocking the protocol,
with ``open'' dependencies on $c_1, c_2$.
That is, we have $\G; c_1{:} \mathsf{B1},   c_2{:} \mathsf{S} \vdash \mathit{Eager_{\text{BS}}} ::z{:}\one$.
It could be desirable to postpone the invocation on $u_3$ as much as possible. 
By combining 
the commutations on process prefixes realised by \eqcc~\cite{DBLP:conf/esop/PerezCPT12}
and Theorem~\ref{t:soundcc},
we may obtain: 
$$
\G; c_1{:} \mathsf{B1},   c_2{:} \mathsf{S} \vdash \mathit{Eager_{\text{BS}}} \tybis \mathit{Lazy_{\text{BS}}} ::-{:}\one
$$
where $\mathit{Lazy_{\text{BS}}}$ is the process obtained 
from $\mathit{Eager_{\text{BS}}}$ by ``pushing inside'' prefix 
$\outp{u_3}{c_3}$ as deep as possible in the process structure.

\section{Further Developments and Extensions}\label{s:exts}
We now briefly describe two possible extensions to multiparty session types which exploit our analysis technique based on mediums.

 \paragraph{Adding A Join Primitive.}
 As we have seen, the mediums offer a clean and simple representation for name-passing in multiparty exchanges.
 Exploiting this feature, we may extend the syntax of global types with a primitive 
 $\mathsf{join}\,s(\pt{r}).G$, which denotes the fact that participant $\pt{r}$, declared in $G$, is to be realized by invoking a shared service $s$.
 This kind of primitive 
 can be found in  Conversation Types~\cite{DBLP:journals/tcs/CairesV10}, but has not been yet considered within 
 multiparty session types. 
 Let $G_1(\pt{r})$ and $G_2(\pt{r})$ be two global types in which participant $\pt{r}$ is declared.
 We may write, e.g., the global type 
$$ 
\gto{p}{q}\big\{\lb{l}_1\langle \mathsf{int}\rangle.\mathsf{join}\,s_1(\pt{r}).G_1(\pt{r}) \, , \,  \lb{l}_2\langle \mathsf{bool}\rangle.\mathsf{join}\,s_2(\pt{r}).G_2(\pt{r})\big\}
$$ 
in which participant $\pt{r}$ may be implemented by different shared services ($s_1$ or $s_2$) depending on the selected label.
 The medium for this primitive would be:
 $$
 \ether{\mathsf{join}\,s(\pt{r}).G} = \outp{s}{c_{\pt{r}}}.\ether{G}
 $$
 Suppose that  $\lt{\proj{G(\pt{r})}{\pt{r}}} = A_{\pt{r}}$.
 The above medium could be typed in the system of~\cite{CairesP10,Toninho2014} as follows:
$$
\G \, ; \, \D,  s:!A_{\pt{r}}, \vdash_\eta \outp{s}{c_{\pt{r}}}.\ether{G} :: z{:}\one
$$
where $\D$ describes the behaviors of other participants in $G$, 
reflecting the fact that $s$ is a shared service.


\paragraph{Parametric Polymorphism.}
Building upon mediums, we may also extend known multiparty session type theories 
with features well-understood in the binary setting but not yet developed for multiparty sessions. 
A particularly relevant such features is \emph{parametric polymorphism} (in the style of the Girard-Reynolds polymorphic $\lambda$-calculus), studied for binary sessions by Caires et al.~\cite{DBLP:conf/esop/CairesPPT13} and Wadler~\cite{DBLP:journals/jfp/Wadler14}. We do not know of multiparty session theories supporting polymorphism; so an extension through our approach would be particularly 
significant.

We follow the approach in~\cite{DBLP:conf/esop/CairesPPT13}, which 
extends the system of~\cite{CairesP10} 
with two kinds of session types, $\forall X.A$ and $\exists X. A$,
corresponding to impredicative universal and existential
quantification over sessions. They are interpreted as 
the input and output of a session type, respectively. The syntax of processes is extended accordingly, with prefixes
$\out{x}A.P$ and $x(X).P$. 
To define global types with polymorphism, we may extend the syntax of $U$ in Def.~\ref{d:gltypes} with session types $A$. Here is a simple example of a polymorphic global type: 
$$
G_{\text{poly}} = \gto{p}{q}\{\lb{l}_1\langle A\rangle.G\}
$$
Global type 
$G_{\text{poly}}$ abstracts a scenario in which \pt{p} sends to \pt{q} a session type $A$ using label $\lb{l}_1$.
This means that 
the local implementation for \pt{q} should be \emph{parametric}
on any session type which is to be received from \pt{p}. We would have the following medium for $G_{\text{poly}}$:
$$\ether{\gto{p}{q}\{\lb{l}_1\langle A\rangle.G\}}   =
\mycasebig{c_\pt{p}}{\lb{l}_1 : c_\pt{p}(X).\mysel{c_\pt{q}}{\lb{l}_1};\out{c_\pt{q}}{X}.\ether{G_i} }{} $$
It is worth stressing that this extension should be completely orthogonal to the results in \S\,\ref{ss:frela},
for the polymorphic binary sessions in~\cite{DBLP:conf/esop/CairesPPT13} are type-preserving, deadlock-free, and terminating.
\section{Related Work}\label{s:relwork}
As already discussed, the key obstacle  in 
reducing multiparty session types into binary ones
consists in defining binary fragments which preserve the sequencing information of the global specification. 
\emph{Correspondence assertions}~\cite{Bonelli05} offer one way of 
relating otherwise independent binary sessions.
Present in the syntax of processes and types, such assertions 
may track data dependencies and detect unintended operations. 
Retaining a standard syntax for binary and multiparty types,
here we capture the sequencing information using a process extracted from a global type.
Our  approach relies on deadlock-freedom (not available in~\cite{Bonelli05}) and 
offers a principled way of transferring it to multiparty systems.
In~\cite{Toninho2011,Toninho2011b} Toninho et al. studied the integration of assertions in session types 
via dependent types and authorization logics,
allowing expressive certified contracts; 
based on the results in this paper, the added expressiveness brought in by~\cite{Toninho2011,Toninho2011b} would carry to the multiparty setting, similarly as described for parametric polymorphism in \S\,\ref{s:exts}.

Typed frameworks  for multiparty interactions were first proposed in~\cite{DBLP:conf/tgc/BonelliC07,DBLP:conf/popl/HondaYC08}. 
To  our knowledge, ours is the first formal characterization
of multiparty session types using binary session types.
Previous works 
have encoded binary session types into other type systems.
For instance,~\cite{DBLP:conf/ppdp/DardhaGS12} 
encodes binary session types into the linear  types of~\cite{DBLP:conf/popl/KobayashiPT96}. 
Combined with~\cite{DBLP:conf/ppdp/DardhaGS12}, our work connects a standard theory of global types with 
the linear types of~\cite{DBLP:conf/popl/KobayashiPT96}; this further results appears new, and
deserves investigation. 
Related to this, 
as a case study for a theory of deadlock-free processes,
 the work~\cite{padovani:hal-00932356}  identified a class of multiparty systems for which 
 the analysis of deadlock-freedom can be reduced to analysis of linear $\pi$-calculus processes.
 In contrast with our work, the reduction in~\cite{padovani:hal-00932356} 
 does not establish formal connections with binary session types, nor exploits other properties of processes
 to reason about global specifications.

 Building upon~\cite{CairesP10}, in~\cite[Ch. 4]{Montesi13} a correspondence between
 two-party choreographies and proofs from LCL, a linear logic with hypersequents, is given.
 Projection is cleanly defined at the level of proofs, but
 the analysis of $n$-ary choreographies, exponentials, and forms of iterative behavior are left for future work.


Our medium processes, the key technical device in our developments, 
are loosely related to the concept of
\emph{orchestrators}  in service-oriented programming. 
The work~\cite{DBLP:conf/apccm/McIlvennaDW09} shows how to synthesize a orchestrator from a  service choreography, using finite state machines to model both choreography and orchestrator, which already distinguishes~
this work
from ours. We consider choreographies specified as behavioral  types; mediums are processes obtained directly from those types. Our work has a foundational character, for it formally connects communicating automata (related to global types) and a Curry-Howard correspondence based on linear logic propositions (which supports binary session types); in contrast, the results in~\cite{DBLP:conf/apccm/McIlvennaDW09} have a more pragmatic spirit, for the goal is to generate Petri net and BPMN models from the obtained orchestrator.

\section{Concluding Remarks}\label{s:concl}

We have developed a comprehensive  
 analysis of multiparty session types 
 on top of an elementary type theory for binary sessions.
Our results rely on
\emph{medium processes}, 
a simple but effective characterization of multiparty interactions as expressed by standard global types. 
Using well-typed mediums under the theory of (linear logic based) binary session types in~\cite{CairesP10,Toninho2014}
we obtained strong characterizations of mediums with respect to the projections of a global type.
Such characterizations allow us
to uniformly transfer to the multiparty setting key properties of the binary session theory (notably, 
deadlock-freedom and
behavioral equivalences). 
In our view, our characterizations do not diminish the relevance  of 
existing 
frameworks of multiparty sessions. Rather, it is most reasonable that 
in applications 
the analysis of multiparty protocols can be effectively done with 
frameworks in which multiparty interaction is a first-class idiom. 
On the other hand, our results
provide further evidence of the fundamental character of key ingredients in multiparty session types,
and build on (perhaps unexpected, but certainly welcome) 
tight connections between two independently motivated
theories of session types with foundational significance: 
one
based on linear logic~\cite{CairesP10}, 
the other 
 based on
communicating automata~\cite{DBLP:conf/icalp/DenielouY13}. These correspondences should be further explored
 in future research, 
in connection with more expressive types (e.g., dependent types) and computational models 
(e.g. asynchrony).

\bibliographystyle{abbrvnat}
 \bibliography{shortreferen}


\newcommand{\jblue}[1]{{\color{blue} #1}}

\appendix
\onecolumn
\tableofcontents

\section{Additional Definitions for \S\,\ref{s:tmedium}}

\subsection{Independent Composition}\label{app:indc}
As mentioned in \S\,\ref{ss:btypes}, the 
following rule 
for \emph{independent parallel composition} is derivable: 
\begin{equation}
\inferrule*[left=\name{\textsc{indComp}}]{\Gamma; \Delta_1 \vdash P :: - : \one \quad \Gamma; \Delta_2 \vdash Q :: z{:}C}{\Gamma; \Delta_1,\Delta_2 \vdash P \para Q :: z{:}C} \nonumber
\end{equation}
where `$-$' denotes a ``dummy name'' not in $\fn{P}$.

\subsection{Simple Projectability and Well-Formedness}
\begin{definition}[Simple Projection \cite{DBLP:conf/popl/HondaYC08}]\label{d:sproj}
Let $G$ be a global type.
The \emph{simple projection} of $G$ under participant $\pt{r}$, 
denoted \sproj{G}{\pt{r}}, is inductively defined as follows:
\begin{enumerate}[$\bullet$]
\item $\sproj{\gend}{\pt{r}}  =  \lend$
\item $\sproj{\gto{p}{q}\{\lb{l}_i\langle U_i\rangle.G_i\}_{i \in I}}{\pt{r}}  = 
  \begin{cases}
\pt{p}!\{\lb{l}_i \langle U_i\rangle.\sproj{G_i}{\pt{r}}\}_{i \in I} & \text{if $\pt{r} = \pt{p}$} \\
\pt{p}?\{\lb{l}_i \langle U_i\rangle.\sproj{G_i}{\pt{r}}\}_{i \in I} & \text{if $\pt{r} = \pt{q}$} \\
\sproj{G_1}{\pt{r}} & \text{if $\pt{r} \neq \pt{p}$ and $\pt{r} \neq \pt{q}$} 
\text{ and $\forall i, j \in I. \sproj{G_i}{\pt{r}} = \sproj{G_j}{\pt{r}}$}
\end{cases} $
\item $\sproj{(G_1 \para G_2) }{\pt{r}}  =  
\begin{cases} \sproj{G_i}{\pt{r}} & \text{if $\pt{r} \in G_i$ and $\pt{r} \not\in G_j$, with $i \neq j \in \{1,2\}$ } \\
														 \lend & \text{if  $\pt{r} \not\in G_1$ and $\pt{r} \not\in G_2$}
											\end{cases} $
\end{enumerate}
When a side condition does not hold, the map is undefined.
\end{definition}

We then may define \emph{simple well-formedness}:

\begin{definition}[Simple Well-Formedness]\label{d:swfltypes}
We say that global type $G$ is \emph{simply well-formed (SWF, in the following)} if for all $\pt{r} \in G$, 
the simple projection $\sproj{G}{\pt{r}}$ is defined.
\end{definition}

In what follows, we often say that a global type is MWF (merge-based well-formed) if it is well-formed according to Def.~\ref{d:wfltypes}.

\subsection{Proof Conversions}
Figs.~\ref{f:cc1} and~\ref{f:cc2} give the commuting conversions relevant to the present development, in particular for Theorem~\ref{p:swapmeds} (proven in \S\,\ref{app:swapmeds}).
Intuitively, they concern the interaction of (i) two left rules and (ii) a left rule and a rule for composition.
There are other commuting conversions 
(see, e.g.,~\cite{DBLP:conf/esop/PerezCPT12}); however, given our focus on
compositional typings (in which the only typing admitted in the right-hand side typing is $\one$, cf. Def.~\ref{d:compty})
the conversions in Figs.~\ref{f:cc1} and~\ref{f:cc2} are the only relevant ones. 
In the figures, we sometimes appeal to the 
following 
notational abbreviations:

\begin{convention}[Additives]\label{cv:types}
We abbreviate 
$\mywith{\lb{l}_i{:}A_i}{i \in I}$ as 
 $\mywith{\lb{l}_i{:}A_i, \lb{l}_j{:}A_j}{}$
 when $I = |2|$.
 When labels are unimportant, 
 we write  $A_i\with A_j$, with labels having left/right readings, as in~\cite{CairesP10}.
 Similar abbreviations apply for
$\myoplus{\lb{l}_i{:}A_i}{i \in I}$.
\end{convention}

We may define:

\begin{definition}[Proof Conversions]\label{d:eqcc}
We define \eqcc as the least congruence on 
processes induced by the process equalities in Figures~\ref{f:cc1} and~\ref{f:cc2}
(Pages \pageref{f:cc1}--\pageref{f:cc2}).
\end{definition}

\begin{figure}[!t]
\vspace{-2ex}
{\small
\begin{align}
\tag{\text{I-}3} \G ; \D, y{:}A\otimes B & \vdash (\nu x)(P \para y(z).Q)  \typconp   y(z).(\nu x) (P \para Q)  :: k{:}E \\
\tag{\text{I-}4} \G ; \D, y{:}A\otimes B & \vdash (\nu x)(y(z).P \para Q)  \typconp   y(z).(\nu x) (P \para Q)  :: k{:}E \\
\tag{\text{I-}6} \G ; \D, y{:}A\lolli B & \vdash (\nu x)(P \para \outp{y}{z}.(Q \para R))  \typconp  \outp{y}{z}.((\nu x)(P \para  Q)  \para R) :: k{:}E \\
\tag{\text{I-}7} \G ; \D, y{:}A\lolli B & \vdash (\nu x)(P \para \outp{y}{z}.(Q \para R))   \typconp   \outp{y}{z}.(Q \para (\nu x)(P \para  R)) :: k{:}E \\
\tag{\text{I-}8} \G ; \D, y{:}A\lolli B & \vdash (\nu x)(\outp{y}{z}.(Q \para P) \para R)   \typconp  \outp{y}{z}.(Q \para (\nu x)(P \para  R)) :: k{:}E \\
\tag{\text{I-}10} \G ; \D, y{:}A \with B  & \vdash (\nu x)(P \para \mysel{y}{\lb{l}_i}; Q)  \typconp    \mysel{y}{\lb{l}_i};(\nu x)(P \para  Q)  ::k{:}E \\
\tag{\text{I-}14} \G ; \D,y{:}A\oplus B & \vdash (\nu x)(P \para \mycaseb{y}{Q, R}{})  \typconp \mycaseb{y}{(\nu x)(P \para  Q), (\nu x)(P \para R)}{} ::k{:}E  \\
\tag{\text{I-}15} \G,u{:}A ; \D & \vdash (\nu x)(P \para \outp{u}{y}.Q)   \typconp  \outp{u}{y}.(\nu x)(P \para Q)::k{:}E \\  
\tag{\text{I-}18} \G ; \D,y{:}A \with B & \vdash (\nu x)( \mysel{y}{\lb{l}_i}; P \para R )  \typconp    \mysel{y}{\lb{l}_i};(\nu x)(P \para  R) ::k{:}E \\
 \nnum 
\tag{\text{I-}20}
 \G ; \D,y{:}A \oplus B & \vdash (\nu x)(\mycaseb{y}{P, Q}{} \para  R )  \typconp   \mycaseb{y}{(\nu x)(P \para  R), (\nu x)(Q \para R)}{}::k{:}E \\
\tag{\text{I-}21} \G ; \D & \vdash (\nu x)(P \subst{y}{u} \para  Q )  \typconp    (\nu x)(P \para Q)\subst{y}{u}::k{:}E \\
\tag{\text{I-}22} \G ; \D & \vdash (\nu x)(P  \para  Q\subst{y}{u} )  \typconp    (\nu x)(P \para Q)\subst{y}{u}::k{:}E \\
\tag{\text{I-}23} \G,u{:}A ; \D & \vdash (\nu x)(\outp{u}{y}.P  \para  R)  \typconp    \outp{u}{y}.(\nu x)(P \para R)  :: k{:}E \\
\tag{\text{I-}24} \G,u{:}A ; \D & \vdash (\nu x)( P  \para  \outp{u}{y}.R)  \typconp    \outp{u}{y}.(\nu x)(P \para R)  :: k{:}E \\
 \G ; \cdot   \vdash & ~(\nu u )( (!u(y).P)\para \zero ) \typconp   \zero  :: -{:}\one & \tag{\text{I-}25}   \\
  \G ; \D,y{:}A \otimes B  \vdash & ~(\nu u )( (!u(y).P) \para y(z).Q) \typconp  y(z).(\nu u )( (!u(y).P) \para Q)::k{:}E & \tag{\text{I-}28}\\
 \G ; \D,y{:}A\lolli B  \vdash & ~(\nu u )( (!u(y).P) \para  \outp{y}{z}.(Q\para R)) \typconp & \nnum \\
&     ~~ \outp{y}{z}.( ((\nu u )(!u(y).P \para Q) \para (\nu u )( (!u(y).P) \para R))) ::k{:}E & \tag{\text{I-}30}\\
 \G ; \D,y{:}A \with B  \vdash &  ~(\nu u )(!u(z).P \para \mysel{y}{\lb{l}_i};Q)   \typconp  \mysel{y}{\lb{l}_i};(\nu u )( !u(z).P \para Q) :: k{:}E & \tag{\text{I-}32} & \nonumber \\
 \G ; \D, y{:}A \oplus B  \vdash & ~(\nu u )( !u(z).P \para \mycaseb{y}{Q, R}{} )  \typconp  & \nonumber \\
 &    \qquad \mycaseb{y}{(\nu u )( !u(z).P \para Q), (\nu u )( !u(z).P \para R)}{}  :: k{:}E & \tag{\text{I-}36}\\
 \G ; \D  \vdash & ~(\nu u )( !u(y).P \para Q\subst{y}{v}) \typconp     (\nu u )( !u(y).P \para Q) \subst{y}{v} {::}k{:}E & \tag{\text{I-}38}\\
 \G ; \D  \vdash & ~(\nu u )( !u(y).P \para \outp{v}{y}.Q)   \typconp  \outp{v}{y}.(\nu u )( !u(y).P \para Q) ) :: k{:}E & \tag{\text{I-}39}
\end{align} 
\vspace{-2ex}
\caption{Process equalities induced by  proof conversions (Part I - cf. Def.~\ref{d:eqcc})}\label{f:cc1}
}
\end{figure}


\begin{figure}[!t]
\vspace{-2ex}
{\small
\begin{align}
\tag{\text{II-}1}\Gamma; \Delta, x{:}A\otimes B, z{:}C\otimes D  &\vdash  x(y).z(w).P \pceq z(w).x(y).P :: k{:}E \\
\tag{\text{II-}2}
\Gamma; \Delta, z{:}D \lolli C, x{:}A \lolli B & \vdash  \outp{z}{w}. (R \para \outp{x}{y}.(P \para Q)) \pceq \outp{x}{y}.(P \para \outp{z}{w}. (R \para Q)) :: k{:}E\\
\tag{\text{II-}3} \Gamma; \Delta, z{:}D \lolli C, x{:}A \lolli B & \vdash  \outp{z}{w}. (R \para \outp{x}{y}.(P \para Q)) \pceq  \outp{x}{y}.(\outp{z}{w}. (R \para P) \para  Q ) :: k{:}E\\
\tag{\text{II-}4}\Gamma; \Delta, w{:}C \lolli D, x{:}A \otimes B & \vdash  \outp{w}{z}.(Q \para x(y).P) \pceq x(y).\outp{w}{z}.(Q \para P) :: k{:}E \\
\tag{\text{II-}5}\Gamma; \Delta, w{:}C \lolli D, x{:}A \otimes B & \vdash  \outp{w}{z}.( x(y).P  \para Q) \pceq x(y).\outp{w}{z}.(P \para Q) :: k{:}E \\
\tag{\text{II-}6}\Gamma,u{:}A,v{:}C; \Delta & \vdash   \outp{u}{y}.\outp{v}{x}.  P \pceq \outp{v}{x}. \outp{u}{y}. P  :: k{:}E \\
\tag{\text{II-}7}\Gamma, u{:}C; \Delta,x{:}A \lolli B & \vdash  \outp{u}{ z}.\outp{x}{y}.(P \para Q) \pceq \outp{x}{y}.(\outp{u}{ z}.P \para Q):: k{:}E \\
\tag{\text{II-}8}\Gamma, u{:}C; \Delta,x{:}A \lolli B & \vdash  \outp{u}{ z}.\outp{x}{y}.(P \para Q) \pceq \outp{x}{y}.(P \para \outp{u}{ z}.Q):: k{:}E \\
\tag{\text{II-}9}\Gamma,u{:}A; \Delta, z{:}C \otimes D & \vdash  \outp{u}{y}.z(w).P \pceq z(w).\outp{u}{y}.P :: k{:}E \\
\Gamma; \Delta, x{:}A \oplus B, y{:}C \oplus D & \vdash  \mycaseb{y}{\mycaseb{x}{P_{1},Q_{1}}{}, \mycaseb{x}{P_{2},Q_{2}}{}}{}  \pceq \mycaseb{x}{\mycaseb{y}{P_{1},P_{2}}{}, \mycaseb{y}{Q_{1},Q_{2}}{}}{}::k{:}E \tag{\text{II-}10} \\
\tag{\text{II-}11}\Gamma,u{:}C; \Delta,x{:}A \oplus B & \vdash \outp{u}{z}.\mycaseb{x}{P, Q}{} \pceq \mycaseb{x}{\outp{u}{z}.P \, , \, \outp{u}{z}.Q}{} :: k{:}E\\
\Gamma; \Delta, w{:}A \lolli E, z{:}C \oplus D & \vdash  \mycaseb{z}{\outp{w}{y}.(P \para R_1) \, , \, \outp{w}{y}.(P \para R_2)}{} \pceq  \outp{w}{y}.(P \para \mycaseb{z}{R_{1}, R_{2}}{})::k{:}E  & \tag{\text{II-}12} \\
\tag{\text{II-}13} \Gamma; \Delta, z{:}C \oplus D, x{:}A \otimes B & \vdash  \mycaseb{z}{x(y).P, x(y).Q}{} \pceq x(y).\mycaseb{z}{P, Q}{} :: k{:}E \\
\tag{\text{II-}14}\Gamma; \Delta, x{:}A \with B, y{:}C \with D & \vdash  \mysel{x}{\lb{l}'_i};\mysel{y}{\lb{l}_i};P \pceq \mysel{y}{\lb{l}_i};\mysel{x}{\lb{l}'_i};P ::k{:}E \\
\tag{\text{II-}15} \Gamma; \Delta, x{:}A \oplus B, y{:}C \with D & \vdash  \mycaseb{x}{\mysel{y}{\lb{l}_i};P , \mysel{y}{\lb{l}_i};Q}{} \pceq \mysel{y}{\lb{l}_i};\mycaseb{x}{P, Q}{}::k{:}E  \\
\tag{\text{II-}16}\Gamma,u{:}C;\Delta,  z{:}A \with B & \vdash  \mysel{z}{\lb{l}_i};\outp{u}{y}.P \pceq \outp{u}{y}.\mysel{z}{\lb{l}_i};P :: k{:}E\\
\tag{\text{II-}17}\Gamma; \Delta, z{:}C \with D, x{:}A \lolli B & \vdash  \mysel{z}{\lb{l}_i}; \outp{x}{y}.(P \para Q) \pceq  \outp{x}{y}.(\mysel{z}{\lb{l}_i};P \para Q)::k{:}E \\
\tag{\text{II-}18}\Gamma; \Delta, z{:}C \with D, x{:}A \lolli B & \vdash  \mysel{z}{\lb{l}_i};\outp{x}{y}.(P \para Q) \pceq  \outp{x}{y}.(P \para \mysel{z}{\lb{l}_i};Q)::k{:}E \\
\tag{\text{II-}19}\Gamma; \Delta, z{:}C \with D, x{:}A \otimes B & \vdash  \mysel{z}{\lb{l}_i};x(y).P \pceq x(y).\mysel{z}{\lb{l}_i};P::k{:}E 
\end{align}
\caption{Process equalities induced by proof conversions, second kind (Part II - cf. Def.~\ref{d:eqcc}). 
}\label{f:cc2}
}
\vspace{-2ex}
\end{figure}

\subsection{Labeled Transition Systems for Processes and Global Types}\label{app:lts}
We now present auxiliary notions, needed for the operational correspondence result stated in \S\,\ref{s:opcorr} (and proved in \S\,\ref{app:opcorr}).

\paragraph{LTS for Processes.}
To characterize the interactions of a well-typed process with its environment,
we extend
the 
early labeled transition system (LTS) for the
$\pi$-calculus~\cite{sangiorgi-walker:book} with 
labels and transition rules for  choice and forwarding constructs.  
A transition
$P\tra{\,\labelset\,}Q$ denotes that 
$P$ may evolve to 
$Q$
by performing the action represented by label $\labelset$. 
Transition labels are defined below:
\begin{eqnarray*}
\labelset & ::= &  \tau  \sep x(y)  \sep  \mysel{x}{\lb{l}}    
  \sep   \out{x}y  \sep  \outa{x}{y} \sep  \ov{\mysel{x}{\lb{l}}}   
\end{eqnarray*}
Actions are name input $x(y)$, 
the 
 offer
$\mysel{x}{\lb{l}}$,   and their matching
co-actions,
respectively the output $\out{x}y $ and bound output $\outa{x}{y}$ actions, 
and the selection
$\ov{\mysel{x}{\lb{l}}}$. 
The bound output $\outa{x}{y}$ denotes extrusion of a fresh name $y$
along 
$x$. Internal action is denoted by $\tau$. 
In general, an action
requires a matching 
co-action
in the environment to enable progress. 

\begin{definition}[Labeled Transition System]\label{def:lts} 
  The relation \emph{labeled transition} ($ P\tra{\labelset}Q$) is defined
  by the rules in Fig.~\ref{fig:LTS},
   subject to the side
  conditions: in rule $(\mathsf{res})$, we require $y_{}\not\in\fn{\labelset}$; 
  in rule $(\mathsf{par})$, we require $\bn{\labelset} \cap \fn{R} = \emptyset$; in rule
  $(\mathsf{close})$, we require $y_{}\not\in\fn{Q}$. We omit the symmetric versions
  of rules $(\mathsf{par})$, $(\mathsf{com})$, and $(\mathsf{close})$.
  \end{definition} 
  
We write $subj(\lambda)$ for the subject of the action $\lambda$, that
is, the channel along which the action takes place. Weak transitions are defined as usual.
Let us write $\rho_1 \rho_2$ for the composition of relations $\rho_1, \rho_2$
and $\wtra{}$ for the reflexive, transitive closure of
$\tra{\tau}$. 
Notation $\wtra{\labelset}$ stands for $\wtra{~}\tra{\labelset}\wtra{~}$ (given $\labelset \neq \tau$)
and $\wtra{\tau}$ stands for $\wtra{}$.
We recall basic facts about reduction, structural congruence,
and labeled transition: closure of labeled transitions under
structural congruence, and coincidence of $\tau$-labeled transition
and reduction~\cite{sangiorgi-walker:book}: (1) if $P
\equiv\tra{\labelset}Q$ then $P \tra{\labelset}\equiv Q$; 
(2) $P\to Q$ iff $P \tra{\tau} \equiv Q$.

\begin{figure}[t!]
{\small
 $$
\begin{array}{c}
\inferrule*[Left=\name{$\mathsf{id}$}]{}{(\nub x_{})(\linkr{x_{}}{y_{}} \para P) \tra{\,\tau\,}  P\subst{y_{}}{x_{}}} 
\hspace{1.5cm}
\inferrule*[Left=\name{$\mathsf{rep}$}]{}{\bang x_{}(y_{}).P \tra{x(z)} P \subst{z_{}}{y_{}}\para \bang x_{}(y_{}).P}
\vspace{0.2cm}
\\
\inferrule*[Left=\name{$\mathsf{n.out}$}]{}{\out{x} y.P \tra{\ov{x \, y_{}}} P}
\hspace{1.5cm}
\inferrule*[Left=\name{$\mathsf{n.in}$}]{}{x_{}(y).P \tra{x_{}(z_{})} P \subst{z_{}}{y}}
\hspace{1.5cm}
\vspace{0.2cm}
\inferrule*[Left=\name{$\mathsf{s.out}$}]{}{\mysel{x}{\lb{l}};P \tra{ \overline{\mysel{x}{\,\lb{l}}} } P}
\hspace{1.5cm}
\inferrule*[Left=\name{$\mathsf{s.in}$}]{}{\mycase{x}{\lb{l}_i : P_i}{i \in I} \tra{\mysel{x}{\,\,\lb{l}_j}} P_j} \quad (j \in I)
\\
\inferrule[\name{$\mathsf{par}$}]{P \tra{\labelset} Q}{P\para R \tra{\labelset} Q \para R}
\hspace{0.5cm} 
\inferrule[\name{$\mathsf{com}$}]{P\tra{\overline{\labelset}} P' \;\;\; Q \tra{\labelset} Q'} {P \para
  Q \tra{\tau} P' \para Q'}
\hspace{0.5cm}
\inferrule[\name{$\mathsf{res}$}]{P \tra{\labelset} Q} {(\nub y)P \tra{\labelset} (\nub y_{})Q}
  \vspace{0.1cm}
\hspace{0.5cm} 
\inferrule[\name{$\mathsf{open}$}]{P \tra{\out{x}y} Q} {(\nub y_{})P \tra{\outa{x}{y}} Q}
\hspace{0.5cm}
\inferrule[\name{$\mathsf{close}$}]{P \tra{\outa{x}{y}} P' \;\;\; Q \tra{x(y)} Q'}
{P \para Q \tra{\tau} (\nub y_{})(P' \para Q')}
\end{array}
$$
}
\caption{\label{fig:LTS} LTS for Processes.}
\end{figure}

\paragraph{LTS for Global Types With Recursion.}
We define an LTS over extended global types, which adapts the LTS in~\cite{DBLP:conf/icalp/DenielouY13}, 
with refined intermediate steps. We first define the set of observables:
\begin{eqnarray*}
\gtlabelset & ::= &  \pt{p}  \sep  \mysel{\pt{p}}{\lb{l}}   \sep  \out{\pt{p}}  \sep   \ov{\mysel{\pt{p}}{\lb{l}}}   
\end{eqnarray*}

\begin{definition}[LTS on Global Types]
  The relation \emph{labeled transition over global types} ($ G\tra{\gtlabelset}G'$) is defined
  by the rules in Figure~\ref{f:gtltss}.
  \end{definition}

  \begin{figure}[t!]
  $$
\begin{array}{c}
\inferrule*[Left=\name{G1}]{}{\gto{p}{q}\{\lb{l}_i\langle U_i\rangle.G_i\}_{i \in I}  \tra{\,\ov{\mysel{\pt{p}}{~\lb{l}_j}}\,}  \igto{p}{q}\lb{l}_j\langle U_j\rangle.G_j}  \quad (j \in I) 
\vspace{0.2cm}
\\
\inferrule*[Left=\name{G2}]{}{\igto{p}{q}\lb{l}\langle U\rangle.G  \tra{\,\out{\pt{p}}\,}  \igto{p}{q}\lb{l}(\!(U)\!).G}  
\hspace{1.5cm}
\inferrule*[Left=\name{G3}]{}{ \igto{p}{q}\lb{l}(\!(U)\!).G \tra{\,\mysel{\pt{q}}{~\lb{l}}\,}  \igto{p}{q}(\!(U)\!).G }  
\vspace{0.2cm}
\\
\inferrule*[Left=\name{G4}]{ }{ \igto{p}{q}(\!(U)\!).G \tra{\,{\pt{q}}\,}  G}  
\hspace{1.5cm}
\inferrule*[Left=\name{G5}]{G\subst{\mu\rv{X}.G}{\rv{X}} \tra{\gtlabelset}  G' }{ \mu\rv{X}.G \tra{\gtlabelset}  G'}  
\end{array}
$$
\caption{LTS over (Extended) Global Types \label{f:gtltss}}
\vspace{-4mm}
\end{figure}

\paragraph{Relating Labels.}
We now relate the labels for a global type $G$  and the labels for processes in $\mathcal{S}^k(G)$.
The latter are defined as follows, with associated transition rules as in, e.g.,~\cite{CairesP10}:
\begin{eqnarray*}
\labelset & ::= &  \tau  \sep x(y)  \sep  \mysel{x}{\lb{l}}    
  \sep   \out{x}y  \sep  \outa{x}{y} \sep  \ov{\mysel{x}{\lb{l}}}   
\end{eqnarray*}

This relation is tight, with minor differences on output actions:

\begin{definition}\label{d:labelst}
Let $k$ and $\pt{p}$ be a name and a participant identity, respectively.
The mapping $\mlab{\cdot}{k}$ from global type labels $\gtlabelset$ to process labels 
$\labelset$,
and the mapping $\glab{\cdot}{\pt{p}}$ from process labels $\labelset$ 
to global type labels $\gtlabelset$ 
are defined inductively as follows:
$$
\begin{array}{ccc}
\mlab{\pt{p}}{k}  =  k & & \glab{k}{\pt{p}}  =  \pt{p} \\
\mlab{\mysel{\pt{p}}{\lb{l}}}{k}  =  \mysel{k}{\lb{l}} & & \glab{\mysel{k}{\lb{l}}}{\pt{p}}  =   \mysel{\pt{p}}{\lb{l}}\\
\mlab{\out{\pt{p}}}{k}  =  \about{k} & & \glab{\about{k}}{\pt{p}}  =  \out{\pt{p}} \\ 
\mlab{\,\ov{\mysel{\pt{p}}{\lb{l}}} \,}{k}  =  \ov{\mysel{k}{\lb{l}}} & & \glab{\,\ov{\mysel{k}{\lb{l}}}\,}{\pt{p}}  =  \ov{\mysel{\pt{p}}{\lb{l}}}
\end{array}
$$
\end{definition}

\paragraph{Extended Global Types and Annotated Mediums.}
To establish operational correspondence, we consider \emph{extended} global types, defined as follows:
\begin{eqnarray}
G & ::= & \gend \sep \gto{p}{q}\{\lb{l}_i\langle U_i\rangle.G_i\}_{i \in I} \sep \mu \rv{X}.G \sep \rv{X} \nonumber \\ 
& \sep &  \igto{p}{q}\lb{l}\langle U\rangle.G \sep \igto{p}{q}\lb{l}(\!(U)\!).G \sep \igto{p}{q}(\!(U)\!).G \nonumber
\end{eqnarray}
We have introduced three auxiliary forms for global types; denoted with $\leadsto$, 
they represent intermediate steps,  as we describe next.

First, global type $ \igto{p}{q}\lb{l}\langle U\rangle.G$ denotes the commitment of $\pt{p}$ to \emph{output} along label $\lb{l}$.
The global type $ \igto{p}{q}\lb{l}(\!(U)\!).G$ denotes the commitment of $\pt{q}$ to \emph{input} along $\lb{l}$.
Finally, type
$\igto{p}{q}(\!(U)\!).G$ represents the state just before the actual input action by $\pt{q}$.
The definition of annotated mediums (Def.~\ref{d:raether}) is extended as well:
\begin{enumerate}[$\bullet$]
\item $\raether{ \igto{p}{q}\lb{l}\langle U\rangle.G }{\lb{l}}{k}    = $ \\
$ c_\pt{p}(u).\about{k}.\big(\mysel{c_\pt{q}}{\lb{l}};\mycase{k}{\lb{l} : \outp{c_\pt{q}}{v}.( \linkr{u}{v} \para k.\raether{G}{\lb{l}}{k} )}{}\, \big) $

\item $\raether{\igto{p}{q}\lb{l}(\!(U)\!).G}{\lb{l}}{k}    = $ \\
$\mysel{c_\pt{q}}{\lb{l}};\mycase{k}{\lb{l} : \outp{c_\pt{q}}{v}.( \linkr{u}{v} \para k.\raether{G}{\lb{l}_i}{k} )}{} $

\item $\raether{\igto{p}{q}(\!(U)\!).G}{\lb{l}}{k}    = \outp{c_\pt{q}}{v}.( \linkr{u}{v} \para k.\raether{G}{\lb{l}_i}{k} )  $
\end{enumerate}

\section{Omitted Proofs from \S\,\ref{ss:frela}} 

\subsection{Relating SWF Global Types and Typed Mediums\label{app:sec:swf}}
Following Remark~\ref{rem:swf}, 
here we consider the analogous of Theorems~\ref{l:ltypesmedp} and \ref{l:medltypes} but in the setting of global types
which are SWF.

\begin{proposition}\label{p:indpc}
Let $G = G_1 \para G_2$ be a 
global type in \fgtypes. 
If $\G; \D \vdash \ether{G}$
is a compositional typing 
then there exist disjoint $\D_1, \D_2$ such that: 
(i)~$\D = \D_1 \cup \D_2$, and 
(ii) $\G; \D_1 \vdash \ether{G_1} $,
and 
$\G; \D_2 \vdash \ether{G_2} $ are compositional typings.
\end{proposition}

\begin{proof}
By inversion on typing. By assumption the typing for \ether{G} is a compositional one; hence, 
there is a $c_j{:}A_j \in \D$ for each $\mathtt{r}_j \in \partp{G} = \partp{G_1} \cup \partp{G_2}$.
This means that, necessarily, the typed parallel composition between \ether{G_1} and \ether{G_2} is independent
(in the sense of the derived rule~$\name{\textsc{indComp}}$ in \S\,\ref{app:indc}).
In fact, a non independent composition (i.e., using rule~\name{T$\cut$}) would contradict the assumption that 
$c_j{:}A_j \in \D$ for each $\mathtt{r}_j \in \partp{G} = \partp{G_1} \cup \partp{G_2}$, for there would have to exist
a participant $\pt{r}_k \in \partp{G_1} \cup \partp{G_2}$ but without a  $c_k : A_k \in \D$.
\end{proof}

\subsubsection{SWF Global Types Ensure Compositional Typings\label{app:ltypesmed}}

We state and prove 
Theorem~\ref{l:ltypesmed}, the analogous of Theorem~\ref{l:ltypesmedp}. 
The proof uses the following auxiliary proposition, whose proof follows by construction.

\begin{proposition}\label{p:gtypecons}
If $ \gto{p_1}{p_2}\{\lb{l}_i\langle U_i\rangle.G^i\}_{i \in I}$ is 
SWF
in \fgtypes
then
all $G^i$ ($i \in I$) are 
SWF
too. Also, if 
$G = G_1 \para G_2$ is 
SWF
then $G_j$ ($j \in \{1,2\}$) are 
SWF
too.
\end{proposition}

\begin{theorem}\label{l:ltypesmed}
Let $G \in \fgtypes$ be a global type, with $\partp{G} = \{\pt{p}_1, \ldots, \pt{p}_n\}$.
If $G$ is SWF then $$\G; c_1{:}\lt{\sproj{G}{\pt{p}_1}}, \ldots, c_n{:}\lt{\sproj{G}{\pt{p}_n}} \vdash \ether{G}$$ is a compositional typing, for some $\G$.
\end{theorem}

\begin{proof} 
By induction on the structure of $G$, relying on  simple projection (Def.~\ref{d:sproj}).
\begin{enumerate}[$\bullet$]
\item (Case $G = \gend$): Then the thesis holds vacuously, for $\partp{G} = \emptyset$.

\item (Case $G = \gto{p_1}{p_2}\{\lb{l}_i\langle U_i\rangle.G^i\}_{i \in I}$~):
By the well-formedness assumption (Def.~\ref{d:swfltypes}), local types $\sproj{G}{\pt{p}_1}, \ldots, \sproj{G}{\pt{p}_n}$
are all defined. 
Writing $\mathtt{p}$ and $\mathtt{q}$ instead of
$ \mathtt{p}_1$ and $ \mathtt{p}_2$,
by Def.~\ref{d:sproj} we have:
\begin{eqnarray}
\sproj{G}{\pt{p}} & = & \pt{p}!\{\lb{l}_i\langle U_i\rangle.\sproj{G^i}{\pt{p}}\}_{i \in I} \label{eq:sp00000} \\
\sproj{G}{\pt{q}} & = & \pt{p}?\{\lb{l}_i\langle U_i\rangle.\sproj{G^i}{\pt{q}}\}_{i \in I} \label{eq:sp0000}\\
\sproj{G}{\pt{p}_j} & = & \sproj{G^1}{\pt{p}_j} \quad\text{for  every $j \in  \{3,\ldots, n\}$}  \label{eq:sp000}
\end{eqnarray}
In \eqref{eq:sp000}, it is useful to recall that Def.~\ref{d:sproj} decrees that, for  every $j \in  \{3,\ldots, n\}$, 
\begin{equation}
\sproj{G^l}{\pt{p}_j} = \sproj{G^k}{\pt{p}_j}, ~~\text{for any $l, k \in I$}\label{eq:sp1}
\end{equation}
which explains why taking $\sproj{G^1}{\pt{p}_j}$ is enough. 
We need to show that:
\begin{equation}
\G; c_\pt{p}{:}\lt{\sproj{G}{\pt{p}}},~ c_\pt{q}{:}\lt{\sproj{G}{\pt{q}}},~  c_{\pt{p}_3}{:}\lt{\sproj{G^1}{\pt{p}_3}},~ \ldots, 
~c_{\pt{p}_n}{:}\lt{\sproj{G^1}{\pt{p}_n}}  \vdash \ether{G} \label{eq:sp0}
\end{equation}
is a compositional typing, for some $\G$.
To improve readability, and without loss of generality,
we describe the case  $I = \{1,2\}$. 
By Def.~\ref{d:ether} we have:
$$
\ether{G} = c_\mathtt{p}\triangleright\!\begin{cases}
\lb{l}_1 : c_\mathtt{p}(u).\mysel{c_\mathtt{q}}{\lb{l}_1};\outp{c_\mathtt{q}}{v}.( \linkr{u}{v} \para \ether{G^1}) \label{eq:sp2}\\
\lb{l}_2 : c_\mathtt{p}(u).\mysel{c_\mathtt{q}}{\lb{l}_2};\outp{c_\mathtt{q}}{v}.( \linkr{u}{v} \para \ether{G^2}) 
\end{cases}
$$
and 
by combining \eqref{eq:sp00000} and \eqref{eq:sp0000} with Def.~\ref{d:loclogt}  we have:
\begin{eqnarray}
\lt{\sproj{G}{\pt{p}}}  &=&  \myoplus{\lb{l}_1: \lt{U_1} \otimes \lt{\sproj{G^1}{\pt{p}}} \,,\,\lb{l}_2: U_2 \otimes \lt{\sproj{G^2}{\pt{p}}} }{\{1,2\}} \\
\lt{\sproj{G}{\pt{q}}}  &=& \mywith{\lb{l}_1: \lt{U_1} \lolli \, \lt{\sproj{G^1}{\pt{q}}}\,,\,\lb{l}_2: U_2 \lolli \, \lt{\sproj{G^2}{\pt{q}}}}{\{1,2\}} 
\end{eqnarray}
Now, 
by assumption $G$ is SWF; by Prop.~\ref{p:gtypecons}, 
then also both its sub-types $G^1$ and $G^2$ are SWF.
Thus, using IH twice we may infer that both
\begin{eqnarray}
\G; c_\pt{p}{:}\lt{\sproj{G^1}{\pt{p}}},\, 
c_\pt{q}{:}\lt{\sproj{G^1}{\pt{q}}},\, 
\underbrace{c_{\pt{p}_3}{:}\lt{\sproj{G^1}{\pt{p}_3}},\, \ldots, 
c_{\pt{p}_n}{:}\lt{\sproj{G^1}{\pt{p}_n}}}_{\Delta_1} \!\vdash \!\ether{G^1} \label{eq:sp2}\\
\G; c_\pt{p}{:}\lt{\sproj{G^2}{\pt{p}}},\, 
c_\pt{q}{:}\lt{\sproj{G^2}{\pt{q}}},\, 
\underbrace{c_{\pt{p}_3}{:}\lt{\sproj{G^2}{\pt{p}_3}},\, \ldots, 
c_{\pt{p}_n}{:}\lt{\sproj{G^2}{\pt{p}_n}}}_{\D_2} \!\vdash \!\ether{G^2}\label{eq:sp3}
\end{eqnarray}
are compositional typings for some $\G$. Now, using \eqref{eq:sp1} we infer  
$\sproj{G^1}{\pt{p}_j} = \sproj{G^2}{\pt{p}_j}$, for all $j \in \{3,\ldots, n\}$. Therefore,
\begin{equation}
\D_1 = \D_2 \label{eq:sp4}.
\end{equation}
Using \eqref{eq:sp2}, the derivation for
\begin{align}
\nnum \G; c_\mathtt{p}{:}\lt{U_1} \otimes \lt{\sproj{G^1}{\pt{p}}}, \, c_\mathtt{q}{:}\mywith{\lb{l}_1 :\lt{U_1} \lolli \lt{\sproj{G^1}{\pt{q}}}, \lb{l}_2 :\lt{U_2} \lolli \lt{\sproj{G^2}{\pt{q}}}}{\{1,2\}}  , \D_1 \vdash \qquad  \\
  \qquad    c_\mathtt{p}(u).\mysel{c_\mathtt{q}}{\lb{l}_1};\outp{c_\mathtt{q}}{v}.(\linkr{u}{v} \para  \ether{G^1}) \label{eq:sp5}
      \end{align}
      is as in Fig.~\ref{fig:ltypesmed}. Using \eqref{eq:sp3}, the derivation for 
\begin{align}
\nnum \G;  c_\mathtt{p}{:}\lt{U_2} \otimes \lt{\sproj{G^2}{\pt{p}}}, \,c_\mathtt{q}{:}\mywith{\lb{l}_1 :\lt{U_1} \lolli \lt{\sproj{G^1}{\pt{q}}}, \lb{l}_2 :\lt{U_2} \lolli \lt{\sproj{G^2}{\pt{q}}}}{\{1,2\}}, \D_2 \vdash \qquad  \\
      c_\mathtt{p}(u).\mysel{c_\mathtt{q}}{\lb{l}_2};\outp{c_\mathtt{q}}{v}.(\linkr{u}{v} \para  \ether{G^2}) \label{eq:sp6}
\end{align}
is obtained in an analogous way. 
We now have all requirements for completing the desired typing.
Using rule \name{T$\lft\oplus$} (cf. Fig.~\ref{fig:type-systemii}) 
with \eqref{eq:sp5} and \eqref{eq:sp6} as premises, and crucially relying on \eqref{eq:sp4}, we may derive:
\begin{align*}
\nnum \G; c_\mathtt{p}{:} \myoplus{\lb{l}_1: \lt{U_1} \otimes \lt{\sproj{G^1}{\pt{p}}} \,,\,\lb{l}_2: \lt{U_2} \otimes \lt{\sproj{G^2}{\pt{p}}} }{ \{1,2\}}, \qquad \qquad \qquad \\
c_\mathtt{q}{:}\mywith{\lb{l}_1 :\lt{U_1} \lolli \lt{\sproj{G^1}{\pt{q}}}, \lb{l}_2 :\lt{U_2} \lolli \lt{\sproj{G^2}{\pt{q}}}}{\{1,2\}}, \D_1 \vdash \ether{G} \label{eq:sp6}
\end{align*}
which is easily seen to be a compositional typing.

\begin{figure}[t!]
{\footnotesize
\[
{
\infer[\name{T${\lft\otimes}$}]
{ \G; c_\mathtt{p} : \lt{U_1} \otimes \lt{\sproj{G^1}{\pt{p}}}, \, c_\mathtt{q}{:}\mywith{\lb{l}_i :(\lt{U_i} \lolli \lt{\sproj{G^i}{\pt{q}}})}{i \in I}  , \D_1 \vdash 
      c_\mathtt{p}(u).\mysel{c_\mathtt{q}}{\lb{l}_1};\outp{c_\mathtt{q}}{v}.(\linkr{u}{v} \para  \ether{G^1} ) }{
      \infer=[\name{T$\lft\with_2$}]{\G; u: \lt{U_1},\,  c_\mathtt{p} :  \lt{\sproj{G^1}{\pt{p}}},\, c_\mathtt{q}{:}\mywith{\lb{l}_i :(\lt{U_i} \lolli \lt{\sproj{G^i}{\pt{q}}})}{i \in I} , \D_1 \vdash \mysel{c_\mathtt{q}}{\lb{l}_1};\outp{c_\mathtt{q}}{v}.( \linkr{u}{v} \para  \ether{G^1} ) \mathstrut}{
\infer[\name{T$\lft\with_1$}]
      {  \G; u: \lt{U_1},\,  c_\mathtt{p} :  \lt{\sproj{G^1}{\pt{p}}},\, c_\mathtt{q}{:}\mywith{\lb{l}_1 :\lt{U_1} \lolli \lt{\sproj{G^1}{\pt{q}}}}{\{1\}} , \D_1 \vdash \mysel{c_\mathtt{q}}{\lb{l}_1};\outp{c_\mathtt{q}}{v}.( \linkr{u}{v} \para  \ether{G^1} )}
      {\infer[\name{T${\lft\lolli}$}]
             { \G; u: \lt{U_1},\,  c_\mathtt{p} :  \lt{\sproj{G^1}{\pt{p}}},\, c_\mathtt{q}{:}\lt{U_1} \lolli \lt{\sproj{G^1}{\pt{q}}} , \D_1 \vdash \outp{c_\mathtt{q}}{v}.( \linkr{u}{v} \para  \ether{G^1} )}
             {\infer[\name{T$\mathsf{id}$}]{\G; u{:}\lt{U_1} \vdash \linkr{u}{v} :: v : \lt{U_1}}{} &
              \infer[\name{T${\lft\one}$}]{ \G; c_\pt{p}{:}\lt{\sproj{G^1}{\pt{p}}},\, c_\pt{q}{:}\lt{\sproj{G^1}{\pt{q}}},\,  \D_1 \vdash  \ether{G^1}}{
            }}}}}
              }
\]
}
\caption{Derivation for $c_\mathtt{p}(u).\mysel{c_\mathtt{q}}{\lb{l}_1};\outp{c_\mathtt{q}}{v}.(\linkr{u}{v} \para  \ether{G^1})$  (cf. \eqref{eq:sp5} in Page~\pageref{eq:sp5})\label{fig:ltypesmed}}
\end{figure}

\item (Case $G = G_1 \para G_2$): 
By Def.~\ref{d:gltypes}, 
we know that $\partp{G} = \partp{G_1} \cup \partp{G_2}$.
Let us write $\partp{G_1} = \{\pt{p}_1, \ldots, \pt{p}_k\}$
and 
$\partp{G_2} = \{\pt{p}_{k+1}, \ldots, \pt{p}_n\}$.
By assumption $G$ is SWF; by Prop.~\ref{p:gtypecons}, then also $G_1$ and $G_2$ are SWF.
Thus, using IH twice we may infer that both
\begin{eqnarray}
\G; \underbrace{c_1{:}\lt{\sproj{G_1}{\pt{p}_1}}, \ldots, c_k{:}\lt{\sproj{G_1}{\pt{p}_k}}}_{\D_1} \vdash \ether{G_1} \\ \label{eq:sp7}
\G; \underbrace{c_{k+1}{:}\lt{\sproj{G_2}{\pt{p}_{k+1}}}, \ldots, c_n{:}\lt{\sproj{G_2}{\pt{p}_n}}}_{\D_2} \vdash \ether{G_2} \label{eq:sp8}
\end{eqnarray}
are compositional typings for some $\G$.
Now, 
by Def.~\ref{d:sproj}, we infer that for all $\pt{p}_i \in \partp{G}$
then either $\sproj{G_1}{\pt{p}_i}$ is defined 
or  $\sproj{G_2}{\pt{p}_i}$ is defined, but not both.
We therefore infer that $\D_1 \# \D_2$. Then, 
using independent parallel composition (cf. rule~\name{\textsc{indComp}} in \S\,\ref{app:indc}) we may infer the typing:
$$
\G;  \D_1, \D_2 \vdash \ether{G_1} \para  \ether{G_2} 
$$
Hence, since by Def.~\ref{d:ether} $\ether{G} = \ether{G_1} \para  \ether{G_2} $, the thesis follows.
\end{enumerate}

\end{proof}

\subsubsection{Compositional Typings Induce SWF Global Types \label{app:s-medltypes}}
We state and prove the analogous of 
Theorem~\ref{l:medltypes} for SWF global types.
\begin{theorem}\label{l:s-medltypes}
Let $G \in \fgtypes$ be a global type. If 
$$\G; c_1{:}A_1, \ldots, c_ n{:}A_ n \vdash \ether{G}$$
is a compositional typing for \ether{G}
then 
$\exists T_1, \ldots, T_ n$ 
s.t. $\sproj{G}{\mathtt{r}_j} \subt T_j$ and 
$\lt{T_j} = A_j$, for all $\pt{r}_j \in G$.
\end{theorem}

\begin{proof} 
By induction on the structure of $G$:
\begin{enumerate}[$\bullet$]
\item (Case $G = \gend$): Then $\ether{G} = \zero$, $\partp{G} = \emptyset$, and the thesis follows vacuously.
Notice that from the assumption $\G; \cdot \vdash \ether{G} $ and rule \name{T$\lft\one$}
 we may derive 
$\G; c_j{:}\one \vdash \ether{G}$, for any name $c_j$. 
In such a case, we observe that Def.~\ref{d:sproj} decrees that 
$\sproj{\gend}{\mathtt{r}_j} = \lend$, for any $\mathtt{r}_j$.
The thesis holds, for $\lt{\lend} = \one$.


\item (Case $G = \gto{p_1}{p_2}\{\lb{l}_i\langle U_i\rangle.G^i\}_{i \in I}$~):  
Then 
$\partp{G} = \{\mathtt{p}_1,  \ldots, \mathtt{p}_n\}$, with $n \geq 2$. 
By Def.~\ref{d:sproj}, we have that
\begin{align}
\sproj{G}{\pt{p}_1} & = \pt{p}_1!\{\lb{l}_i\langle U_i \rangle. \sproj{G^i}{\pt{p}_1}\}_{i \in I} \label{eq:sp101} \\
\sproj{G}{\pt{p}_2} &  = \pt{p}_1?\{\lb{l}_i\langle U_i \rangle. \sproj{G^i}{\pt{p}_2}\}_{i \in I} \label{eq:sp102} \\
\sproj{G}{\pt{p}_j} & =  \sproj{G^1}{\pt{p}_j}, \quad \text{for all $\pt{p}_j \in \{\pt{p}_3, \ldots, \pt{p}_n\}$} \label{eq:sp103}
\end{align}
Without loss of generality,
we describe the case  $I = \{1,2\}$.
Writing $\mathtt{p}$ and $\mathtt{q}$ instead of
$ \mathtt{p}_1$ and $ \mathtt{p}_2$,
by expanding Def.~\ref{d:ether} we obtain:
\begin{equation}
{\small
\ether{G} = c_\mathtt{p}\triangleright\!\begin{cases}
\lb{l}_1 : c_\mathtt{p}(u).\mysel{c_\mathtt{q}}{\lb{l}_1};\outp{c_\mathtt{q}}{v}.( \linkr{u}{v} \para \ether{G^1}) \\
\lb{l}_2 : c_\mathtt{p}(u).\mysel{c_\mathtt{q}}{\lb{l}_2};\outp{c_\mathtt{q}}{v}.( \linkr{u}{v} \para \ether{G^2})
\end{cases} \label{eq:spnewp00}
}
\end{equation}
while by assumption we have the compositional typing
\begin{equation}
\G; c_\mathtt{p}{:}A_1, \, c_\mathtt{q}{:}A_2, \, c_{\mathtt{p}3}{:}A_3, \ldots, c_{\mathtt{p}n}{:}A_n \vdash \ether{G} \label{eq:spnewp0}
\end{equation}
We must exhibit local types $T_1, \ldots, T_n$ 
such that, for all $i \in \{1,\ldots, n\}$: \\
 (i)~$\sproj{G}{\mathtt{p}_i} \subt T_i$ and  $A_i = \lt{T_i}$.

First, by inversion on typing on \eqref{eq:spnewp00} and \eqref{eq:spnewp0}, we  infer 
that 
there exist binary session types
$C_1, C_2$, $D_1, D_2, \ldots,  D_k $, and $U_1, U_2,  \ldots, U_k$ such that
\begin{eqnarray} 
A_1  &=& \myoplus{\lb{l}_1:U_1 \otimes C_1 \, , \, \lb{l}_2:U_2 \otimes C_2}{ } \\
A_2  &=&  \mywith{\lb{l}_1:U_1 \lolli D_1  \, , \, \lb{l}_2:U_2 \lolli D_2 \, , \,  \lb{l}_3{:}U_3 \lolli D_3, \, \cdots \, , \lb{l}_k{:}U_k \lolli D_k}{} \label{eq:rspnewp0}
\end{eqnarray} 
In \eqref{eq:rspnewp0},
 $\lb{l}_3{:}U_3 \lolli D_3, \, \cdots \, , \lb{l}_k{:}U_k \lolli D_k$
correspond to labelled alternatives that may be silently added by rule~\name{T$\lft\with_2$}.
Now, using rule~\name{T$\lft\oplus$}:
\begin{small}
\[\hspace{-1ex}
\inferrule[]
{
\begin{array}{cc}
\G; c_\mathtt{p}{:}U_1 \otimes C _1, \, c_\mathtt{q}{:}A_2, 
 \D \vdash  c_\mathtt{p}(u).\mysel{c_\mathtt{q}}{\lb{l}_1};\outp{c_\mathtt{q}}{v}.(\linkr{u}{v} \para \ether{G^1} ) & (*)
\vspace{1mm}\\ 
\G; c_\mathtt{p}{:}U_2 \otimes C _2, \, c_\mathtt{q}{:}A_2,  \D \vdash  c_\mathtt{p}(u).\mysel{c_\mathtt{q}}{\lb{l}_2};\outp{c_\mathtt{q}}{v}.(\linkr{u}{v} \para \ether{G^2} ) & (**)
\end{array}
}{\G; c_\mathtt{p}{:}A_1, \, c_\mathtt{q}{:}A_2,  \D \vdash \ether{G} :: - : \one\mathstrut}
\]
\end{small}

where $\D= \, c_{\mathtt{p}3}{:}A_3, \ldots, c_{\mathtt{p}n}{:}A_n$, i.e.,
$\D$ collects typings for participants not involved in the exchange. 
Notice that the assumption of compositional typing (in particular, the fact that \ether{G} does not
offer any behavior on the right-hand side typing) is crucial in the above inversion.
By further inversion on typing, we may infer typings for $\ether{G^1}$ and $\ether{G^2}$:
\begin{align}
\G; c_\mathtt{p}{:}C_1, c_\mathtt{q}{:}D_1, \D & \vdash  \ether{G^1}   \label{eq:spnewp1} \\
\G; c_\mathtt{p}{:}C_2, c_\mathtt{q}{:}D_2, \D & \vdash  \ether{G^2}    \label{eq:spnewp2}
\end{align}
This way, e.g., the derivation for~\eqref{eq:spnewp2} is below, based on premise $(**)$ above:
where we have denoted explicitly the several possible uses of silent rule~\name{T$\lft\with_2$}.
It is easy to see that \eqref{eq:spnewp1} and \eqref{eq:spnewp2} are compositional typings.
Then, IH  ensures the existence of 
 local types $R_1,  \ldots, R_n, S_1,  \ldots, S_n$ such that: 
$$\sproj{G^1}{\mathtt{p}_j} \subt R_j \text{  and  } \sproj{G^2}{\mathtt{p}_j} \subt S_j, \quad \text{ for all $j \in \{1,\ldots,n\}$ }$$
In particular, 
$\lt{R_1} = C_1$, 
$\lt{R_2} = D_1$, 
$\lt{S_1} = C_2$, and
$\lt{S_2} = D_2$.

We notice that $\D$ is always kept unchanged in the 
typing derivations for~\eqref{eq:spnewp1} and~\eqref{eq:spnewp2}.
Therefore, for all $k \in \{3, \ldots, n\}$, we have: 
\begin{align}
A_k = \lt{R_k} = \lt{S_k} \quad \text{and} \quad
R_k =  S_k \label{eq:sp104}
\end{align}
In turn, by combining \eqref{eq:sp103} and \eqref{eq:sp104}
we infer the thesis for participants $\pt{p}_3, \ldots, \pt{p}_n$:
\begin{equation*}
\sproj{G}{\pt{p}_k} \subt R_k = T_k,  \quad \text{for all $k \in \{3, \ldots, n\}$}
\end{equation*}

We are thus left to show the thesis for $\pt{p}_1$ and $\pt{p}_2$.
We first establish $T_1$ and $T_2$ by 
building upon local types $R_1, R_2, S_1, S_2$ (just established), following
the typing derivation  for $$c_\mathtt{p}(u).\mysel{c_\mathtt{q}}{\lb{l}_2};\outp{c_\mathtt{q}}{v}.(\linkr{u}{v} \para  \ether{G^2} )$$
(shown above) and 
for $$c_\mathtt{p}(u).\mysel{c_\mathtt{q}}{\lb{l}_1};\outp{c_\mathtt{q}}{v}.(\linkr{u}{v} \para  \ether{G^1} )$$ (which is built analogously).
We thus have:
  \begin{eqnarray*}
  T_1 & = & \mathtt{p}_1!\{\lb{l}_1\langle U_1 \rangle.R_1 \, , \, \lb{l}_2\langle U_2 \rangle.S_1\}  \\
  T_2 & = & \mathtt{p}_1?\{\lb{l}_1\langle U_1 \rangle.R_2 \, , \, \lb{l}_2\langle U_2 \rangle.S_2 \, , \, \lb{l}_3\langle U_3 \rangle.S_3 \, , \, \ldots \, , \lb{l}_k\langle U_k \rangle.S_k\}  
  \end{eqnarray*}
and using \eqref{eq:sp101}, \eqref{eq:sp102}, and Definitions~\ref{d:subt} and~\ref{d:loclogt}, we may verify that: \\
(i)~$\sproj{G}{\pt{p}_1} = T_1$, $\sproj{G}{\pt{p}_2} \subt T_2$
and (ii)~$\lt{T_1} = A_1$  and $\lt{T_2} = A_2$. 

\item (Case $G = G_1 \para G_2$): 
Then $\partp{G} =  \{\mathtt{p}_1, \ldots,  \mathtt{p}_n\}$, with $n \geq 2$.
Recall that $\partp{G_1 \para G_2} =  \partp{G_1} \cup \partp{G_2}$.
By Def.~\ref{d:ether} we may state the compositional typing assumption as
$$
\G; \D \vdash \ether{G_1} \para \ether{G_2} 
$$
where $\D = c_1{:}A_1, \ldots, c_n{:}A_n$, with  a $c_j$ for each $\mathtt{r}_j \in \partp{G_1 \para G_2}$ (with $j \in\{1,\ldots,n\}$).
Let $\partp{G_1} = \{\pt{p}_1, \ldots, \pt{p}_k\}$ and $\partp{G_2} = \{\pt{p}_{k+1}, \ldots, \pt{p}_n\}$.
By Prop.~\ref{p:indpc},
there exist 
disjoint
$\D_1, \D_2$ such that $\D = \D_1 \cup \D_2$
and the compositional typings 
\begin{align*}
\G; \underbrace{c_1{:}A_1, \ldots, c_k{:}A_k}_{\D_1} \vdash & \ether{G_1}  \\
\G; \underbrace{c_{k+1}{:}A_{k+1}, \ldots, c_n{:}A_n}_{\D_2}   \vdash & \ether{G_2} 
\end{align*}
hold. 
We may then apply IH on both $\ether{G_1}$ and $\ether{G_2}$, and
so infer local types $R_1, \ldots, R_k$ and $S_{k+1}, \ldots, S_n$ such that
(i)~$\sproj{G_1}{\pt{p}_i} \subt R_i$
 and 
 (ii)~$A_i = \lt{R_i}$ (with $i \in \{1, \ldots, k\}$)
and 
(iii)~$\sproj{G_1}{\pt{p}_h} \subt S_h$
(iv)~$A_h = \lt{S_h}$ (with $h \in \{{k+1}, \ldots, n\}$).
This is enough to conclude the thesis, for Def.~\ref{d:sproj} says that parallel global types do not share participants.
Therefore, for every $\pt{p}_l \in \{\pt{p}_1, \ldots, \pt{p}_n\}$ then either
$\sproj{G}{\pt{p}_l} = \sproj{G_1}{\pt{p}_l}$ or 
$\sproj{G}{\pt{p}_l} = \sproj{G_2}{\pt{p}_l}$.
\end{enumerate}
\end{proof}

\subsection{Proof of Theorem~\ref{l:ltypesmedp}: MWF Global Types Ensure Compositional Typings\label{app:ltypesmedp}}
We first state the analogous of Prop.~\ref{p:gtypecons}:

\begin{proposition}\label{p:gtypeconsm}
If $ \gto{p_1}{p_2}\{\lb{l}_i\langle U_i\rangle.G^i\}_{i \in I}$ is 
MWF 
in \fgtypes
then
all $G^i$ ($i \in I$) are 
MWF
too. Also, if 
$G = G_1 \para G_2$ is 
MWF
then $G_j$ ($j \in \{1,2\}$) are 
MWF
too.
\end{proposition}

We now repeat the statement in Page~\pageref{l:ltypesmedp}:

\begin{theorem}[\ref{l:ltypesmedp}]
Let $G \in \fgtypes$ be a global type, with $\partp{G} = \{\pt{p}_1, \ldots, \pt{p}_n\}$. \\
If $G$ is MWF then $\G; c_1{:}\lt{\proj{G}{\pt{p}_1}}, \ldots, c_n{:}\lt{\proj{G}{\pt{p}_n}} \vdash \ether{G}$ is a compositional typing, for some $\G$.
\end{theorem}

\begin{proof} 
By induction on the structure of $G$. 
The most interesting case is when 
$G = \gto{p_1}{p_2}\{\lb{l}_i\langle U_i\rangle.G^i\}_{i \in I}$.
Remaining cases are as in the proof of Theorem~\ref{l:ltypesmed}~(Page~\pageref{l:ltypesmed}).

By the well-formedness assumption (Def.~\ref{d:wfltypes}), local types $\proj{G}{\pt{p}_1}, \ldots, \proj{G}{\pt{p}_n}$
are all defined. 
Writing $\mathtt{p}$ and $\mathtt{q}$ instead of
$ \mathtt{p}_1$ and $ \mathtt{p}_2$,
by Def.~\ref{d:proj} we have:
\begin{eqnarray}
\proj{G}{\pt{p}} & = & \pt{p}!\{\lb{l}_i\langle U_i\rangle.\proj{G^i}{\pt{p}}\}_{i \in I} \label{eq:p00000} \\
\proj{G}{\pt{q}} & = & \pt{p}?\{\lb{l}_i\langle U_i\rangle.\proj{G^i}{\pt{q}}\}_{i \in I} \label{eq:p0000}\\
\proj{G}{\pt{p}_j} & = & \fuse_{i \in I} \, \proj{G^i}{\pt{p}_j} \quad\text{for  every $j \in  \{3,\ldots, n\}$}  \label{eq:p000}
\end{eqnarray}
We need to show that, for some $\G$, 
\begin{equation}
\G; c_\pt{p}{:}\lt{\proj{G}{\pt{p}}},~ 
c_\pt{q}{:}\lt{\proj{G}{\pt{q}}},~ 
c_{\pt{p}_3}{:}\lt{\proj{G}{\pt{p}_3}},~ \ldots, 
~c_{\pt{p}_n}{:}\lt{\proj{G}{\pt{p}_n}}  \vdash \ether{G}  \label{eq:p0}
\end{equation}
is a compositional typing.
Without loss of generality, we detail the case  $I = \{1,2\}$. 
By Def.~\ref{d:ether}, we have:
$$
\ether{G} = c_\mathtt{p}\triangleright\! \begin{cases}
\lb{l}_1 : c_\mathtt{p}(u).\mysel{c_\mathtt{q}}{\lb{l}_1};\outp{c_\mathtt{q}}{v}.( \linkr{u}{v} \para \ether{G^1}) \label{eq:p2}\\
\lb{l}_2 : c_\mathtt{p}(u).\mysel{c_\mathtt{q}}{\lb{l}_2};\outp{c_\mathtt{q}}{v}.( \linkr{u}{v} \para \ether{G^2}) 
\end{cases}
$$
and 
by combining \eqref{eq:p00000} and \eqref{eq:p0000} with Def.~\ref{d:loclogt}  we have:
\begin{eqnarray*}
\lt{\proj{G}{\pt{p}}}  &=&  \myoplus{\lb{l}_1: \lt{U_1} \otimes \lt{\proj{G^1}{\pt{p}}} \,,\,\lb{l}_2:\lt{U_2} \otimes \lt{\proj{G^2}{\pt{p}}} }{i \in I} \\
\lt{\proj{G}{\pt{q}}}  &=& \mywith{\lb{l}_1: \lt{U_1} \lolli \, \lt{\proj{G^1}{\pt{q}}}\,,\,\lb{l}_2: \lt{U_2} \lolli \, \lt{\proj{G^2}{\pt{q}}}}{i \in I} 
\end{eqnarray*}
Now, 
by assumption $G$ is MWF; then, by Prop.~\ref{p:gtypeconsm}, 
both $G^1$ and $G^2$ are MWF too.
Therefore, by using IH twice we may infer that both
\begin{eqnarray}
\G; c_\pt{p}{:}\lt{\proj{G^1}{\pt{p}}},\, 
c_\pt{q}{:}\lt{\proj{G^1}{\pt{q}}},\, 
\underbrace{c_{\pt{p}_3}{:}\lt{\proj{G^1}{\pt{p}_3}},\, \ldots, 
c_{\pt{p}_n}{:}\lt{\proj{G^1}{\pt{p}_n}}}_{\Delta_1} \!\vdash \!\ether{G^1}  \label{eq:p2}\\
\G; c_\pt{p}{:}\lt{\proj{G^2}{\pt{p}}},\, 
c_\pt{q}{:}\lt{\proj{G^2}{\pt{q}}},\, 
\underbrace{c_{\pt{p}_3}{:}\lt{\proj{G^2}{\pt{p}_3}},\, \ldots, 
c_{\pt{p}_n}{:}\lt{\proj{G^2}{\pt{p}_n}}}_{\D_2} \!\vdash \!\ether{G^2}  \label{eq:p3}
\end{eqnarray}
are compositional typings, for any $\G$. 

Now, to obtain a compositional typing for \ether{G}, we must first 
address the fact that, 
differently from what occurs in the proof of Theorem~\ref{l:ltypesmed} (cf. equality \eqref{eq:sp4} in Page~\pageref{eq:sp4}), in this case
it is not necessarily the case that $\D_1$ and $\D_2$ are equal.
This discrepancy is due to the merge-based well-formedness assumption, which admits non identical behaviors in branches $G^1$ and $G^2$
in the case of (local) branching types.

We proceed by induction on $k$, defined as the size of $\D_1$ and $\D_2$ (note that $k = n -2$).
\begin{enumerate}[$-$]
\item (Case $k = 0$): Then $\D_1 = \D_2 = \emptyset$ and $\pt{p}$ and $\pt{q}$ are the only participants  in $G$.
Let us write $A_\pt{q}$ to stand for the session type 
$$\mywith{\lb{l}_1 :\lt{U_1} \lolli \lt{\proj{G^1}{\pt{q}}} \,,\, \lb{l}_2 :\lt{U_2} \lolli \lt{\proj{G^2}{\pt{q}}}}{}$$

Based on \eqref{eq:p2} and \eqref{eq:p3}, following the derivation in Fig.~\ref{fig:ltypesmed},
we may derive typings 
\begin{align}
\G; c_\mathtt{p} : \lt{U_1} \otimes \lt{\proj{G^1}{\pt{p}}}, \, c_\mathtt{q}{:}A_\pt{q} \vdash 
      c_\mathtt{p}(u).\mysel{c_\mathtt{q}}{\lb{l}_1};\outp{c_\mathtt{q}}{v}.(\linkr{u}{v} \para  \ether{G^1})  \label{eq:p4}\\
      \G; c_\mathtt{p} : \lt{U_2} \otimes \lt{\proj{G^2}{\pt{p}}}, \, c_\mathtt{q}{:}A_\pt{q} \vdash 
      c_\mathtt{p}(u).\mysel{c_\mathtt{q}}{\lb{l}_2};\outp{c_\mathtt{q}}{v}.(\linkr{u}{v} \para  \ether{G^2}) \label{eq:p5}
\end{align}
Then, using \eqref{eq:p4} and \eqref{eq:p5} we may derive, using rule~\name{T$\lft\oplus$}:
$$
 \G; c_\mathtt{p} : \mywith{\lb{l}_1 : \lt{U_1} \otimes \lt{\proj{G^1}{\pt{p}}} \,,\, \lb{l}_2 : \lt{U_2} \otimes \lt{\proj{G^2}{\pt{p}}}}{}, \, c_\mathtt{q}{:}A_\pt{q} \vdash 
      \ether{G}
$$
and the thesis follows.

\item (Case $k > 0$): Then there exists
a participant  $\pt{p}_k$, 
 types $B_1 = \lt{\proj{G^1}{\pt{p}_k}}$, $B_2  = \lt{\proj{G^2}{\pt{p}_k}}$ and environments
$\D'_1, \D'_2$ such that $\D_1 = c_{\pt{p}_k}{:}B_1, \D'_1$ and $\D_2 = c_{\pt{p}_k}{:}B_2, \D'_2$.

By induction hypothesis, 
there is a compositional typing starting from 
\begin{align*}
\G; c_\pt{p}{:}\lt{\proj{G^1}{\pt{p}}},\, 
c_\pt{q}{:}\lt{\proj{G^1}{\pt{q}}},\, {\D'_1} & \vdash \!\ether{G^1}   \\
\G; c_\pt{p}{:}\lt{\proj{G^2}{\pt{p}}},\, 
c_\pt{q}{:}\lt{\proj{G^2}{\pt{q}}},\, {\D'_2} & \vdash \!\ether{G^2}  
\end{align*}
 resulting into 
$$
 \G; c_\mathtt{p} : \mywith{\lb{l}_1 : \lt{U_1} \otimes \lt{\proj{G^1}{\pt{p}}} \,,\, \lb{l}_2 : \lt{U_2} \otimes \lt{\proj{G^2}{\pt{p}}}}{}, \, c_\mathtt{q}{:}A_\pt{q}, \D'_1 \vdash 
      \ether{G}
$$
since $\D'_1 = \D'_2$.
To extend the typing derivation to $\D_1$ and $\D_2$, 
we proceed by a case analysis on the shape of $B_1$ and $B_2$.
 We aim to show that (a) $B_1$ and $B_2$ are already identical session types or (b) that 
typing allows us to transform them into identical types. 
We rely  on the definition of $\fuse$ (Def.~\ref{d:mymerg}). There are three cases:
\begin{enumerate}[(1)]
\item Case $B_1 = \one$: 
Then, since Def.~\ref{d:mymerg} decrees $\lend \fuse \lend = \lend$
and the fact that merge-based well-definedness depends on $\fuse$, 
we may infer 
$B_2 = \one$.
Hence, $\D_1 = \D_2$ and the
desired derivation is obtained as in the base case.

\item Case $B_1 = \myoplus{\lb{l}_k: A_k}{k \in K}$: Then, similarly as in the previous sub case, 
by Def.~\ref{d:mymerg} we immediately infer that $B_2 = \myoplus{\lb{l}_k: A_k}{k \in K}$. We this may infer that 
 $\D_1 = \D_2$ and complete the derivation.

\item Case $B_1 = \mywith{\lb{l}_h: A_h}{h \in H}$: This is the interesting case, for 
even if 
merge-based well-formedness of $G$
ensures that both 
$B_1$ and $B_2$
are both  selection types, 
it is not necessarily the case
that $B_1$ and $B_2$ are identical.

If $B_1$ and $B_2$ are identical then we proceed as in previous sub cases.
Otherwise, then 
due to $\fuse$
there are some labeled alternatives in $B_1$ but not in $B_2$ and/or viceversa.
Also, Def.~\ref{d:mymerg} ensures that common options (if any) are identical in both branches.
We may then use the rule \name{T$\lft\with_2$} to ``complement'' occurrences of types 
$B_1$ and $B_2$ in \eqref{eq:p2} and \eqref{eq:p3}
with appropriate options, 
so as to make them coincide and achieve identical typing. 
This rule is silent; as labels are finite, this completing task is also finite,  and results into $\D_1 = \D_2$. 
We then may complete the derivation as in the base case.
This concludes the proof.
\end{enumerate}
\end{enumerate}
\end{proof}

\subsection{Proof of Theorem~\ref{l:medltypes}: Compositional Typings Induce MWF Global Types \label{app:medltypes}}

We repeat the statement given in Page \pageref{l:medltypes}:
\begin{theorem}[\ref{l:medltypes}]
Let $G \in \fgtypes$ be a 
global type. 
If 
$$\G; c_1{:}A_1, \ldots, c_ n{:}A_ n \vdash \ether{G}$$
is a compositional typing for \ether{G}
then 
$\exists 
T_1, \ldots, T_ n$ 
s.t. 
$\proj{G}{\mathtt{r}_j} \subt T_j$ and 
$\lt{T_j} = A_j$, 
for all $\pt{r}_j \in G$.
\end{theorem}

\begin{proof}
By induction on the structure of $G$:
\begin{enumerate}[$\bullet$]
\item (Case $G = \gend$): Then $\ether{G} = \zero$, $\partp{G} = \emptyset$, and the thesis follows vacuously.
Notice that from the assumption $\G; \cdot \vdash \ether{G} $ and rule \name{T$\lft\one$}
 we may derive 
$\G; c_j{:}\one \vdash \ether{G}$, for any name $c_j$. 
In such a case, we observe that Def.~\ref{d:proj} decrees that 
$\proj{\gend}{\mathtt{r}_j} = \lend$, for any $\mathtt{r}_j$.
The thesis holds, for $\lt{\lend} = \one$.


\item (Case $G = \gto{p_1}{p_2}\{\lb{l}_i\langle U_i\rangle.G^i\}_{i \in I}$~):  
Then 
$\partp{G} = \{\mathtt{p}_1,  \ldots, \mathtt{p}_n\}$, with $n \geq 2$. 
By Def.~\ref{d:proj}, we have that
\begin{align}
\proj{G}{\pt{p}_1} & = \pt{p}_1!\{\lb{l}_i\langle U_i \rangle. \proj{G^i}{\pt{p}_1}\}_{i \in I} \label{eq:p101} \\
\proj{G}{\pt{p}_2} &  = \pt{p}_1?\{\lb{l}_i\langle U_i \rangle. \proj{G^i}{\pt{p}_2}\}_{i \in I} \label{eq:p102} \\
\proj{G}{\pt{p}_j} & =  \fuse_{i \in I} \, \proj{G^i}{\pt{p}_j} \label{eq:p103}
\end{align}
Without loss of generality,
we detail the case  $|I| = 2$.
Writing $\mathtt{p}$ and $\mathtt{q}$ instead of
$ \mathtt{p}_1$ and $ \mathtt{p}_2$,
by expanding Def.~\ref{d:ether} we obtain:
\begin{equation}
{\small
\ether{G} = c_\mathtt{p}\triangleright\!\begin{cases}
\lb{l}_1 : c_\mathtt{p}(u).\mysel{c_\mathtt{q}}{\lb{l}_1};\outp{c_\mathtt{q}}{v}.( \linkr{u}{v} \para \ether{G^1}) \\
\lb{l}_2 : c_\mathtt{p}(u).\mysel{c_\mathtt{q}}{\lb{l}_2};\outp{c_\mathtt{q}}{v}.( \linkr{u}{v} \para \ether{G^2})
\end{cases} \label{eq:pnewp0s0}
}
\end{equation}
while by assumption we have the compositional typing
\begin{equation}
\G; c_\mathtt{p}{:}A_1, \, c_\mathtt{q}{:}A_2, \, c_{\mathtt{p}3}{:}A_3, \ldots, c_{\mathtt{p}n}{:}A_n \vdash \ether{G} \label{eq:pnewp0}
\end{equation}
We must exhibit local types $T_1, \ldots, T_n$ 
such that, for all $i \in \{1,\ldots, n\}$: \\
 (i)~$\proj{G}{\mathtt{p}_i} \subt T_i$ and  $A_i = \lt{T_i}$.

First, by inversion on typing on \eqref{eq:pnewp0s0} and \eqref{eq:pnewp0}, we  infer 
that 
there exist binary session types
$C_1, C_2$, $D_1, D_2, \ldots, D_k$, and $U_1, U_2, \ldots, U_k$ such that
\begin{eqnarray} 
A_1 & = & \myoplus{\lb{l}_1:U_1 \otimes C_1 \, , \, \lb{l}_2:U_2 \otimes C_2}{ } \\
A_2 & = & \mywith{\lb{l}_1:U_1 \lolli D_1 \, , \, \lb{l}_2:U_2 \lolli D_2 \, , \, \cdots \, , \lb{l}_k:U_k \lolli D_k}{}  \label{eq:sspnewp0}
\end{eqnarray} 
Notice that in \eqref{eq:sspnewp0}, we consider labeled alternatives that may be silently added by rule~\name{T$\lft\with_2$}.
Now, using rule~\name{T$\lft\oplus$}:
\begin{small}
\[\hspace{-1ex}
\inferrule[]
{
\begin{array}{cc}
\G; c_\mathtt{p}{:}U_1 \otimes C _1, \, c_\mathtt{q}{:}A_2, 
 \D \vdash  c_\mathtt{p}(u).\mysel{c_\mathtt{q}}{\lb{l}_1};\outp{c_\mathtt{q}}{v}.(\linkr{u}{v} \para \ether{G^1} ) & (*)
\vspace{1mm}\\ 
\G; c_\mathtt{p}{:}U_2 \otimes C _2, \, c_\mathtt{q}{:}A_2,  \D \vdash  c_\mathtt{p}(u).\mysel{c_\mathtt{q}}{\lb{l}_2};\outp{c_\mathtt{q}}{v}.(\linkr{u}{v} \para \ether{G^2} ) & (**)
\end{array}
}{\G; c_\mathtt{p}{:}A_1, \, c_\mathtt{q}{:}A_2,  \D \vdash \ether{G} :: - : \one\mathstrut}
\]
\end{small}

where $\D= \, c_{\mathtt{p}3}{:}A_3, \ldots, c_{\mathtt{p}n}{:}A_n$, i.e.,
$\D$ collects typings for participants not involved in the exchange. 
Notice that the assumption of compositional typing (in particular, the fact that \ether{G} does not
offer any behavior on the right-hand side typing) is crucial in the above inversion.
By further inversion on typing, we may infer typings for $\ether{G^1}$ and $\ether{G^2}$:
\begin{align}
\G; c_\mathtt{p}{:}C_1, c_\mathtt{q}{:}D_1, \D & \vdash  \ether{G^1}   \label{eq:pnewp1} \\
\G; c_\mathtt{p}{:}C_2, c_\mathtt{q}{:}D_2, \D & \vdash  \ether{G^2}    \label{eq:pnewp2}
\end{align}

This way, e.g., the derivation for~\eqref{eq:pnewp2} is below, based on premise $(**)$ above:
\begin{small}
\[
{
\infer[\name{T$\lft\otimes$}]
{\G;  c_\mathtt{p} : U_2 \otimes C_2, \, c_\mathtt{q}{:}A_2  , \D \vdash 
      c_\mathtt{p}(u).\mysel{c_\mathtt{q}}{\lb{l}_2};\outp{c_\mathtt{q}}{v}.(\linkr{u}{v} \para  \ether{G^2} )  }{
      \infer=[\name{T$\lft\with_2$}]{\G; u: U_2,\,  c_\mathtt{p} :  C_2,\, c_\mathtt{q}{:}A_2 , \D \vdash \mysel{c_\mathtt{q}}{\lb{l}_2};\outp{c_\mathtt{q}}{v}.( \linkr{u}{v} \para  \ether{G^2} )}{
\infer[\name{T$\lft\with_1$}]
      { \G; u: U_2,\,  c_\mathtt{p} :  C_2,\, c_\mathtt{q}{:}\mywith{\lb{l}_2 :U_2 \lolli D_2}{\{2\}} , \D \vdash \mysel{c_\mathtt{q}}{\lb{l}_2};\outp{c_\mathtt{q}}{v}.( \linkr{u}{v} \para  \ether{G^2} ) }
      {\infer[\name{T$\lft{\lolli}$}]
             {\G; u: U_2,\,  c_\mathtt{p} :  C_2,\, c_\mathtt{q}{:}U_2 \lolli D_2 , \D \vdash \outp{c_\mathtt{q}}{v}.( \linkr{u}{v} \para  \ether{G^2} )  \mathstrut}
             {\infer[\name{T$\mathsf{id}$}]{\G; u{:}U_2 \vdash \linkr{u}{v} :: v : U_2 }{} &
              \infer[\name{T$\lft\one$}]{ \G; c_\mathtt{p} : C_2, c_\mathtt{q} : D_2  , \D \vdash  \ether{G^2}\mathstrut}{
              }}}}}
              }
\]
\end{small}
where we have explicitly denoted the several possible uses of silent rule~\name{T$\lft\with_2$}.
It is easy to see that \eqref{eq:pnewp1} and \eqref{eq:pnewp2} are compositional typings.
Then, IH  ensures the existence of 
 local types $R_1,  \ldots, R_n, S_1,  \ldots, S_n$ such that: 
$$\proj{G^1}{\mathtt{p}_j} \subt R_j \text{  and  } \proj{G^2}{\mathtt{p}_j} \subt S_j, \quad \text{ for all $j \in \{1,\ldots,n\}$ }$$
In particular, 
$\lt{R_1} = C_1$, 
$\lt{R_2} = D_1$, 
$\lt{S_1} = C_2$, and
$\lt{S_2} = D_2$.

We notice that 
$\D$ remains unchanged in the derivations for~\eqref{eq:pnewp1} and~\eqref{eq:pnewp2}.
Hence, intuitively, 
considering merge-based projectability does not play a role in the proof, for our assumption is the compositional typing for
\ether{G}. All mergeable branching types for branches of $G$ appear already merged in $\D$; such merged types are propagated in derivations.
Therefore, for all $k \in \{3, \ldots, n\}$, we have: 
\begin{align}
A_k = \lt{R_k} = \lt{S_k} \quad \text{and} \quad
R_k =  S_k \label{eq:p104}
\end{align}
In turn, by combining \eqref{eq:p103} and \eqref{eq:p104}
we infer the thesis for participants $\pt{p}_3, \ldots, \pt{p}_n$:
\begin{equation*}
\proj{G}{\pt{p}_k} \subt R_k = T_k,  \quad \text{for all $k \in \{3, \ldots, n\}$}
\end{equation*}

We are thus left to show the thesis for $\pt{p}_1$ and $\pt{p}_2$.
We first establish $T_1$ and $T_2$ by 
building upon local types $R_1, R_2, S_1, S_2$ (just established), following
the typing derivation  for $$c_\mathtt{p}(u).\mysel{c_\mathtt{q}}{\lb{l}_2};\outp{c_\mathtt{q}}{v}.(\linkr{u}{v} \para  \ether{G^2} )$$
(shown above) and 
for $$c_\mathtt{p}(u).\mysel{c_\mathtt{q}}{\lb{l}_1};\outp{c_\mathtt{q}}{v}.(\linkr{u}{v} \para  \ether{G^1} )$$ (which is built analogously).
We thus have:
  \begin{eqnarray*}
  T_1 & = & \mathtt{p}_1!\{\lb{l}_1\langle U_1 \rangle.R_1 \, , \, \lb{l}_2\langle U_2 \rangle.S_1\}  \\
  T_2 & = & \mathtt{p}_1?\{\lb{l}_1\langle U_1 \rangle.R_2 \, , \, \lb{l}_2\langle U_2 \rangle.S_2 \, , \, \cdots \, , \, \lb{l}_k\langle U_k \rangle.S_k\}  
  \end{eqnarray*}
and using \eqref{eq:p101}, \eqref{eq:p102}, and Definitions~\ref{d:subt} and \ref{d:loclogt}, we may verify that: \\
(i)~$\proj{G}{\pt{p}_1} = T_1$, $\proj{G}{\pt{p}_2} \subt T_2$
and (ii)~$\lt{T_1} = A_1$  and $\lt{T_2} = A_2$. 

\item (Case $G = G_1 \para G_2$): 
Then $\partp{G} =  \{\mathtt{p}_1, \ldots,  \mathtt{p}_n\}$, with $n \geq 2$.
Recall that $\partp{G_1 \para G_2} =  \partp{G_1} \cup \partp{G_2}$.
By Def.~\ref{d:ether} we may state the compositional typing assumption as
$$
\G; \D \vdash \ether{G_1} \para \ether{G_2} 
$$
where $\D = c_1{:}A_1, \ldots, c_n{:}A_n$, with  a $c_j$ for each $\mathtt{r}_j \in \partp{G_1 \para G_2}$ (with $j \in\{1,\ldots,n\}$).
Let $\partp{G_1} = \{\pt{p}_1, \ldots, \pt{p}_k\}$ and $\partp{G_2} = \{\pt{p}_{k+1}, \ldots, \pt{p}_n\}$.
By Prop.~\ref{p:indpc},
there exist 
disjoint
$\D_1, \D_2$ such that $\D = \D_1 \cup \D_2$
and the compositional typings 
\begin{align*}
\G; \underbrace{c_1{:}A_1, \ldots, c_k{:}A_k}_{\D_1} \vdash & \ether{G_1}  \\
\G; \underbrace{c_{k+1}{:}A_{k+1}, \ldots, c_n{:}A_n}_{\D_2}   \vdash & \ether{G_2} 
\end{align*}
hold. 
We may then apply IH on both $\ether{G_1}$ and $\ether{G_2}$, and
so infer local types $R_1, \ldots, R_k$ and $S_{k+1}, \ldots, S_n$ such that
(i)~$\proj{G_1}{\pt{p}_i} \subt R_i$
 and 
 (ii)~$A_i = \lt{R_i}$ (with $i \in \{1, \ldots, k\}$)
and 
(iii)~$\proj{G_1}{\pt{p}_h} \subt S_h$
(iv)~$A_h = \lt{S_h}$ (with $h \in \{{k+1}, \ldots, n\}$).
This is enough to conclude the thesis, for Def.~\ref{d:proj} says that parallel global types do not share participants.
Therefore, for every $\pt{p}_l \in \{\pt{p}_1, \ldots, \pt{p}_n\}$ then either
$\proj{G}{\pt{p}_l} = \proj{G_1}{\pt{p}_l}$ or 
$\proj{G}{\pt{p}_l} = \proj{G_2}{\pt{p}_l}$.
\end{enumerate}
\end{proof}

\subsection{Proof of Theorem~\ref{p:swapmeds}\label{app:swapmeds}: Behavioral Characterization of Swapping}
We repeat the statement given in Page~\pageref{p:swapmeds}:

\begin{theorem}[\ref{p:swapmeds}]
Let $G_1 \in \fgtypes$ be a global type, such that $\ether{G_1}$ has 
a compositional typing $\G; \D \vdash \ether{G_1} $, for some $\G, \D$. 
If $G_1 \eqsw G_2$ then  $\G;  \D \vdash \ether{G_1} \tybis \ether{G_2}  $.
\end{theorem}

\begin{proof}
We first prove the following property:
\begin{equation}
\text{If $G_1 \eqsw G_2$ then $\G; \D \vdash \ether{G_1} \eqcc \ether{G_2}$}  \label{pr:swp}
\end{equation}
Recall that we have defined $\eqcc$ in Def.~\ref{d:eqcc}.
The thesis will then follow by combining~\eqref{pr:swp} with Theorem~\ref{t:soundcc}~(soundness of \eqcc wrt \tybis).
The proof of~\eqref{pr:swp} proceeds by induction on the definition of $\eqsw$ (cf. Def.~\ref{d:gswap} and Fig.~\ref{fig:swap}). 
We consider only the cases for \name{\textsc{sw1}} and \name{\textsc{sw3}}:
\begin{enumerate}[$\bullet$]
\item (Case \name{\textsc{sw3}}): Then we have:
\begin{align*}
G_1 & = \gto{p}{q}\{\lb{l}_i\langle U_i\rangle.(G^i \para G)\}_{i \in I} \\
G_2 & = G \para \gto{p}{q}\big\{\lb{l}_i\langle U_i\rangle.G^i\}_{i \in I}
\end{align*}
with 
\begin{equation}
\{\mathtt{p},\mathtt{q}\} \# \{\mathtt{r}_1, \ldots, \mathtt{r}_n\}. \label{eq:swp2}
\end{equation}
where $\{\mathtt{r}_1, \ldots, \mathtt{r}_n\} = \partp{G}$.
Without loss of generality, we consider the case $I = \{1,2\}$.
Then, by Def.~\ref{d:ether}, we have:
\begin{align*}
\ether{G_1} & =  c_\mathtt{p}\triangleright\!\begin{cases}
\lb{l}_1 : c_\mathtt{p}(u).\mysel{c_\mathtt{q}}{\lb{l}_1};\outp{c_\mathtt{q}}{v}.( \linkr{u}{v}  \para \ether{G^1} \para \ether{G}) \\
\lb{l}_2 : c_\mathtt{p}(u).\mysel{c_\mathtt{q}}{\lb{l}_2};\outp{c_\mathtt{q}}{v}.( \linkr{u}{v} \para \ether{G^2} \para \ether{G})
\end{cases}
\end{align*}
Moreover, since by assumption \ether{G_1} has a compositional typing, by Theorem~\ref{l:medltypes} we have:
\begin{equation*}
\G; c_\pt{p}{:}\lt{\proj{G_1}{\pt{p}}} , c_\pt{q}{:}\lt{\proj{G_1}{\pt{q}}} , 
c_{\pt{r}_1}{:}\lt{\proj{G_1}{\pt{r}_1}} , \ldots, c_{\pt{r}_n}{:}\lt{\proj{G_1}{\pt{r}_n}} \vdash \ether{G_1} 
\end{equation*}
Observe that in \ether{G_1} we have the name distinctions: $\{c_\pt{p},c_\pt{q}\} \# \{c_{\pt{r}_1},\ldots, c_{\pt{r}_n}\}$, thanks to \eqref{eq:swp2}. 
These distinctions will be useful in commuting prefixes inside \ether{G_1} in a type-preserving way, for they ensure that 
all these sessions are causally independent. 
We may now use $\eqcc$ to perform the prefix commutations, using the equalities in 
Figs.~\ref{f:cc1} and~\ref{f:cc2}:
\begin{align*}
\ether{G_1} & =  c_\mathtt{p}\triangleright\!\begin{cases}
\lb{l}_1 : c_\mathtt{p}(u).\mysel{c_\mathtt{q}}{\lb{l}_1};\outp{c_\mathtt{q}}{v}.( \linkr{u}{v}  \para \ether{G^1} \para \ether{G}) \\
\lb{l}_2 : c_\mathtt{p}(u).\mysel{c_\mathtt{q}}{\lb{l}_2};\outp{c_\mathtt{q}}{v}.( \linkr{u}{v} \para \ether{G^2} \para \ether{G})
\end{cases} \\
& \eqcc  c_\mathtt{p}\triangleright\!\begin{cases}
\lb{l}_1 : c_\mathtt{p}(u).\mysel{c_\mathtt{q}}{\lb{l}_1};(\ether{G} \para \outp{c_\mathtt{q}}{v}.( \linkr{u}{v}  \para \ether{G^1})) \\
\lb{l}_2 : c_\mathtt{p}(u).\mysel{c_\mathtt{q}}{\lb{l}_2};(\ether{G} \para \outp{c_\mathtt{q}}{v}.( \linkr{u}{v} \para \ether{G^2})) 
\end{cases} \\
& \eqcc  c_\mathtt{p}\triangleright\!\begin{cases}
\lb{l}_1 : c_\mathtt{p}(u).(\ether{G} \para \mysel{c_\mathtt{q}}{\lb{l}_1};\outp{c_\mathtt{q}}{v}.( \linkr{u}{v}  \para \ether{G^1})) \\
\lb{l}_2 : c_\mathtt{p}(u).(\ether{G} \para \mysel{c_\mathtt{q}}{\lb{l}_2};\outp{c_\mathtt{q}}{v}.( \linkr{u}{v} \para \ether{G^2})) 
\end{cases} \\
& \eqcc  c_\mathtt{p}\triangleright\!\begin{cases}
\lb{l}_1 : \ether{G} \para c_\mathtt{p}(u).\mysel{c_\mathtt{q}}{\lb{l}_1};\outp{c_\mathtt{q}}{v}.( \linkr{u}{v}  \para \ether{G^1}) \\
\lb{l}_2 : \ether{G} \para c_\mathtt{p}(u). \mysel{c_\mathtt{q}}{\lb{l}_2};\outp{c_\mathtt{q}}{v}.( \linkr{u}{v} \para \ether{G^2}) 
\end{cases} \\
& \eqcc  \ether{G} \para c_\mathtt{p}\triangleright\!\begin{cases}
\lb{l}_1 :  c_\mathtt{p}(u).\mysel{c_\mathtt{q}}{\lb{l}_1};\outp{c_\mathtt{q}}{v}.( \linkr{u}{v}  \para \ether{G^1}) \\
\lb{l}_2 :  c_\mathtt{p}(u). \mysel{c_\mathtt{q}}{\lb{l}_2};\outp{c_\mathtt{q}}{v}.( \linkr{u}{v} \para \ether{G^2}) 
\end{cases} \\
& = \ether{G_2}
\end{align*}

\item (Case \name{\textsc{sw1}}): Then we have:
\begin{align*}
G_1 & = \gto{p_1}{q_1}\big\{\lb{l}_i\langle U_i\rangle.\gto{p_2}{q_2}\{\lb{l}'_j\langle U'_j\rangle.G_{ij}\}_{j \in J}\big\}_{i \in I} \\
G_2 & = \gto{p_2}{q_2}\big\{\lb{l}'_j\langle U'_j\rangle.\gto{p_1}{q_1}\{\lb{l}_i\langle U_i\rangle.G_{ij}\}_{i \in I}\big\}_{j \in J}
\end{align*}
with $\{\mathtt{p_1},\mathtt{q_1}\} \# \{\mathtt{p_2},\mathtt{q_2}\}$.
Without loss of generality, 
we consider the following instance:
\begin{align*}
G_1 & = \gto{p}{q}\big\{\lb{l}\langle U_1\rangle.\gto{r}{s}\{\lb{h}\langle U_2\rangle.G^{}\}_{}\big\}_{} \\
G_2 & = \gto{r}{s}\big\{\lb{h}\langle U_2\rangle.\gto{p}{q}\{\lb{l}\langle U_1\rangle.G^{}\}_{}\big\}_{} 
\end{align*}
with
\begin{equation}
\{\pt{p},\pt{q}\} \# \{\pt{r},\pt{s}\}. \label{eq:swp1}
\end{equation}
Then, by Def.~\ref{d:ether} on $G_1$ we have:
\begin{align*}
\ether{G_1} & =  \mycasebig{c_\mathtt{p}}{\lb{l}: c_\mathtt{p}(u_1).\mysel{c_\mathtt{q}}{\lb{l}};\outp{c_\mathtt{q}}{v_1}.\big( \linkr{u_1}{v_1}  \para \\
& \qquad \qquad \qquad \qquad \qquad \qquad \mycase{c_\mathtt{r}}{\lb{h}: c_\mathtt{r}(u_2).\mysel{c_\mathtt{s}}{\lb{h}};\outp{c_\mathtt{s}}{v_2}.( \linkr{u_2}{v_2}  \para \ether{G} )}{} \, \big)}{} 
\end{align*}
Moreover, since by assumption \ether{G_1} has a compositional typing, by Theorem~\ref{l:medltypes} we have:
\begin{equation*}
\G; c_\pt{p}{:}\lt{\proj{G_1}{\pt{p}}} , c_\pt{q}{:}\lt{\proj{G_1}{\pt{q}}} , 
c_\pt{r}{:}\lt{\proj{G_1}{\pt{r}}} , c_\pt{s}{:}\lt{\proj{G_1}{\pt{s}}} \vdash \ether{G_1} 
\end{equation*}
Observe that in \ether{G_1} we have the name distinctions: $\{c_\pt{p},c_\pt{q}\} \# \{c_\pt{r},c_\mathtt{s}\}$, thanks to \eqref{eq:swp1}. 
These distinctions will be useful in commuting prefixes inside \ether{G_1} in a type-preserving way, for they ensure that 
all these sessions are causally independent. 
We may now use $\eqcc$ to perform the prefix commutations shown in Fig.~\ref{f:swap}, using the equalities in 
Figs.~\ref{f:cc1} and~\ref{f:cc2}.
\begin{figure*}
\begin{align*}
\ether{G_1} & =  \mycasebig{c_\mathtt{p}}{\lb{l}: c_\mathtt{p}(u_1).\mysel{c_\mathtt{q}}{\lb{l}};\outp{c_\mathtt{q}}{v_1}.\big( \linkr{u_1}{v_1}  \para \mycase{c_\mathtt{r}}{\lb{h}: c_\mathtt{r}(u_2).\mysel{c_\mathtt{s}}{\lb{h}};\outp{c_\mathtt{s}}{v_2}.( \linkr{u_2}{v_2}  \para \ether{G} )}{} \, \big)}{} \\
& \eqcc  \mycasebig{c_\mathtt{p}}{\lb{l}: c_\mathtt{p}(u_1).\mysel{c_\mathtt{q}}{\lb{l}}; \mycasecol{c_\mathtt{r}}{\lb{h}: \outp{c_\mathtt{q}}{v_1}.( \linkr{u_1}{v_1}  \para c_\mathtt{r}(u_2).\mysel{c_\mathtt{s}}{\lb{h}};\outp{c_\mathtt{s}}{v_2}.( \linkr{u_2}{v_2}  \para \ether{G} )}{} )}{} \\
& \eqcc  \mycasebig{c_\mathtt{p}}{\lb{l}: c_\mathtt{p}(u_1).\mycasecol{c_\mathtt{r}}{\lb{h}: \mysel{c_\mathtt{q}}{\lb{l}}; \outp{c_\mathtt{q}}{v_1}.( \linkr{u_1}{v_1}  \para c_\mathtt{r}(u_2).\mysel{c_\mathtt{s}}{\lb{h}};\outp{c_\mathtt{s}}{v_2}.( \linkr{u_2}{v_2}  \para \ether{G} )}{} )}{} \\
& \eqcc  \mycasebig{c_\pt{p}}{\lb{l}: \mycasecol{c_\pt{r}}{\lb{h}: c_\pt{p}(u_1).\mysel{c_\pt{q}}{\lb{l}}; \outp{c_\mathtt{q}}{v_1}.( \linkr{u_1}{v_1}  \para c_\pt{r}(u_2).\mysel{c_\pt{s}}{\lb{h}};\outp{c_\mathtt{s}}{v_2}.( \linkr{u_2}{v_2}  \para \ether{G} )}{} )}{} \\
& \eqcc  \mycasebigcol{c_\pt{r}}{\lb{h}: \mycase{c_\pt{p}}{\lb{l}: c_\pt{p}(u_1).\mysel{c_\pt{q}}{\lb{l}}; \outp{c_\mathtt{q}}{v_1}.( \linkr{u_1}{v_1}  \para c_\pt{r}(u_2).\mysel{c_\pt{s}}{\lb{h}};\outp{c_\mathtt{s}}{v_2}.( \linkr{u_2}{v_2}  \para \ether{G} )}{} )}{} \\
& \eqcc  \mycasebig{c_\pt{r}}{\lb{h}: \mycase{c_\pt{p}}{\lb{l}: c_\pt{p}(u_1).\mysel{c_\pt{q}}{\lb{l}};\jblue{c_\pt{r}(u_2)}. \outp{c_\mathtt{q}}{v_1}.( \linkr{u_1}{v_1}  \para \mysel{c_\pt{s}}{\lb{h}};\outp{c_\mathtt{s}}{v_2}.( \linkr{u_2}{v_2}  \para \ether{G} )}{} )}{} \\
& \eqcc  \mycasebig{c_\pt{r}}{\lb{h}: \mycase{c_\pt{p}}{\lb{l}: c_\pt{p}(u_1).\jblue{c_\pt{r}(u_2)}.\mysel{c_\pt{q}}{\lb{l}}; \outp{c_\mathtt{q}}{v_1}.( \linkr{u_1}{v_1}  \para \mysel{c_\pt{s}}{\lb{h}};\outp{c_\mathtt{s}}{v_2}.( \linkr{u_2}{v_2}  \para \ether{G} )}{} )}{} \\
& \eqcc  \mycasebig{c_\pt{r}}{\lb{h}: \mycase{c_\pt{p}}{\lb{l}: \jblue{c_\pt{r}(u_2)}.c_\pt{p}(u_1).\mysel{c_\pt{q}}{\lb{l}}; \outp{c_\mathtt{q}}{v_1}.( \linkr{u_1}{v_1}  \para \mysel{c_\pt{s}}{\lb{h}};\outp{c_\mathtt{s}}{v_2}.( \linkr{u_2}{v_2}  \para \ether{G} )}{} )}{} \\
& \eqcc  \mycasebig{c_\pt{r}}{\lb{h}: \mycase{c_\pt{p}}{\lb{l}: c_\pt{r}(u_2).c_\pt{p}(u_1).\mysel{c_\pt{q}}{\lb{l}}; \jblue{\mysel{c_\pt{s}}{\lb{h}};}\outp{c_\mathtt{q}}{v_1}.( \linkr{u_1}{v_1}  \para \outp{c_\mathtt{s}}{v_2}.( \linkr{u_2}{v_2}  \para \ether{G} )}{} )}{} \\
& \eqcc  \mycasebig{c_\pt{r}}{\lb{h}: \mycase{c_\pt{p}}{\lb{l}: c_\pt{r}(u_2).c_\pt{p}(u_1).\jblue{\mysel{c_\pt{s}}{\lb{h}};}\mysel{c_\pt{q}}{\lb{l}}; \outp{c_\mathtt{q}}{v_1}.( \linkr{u_1}{v_1}  \para \outp{c_\mathtt{s}}{v_2}.( \linkr{u_2}{v_2}  \para \ether{G} )}{} )}{} \\
& \eqcc  \mycasebig{c_\pt{r}}{\lb{h}: \mycase{c_\pt{p}}{\lb{l}: c_\pt{r}(u_2).\jblue{\mysel{c_\pt{s}}{\lb{h}};}c_\pt{p}(u_1).\mysel{c_\pt{q}}{\lb{l}}; \outp{c_\mathtt{q}}{v_1}.( \linkr{u_1}{v_1}  \para \outp{c_\mathtt{s}}{v_2}.( \linkr{u_2}{v_2}  \para \ether{G} )}{} )}{} \\
& \eqcc  \mycasebig{c_\pt{r}}{\lb{h}: \mycase{c_\pt{p}}{\lb{l}: c_\pt{r}(u_2).\mysel{c_\pt{s}}{\lb{h}};c_\pt{p}(u_1).\mysel{c_\pt{q}}{\lb{l}}; \jblue{\outp{c_\mathtt{s}}{v_2}}.( \linkr{u_2}{v_2}  \para \outp{c_\mathtt{q}}{v_1}.( \linkr{u_1}{v_1}  \para \ether{G} )}{} )}{} \\
& \eqcc  \mycasebig{c_\pt{r}}{\lb{h}: \mycase{c_\pt{p}}{\lb{l}: c_\pt{r}(u_2).\mysel{c_\pt{s}}{\lb{h}};c_\pt{p}(u_1).\outp{c_\mathtt{s}}{v_2}.( \linkr{u_2}{v_2}  \para \jblue{\mysel{c_\pt{q}}{\lb{l}};}\outp{c_\mathtt{q}}{v_1}.( \linkr{u_1}{v_1}  \para \ether{G} )}{} )}{} \\
& \eqcc  \mycasebig{c_\pt{r}}{\lb{h}: \mycase{c_\pt{p}}{\lb{l}: c_\pt{r}(u_2).\mysel{c_\pt{s}}{\lb{h}}; \outp{c_\mathtt{s}}{v_2}.( \linkr{u_2}{v_2}  \para \jblue{c_\pt{p}(u_1).}\mysel{c_\pt{q}}{\lb{l}};\outp{c_\mathtt{q}}{v_1}.( \linkr{u_1}{v_1}  \para \ether{G} )}{} )}{} \\
& \eqcc  \mycasebig{c_\pt{r}}{\lb{h}: c_\pt{r}(u_2).\mycasecol{c_\pt{p}}{\lb{l}: \mysel{c_\pt{s}}{\lb{h}}; \outp{c_\mathtt{s}}{v_2}.( \linkr{u_2}{v_2}  \para c_\pt{p}(u_1).\mysel{c_\pt{q}}{\lb{l}};\outp{c_\mathtt{q}}{v_1}.( \linkr{u_1}{v_1}  \para \ether{G} )}{} )}{} \\
& \eqcc  \mycasebig{c_\pt{r}}{\lb{h}: c_\pt{r}(u_2).\mysel{c_\pt{s}}{\lb{h}};\mycasecol{c_\pt{p}}{\lb{l}:  \outp{c_\mathtt{s}}{v_2}.( \linkr{u_2}{v_2}  \para c_\pt{p}(u_1).\mysel{c_\pt{q}}{\lb{l}};\outp{c_\mathtt{q}}{v_1}.( \linkr{u_1}{v_1}  \para \ether{G} )}{} )}{} \\
& \eqcc  \mycasebig{c_\pt{r}}{\lb{h}:  c_\pt{r}(u_2).\mysel{c_\pt{s}}{\lb{h}}; \outp{c_\mathtt{s}}{v_2}.( \linkr{u_2}{v_2}  \para \mycasecol{c_\pt{p}}{\lb{l}: c_\pt{p}(u_1).\mysel{c_\pt{q}}{\lb{l}};\outp{c_\mathtt{q}}{v_1}.( \linkr{u_1}{v_1}  \para \ether{G} )}{} )}{} \\
& = \ether{G_2} 
\end{align*}
\caption{Type-preserving process transformations induced by swapping, as required in the proof of Theorem \ref{p:swapmeds}\label{f:swap}}
\end{figure*}
\end{enumerate}
\end{proof}

\section{Omitted Proofs from \S\,\ref{ss:rela}\label{ap:rec}}
\subsection{Characterization Results for Recursive Mediums}

This proof follows the lines of that given in Appendix~\ref{app:ltypesmedp} for Theorem~\ref{l:ltypesmedp}.
%

\begin{myfact}
Let $G \in \rgtypes$ be a MWF global type.
If  $G$ has a free variable $\rv{X}$ then all the projections of $G$ will have $\rv{X}$ as a free variable.
Moreover, for all $\pt{p}_i$ there is a context $K_i$ such that $\proj{G}{\pt{p}_i} = K_i[\rv{X}]$.
\end{myfact}

We repeat the statement of Theorem~\ref{l:ltypesmedprec} (Page~\pageref{l:ltypesmedprec}):
\begin{theorem}[\ref{l:ltypesmedprec}]
Let $G \in \rgtypes$ be a global type, with~$\partp{G} = \mathcal{P} = \{\pt{p}_1, \ldots, \pt{p}_n\}$.
If $G$ is 
MWF
then 
 $$\G; c_1{:}\clt{\proj{G}{\pt{p}_1}}{\eta}{\mathcal{P}}, \ldots, c_n{:}\clt{\proj{G}{\pt{p}_n}}{\eta}{\mathcal{P}} \vdash_\eta \rether{G}{k} :: k{:} A$$
  is a left compositional typing
for \rether{G}{k} 
for some $\G, \eta, A$.
\end{theorem}

\begin{proof}
By structural induction on $G$. We detail only the case $G = \mu \rv{X}.G'$, for it is the most interesting case.
(Notice that case $G = \rv{X}$ does not correspond to a valid global type.)
We need to show that, for some $\G, \eta$, and $A$ the following judgment is derivable:
\begin{equation}
 \G; c_1{:}\clt{\proj{\mu \rv{X}.G'}{\pt{p}_1}}{\eta}{\mathcal{P}}, \ldots, c_n{:}\clt{\proj{\mu \rv{X}.G'}{\pt{p}_n}}{\eta}{\mathcal{P}} \vdash_\eta \rether{\mu \rv{X}.G'}{k} :: k{:} A \label{eq:rec1}
 \end{equation}
 Using the definitions of $\clt{\cdot}{\eta}{\mathcal{P}}$, $\lt{\cdot}$, and $\rether{\cdot}{k}$, respectively, the thesis  \eqref{eq:rec1} can be equivalently stated as
 \begin{eqnarray}
 \G; c_1{:}\lt{\proj{\mu \rv{X}.G'}{\pt{p}_1}}, \ldots, c_n{:}\lt{\proj{\mu \rv{X}.G'}{\pt{p}_n}} & \vdash_\eta &\rether{\mu \rv{X}.G'}{k} :: k{:} A \label{eq:rec2} \\
  \G; c_1{:}\nu \rv{X}.\lt{\proj{G'}{\pt{p}_1}}, \ldots, c_n{:}\mu \rv{X}.\lt{\proj{G'}{\pt{p}_n}} & \vdash_\eta & \rether{\mu \rv{X}.G'}{k} :: k{:} A \label{eq:rec3} \\
    \G; c_1{:}\nu \rv{X}.\lt{\proj{G'}{\pt{p}_1}}, \ldots, c_n{:}\nu \rv{X}.\lt{\proj{G'}{\pt{p}_n}} & \vdash_\eta &  \recc\rv{X}(k).\rether{G'}{k} :: k{:} A \label{eq:rec4}
 \end{eqnarray}
 We show how to infer \eqref{eq:rec4}. 
 First, we record two facts about $G'$:
 \begin{eqnarray}
 \fv{G'} & = &  \{\rv{X}\} \label{eq:rec6} \\
 \partp{\mu \rv{X}.G'} & = & \partp{G'}\label{eq:rec7}
 \end{eqnarray}
 Now, by IH the following is a derivable typing judgment, for some $\eta_0, A_0$:
 \begin{equation}
 \G; c_1{:}\clt{\proj{G'}{\pt{p}_1}}{\eta_0}{\mathcal{P}}, \ldots, c_n{:}\clt{\proj{G'}{\pt{p}_n}}{\eta_0}{\mathcal{P}} \vdash_{\eta_0} \rether{G'}{k} :: k{:} A_0 \label{eq:rec5}
 \end{equation}
We describe how \eqref{eq:rec5} allows us to infer \eqref{eq:rec4}.

First, 
using \eqref{eq:rec6} and the definition of $\rether{\cdot}{k}$, we infer that $\fv{\rether{G'}{k}} = \{\rv{X}\}$.
Therefore, since \eqref{eq:rec5} is well-typed, by inversion we may verify that $\eta_0$ contains an entry for $\rv{X}$:
\begin{equation}
\eta_0 = \eta'[\rv{X}(k) \mapsto \G; \D_0 \vdash k:\rv{Y}]
\label{eq:rec8}
\end{equation}
for some $\eta', \D_0$. 
We show that $\D_0$ contains an assignment $c_i{:}B_i$ for all $\pt{p}_i \in G'$.
By assumption $G$ is MWF, and therefore $G'$ is MWF too. 
As such, for every $\pt{p}_i \in G'$, 
the local type $\proj{G'}{\pt{p}_i}$ is defined. 
By definition of projection and \eqref{eq:rec6},
the recursive variable $\rv{X}$ occurs in every $B_i$.
Consequently, by construction, the mapping $\D_0$ accounts for all $\pt{p}_i \in G'$.

A crucial observation is that 
while the projected type $\lt{\proj{G'}{\pt{p}_i}}$ does have free recursion variables, 
its corresponding $B_i$ does not. This is induced by typing rule~$\name{var}$.  
Hence, one of the following holds:
\begin{eqnarray}
\eta_0({\rv{X}})({c_i}) & = &  B_i = \nu \rv{X}.\lt{\proj{G'}{\pt{p}_i}} \label{eq:rec9}\\
\eta_0({\rv{X}})({c_i}) & = & B_i = \lt{\proj{G'}{\pt{p}_i}}\subst{\nu \rv{X}.\lt{\proj{G'}{\pt{p}_i}}}{\rv{X}} \label{eq:rec10}
\end{eqnarray}
That is, $\eta_0$ stores the co-recursive type of the involved projection or its  unfolding.
We can assume 
all entries in $\eta_0$ are of shape \eqref{eq:rec9}, for we may 
silently rewrite all entries with shape~\eqref{eq:rec10} using rule~$(\lft\nu)$.
Typing inversion ensures that the body of the co-recursive type is the associated projection.

Now, combining \eqref{eq:rec5} with  the definition of $\clt{\cdot}{\eta}{\mathcal{P}}$ and \eqref{eq:rec9} we infer:
 \begin{equation}
 \G; c_1{:}\lt{\proj{G'}{\pt{p}_1}}\subst{\nu \rv{X}.\lt{\proj{G'}{\pt{p}_1}} }{\rv{X}}, \ldots, c_n{:}\lt{\proj{G'}{\pt{p}_n}}\subst{\nu \rv{X}.\lt{\proj{G'}{\pt{p}_n}} }{\rv{X}} \vdash_{\eta_0} \rether{G'}{k} :: k{:} A_0 \label{eq:rec11}
 \end{equation}
 Then, using rule~$(\lft\nu)$, we may rewrite~\eqref{eq:rec11} as 
 \begin{equation}
 \G; c_1{:}\nu \rv{X}.\lt{\proj{G'}{\pt{p}_1}}, \ldots, c_n{:}\nu \rv{X}.\lt{\proj{G'}{\pt{p}_n}} \vdash_{\eta_0} \rether{G'}{k} :: k{:} A_0 \label{eq:rec12}
 \end{equation}
 Finally, using rule~$(\rgt\nu)$, \eqref{eq:rec12}, and \eqref{eq:rec8}
 we may infer 
 \begin{equation}
 \G; c_1{:}\nu \rv{X}.\lt{\proj{G'}{\pt{p}_1}}, \ldots, c_n{:}\nu \rv{X}.\lt{\proj{G'}{\pt{p}_n}} \vdash_{\eta'} \recc\rv{X}(k).\rether{G'}{} :: k{:} \nu \rv{Y}.A_0 
 \end{equation}
 This completes the proof.
\end{proof}


We repeat the statement in Page~\pageref{l:medltypesrec}

\begin{theorem}[From Well-Typedness To MWF Global Types]\label{app:medltypesrec}
Let $G \in \rgtypes$ be a 
global type, 
with~$\partp{G} = \mathcal{P} = \{\pt{p}_1, \ldots, \pt{p}_n\}$. 
If 
$$\G; c_1{:}A_1, \ldots, c_ n{:}A_ n \vdash_\eta \rether{G}{k} ::k{:}B$$
is a 
left compositional typing for \rether{G}{k}
then $\exists 
T_1, \ldots, T_ n$ 
s.t. 
$\proj{G}{\mathtt{p}_j} \subt T_j$ and 
$\clt{T_j}{\eta}{\mathcal{P}} = A_j$
for all $\pt{p}_j \in G$.
\end{theorem}

\begin{proof}
By structural induction on $G$, following the lines of the proof given in Appendix~\ref{app:medltypes} for Theorem~\ref{l:medltypes} (Page~\pageref{l:medltypes}).
\end{proof}

\subsection{Characterization Results for Annotated Mediums}
We repeat the statements in Page~\pageref{l:ltypesmedpan}:

\begin{theorem}[From Well-Formedness To Typed Annotated Mediums]
Let $G \in \rgtypes$ be a 
global type with~$\partp{G} = \mathcal{P} = \{\pt{p}_1, \ldots, \pt{p}_n\}$.
If $G$ is WF
(Def.~\ref{d:wfltypes}) then 
judgment 
$$\G; c_1{:}\clt{\proj{G}{\pt{p}_1}}{\eta}{\mathcal{P}}, \ldots, c_n{:}\clt{\proj{G}{\pt{p}_n}}{\eta}{\mathcal{P}} \vdash_\eta \raether{G}{k}{k} :: k{:} \gt{G}$$
is well-typed,  for some $\G$.
\end{theorem}

\begin{proof}[Proof (Sketch)]
By structural induction on $G$. 
The most interesting case is $G = \gto{p}{q}\{\lb{l}_i\langle U_i\rangle.G_i\}_{i \in I}$, which follows by 
extending the argument detailed in Appendix~\ref{app:ltypesmedp} for the proof of Theorem~\ref{l:ltypesmedp}, using 
the typing derivation in Figure~\ref{f:typann}, where we 
omit $\G$ and
$\D$ stands for the typing of all $\pt{p}_j \in G_i$.
\end{proof}

\begin{figure}[!t]
{\small
\[\hspace{-10ex}
\infer=[\name{T$\lft\oplus$}]{
c_\pt{p}:\lt{\proj{G}{\pt{p}}}, 
c_\pt{q}:\lt{\proj{G}{\pt{q}}},
\D
\vdash 
\raether{\gto{p}{q}\{\lb{l}_i\langle U_i\rangle.G_i\}_{i \in I}}{\lb{l}}{k} ::
k{:}\gt{G} }{
\infer=[\name{T$\rgt\oplus$}]{
c_\pt{p}: U_i \otimes A_i,
c_\pt{q}:\lt{\proj{G}{\pt{q}}},
\D
\vdash \mysel{k}{\lb{l}_i}; c_\pt{p}(u).\outp{k}{\mathsf{p}}.\big(\zero_\mathsf{p}  \para\mysel{c_\pt{q}}{\lb{l}_i};\mycase{k}{\lb{l}_i : \outp{c_\pt{q}}{v}.( \linkr{u}{v} \para k(\mathsf{q}).\raether{G_i}{\lb{l}_i}{k} )}{\{i\}}\, \big)::
k{:}\gt{G} 
}{
\infer[\name{T$\lft\otimes$}]{
c_\pt{p}: U_i \otimes A_i,
c_\pt{q}:\lt{\proj{G}{\pt{q}}},
\D
\vdash c_\pt{p}(u).\outp{k}{\mathsf{p}}.\big(\zero_\mathsf{p}  \para\mysel{c_\pt{q}}{\lb{l}_i};\mycase{k}{\lb{l}_i : \outp{c_\pt{q}}{v}.( \linkr{u}{v} \para k(\mathsf{q}).\raether{G_i}{\lb{l}_i}{k} )}{\{i\}}\, \big)::k{:}\mathtt{P} \otimes \mywith{\lb{l}_i : \mathtt{Q} \lolli \gt{G_i}}{\{i\}}}{
\infer[\name{T$\rgt\otimes$}]{
c_\pt{p}: A_i,
u{:}U_i,
c_\pt{q}:\lt{\proj{G}{\pt{q}}},
\D
\vdash \outp{k}{\mathsf{p}}.\big(\zero_\mathsf{p}  \para\mysel{c_\pt{q}}{\lb{l}_i};\mycase{k}{\lb{l}_i : \outp{c_\pt{q}}{v}.( \linkr{u}{v} \para k(\mathsf{q}).\raether{G_i}{\lb{l}_i}{k} )}{\{i\}}\, \big)::k{:}\mathtt{P} \otimes \mywith{\lb{l}_i : \mathtt{Q} \lolli \gt{G_i}}{\{i\}}}{
\vdash \zero_\mathsf{p} :: \mathsf{p}:\mathtt{P} \quad & 
\infer=[\name{T$\lft\with$}]{
c_\pt{p}: A_i,
u{:}U_i,
c_\pt{q}:\lt{\proj{G}{\pt{q}}},
\D
\vdash \mysel{c_\pt{q}}{\lb{l}_i};\mycase{k}{\lb{l}_i : \outp{c_\pt{q}}{v}.( \linkr{u}{v} \para k(\mathsf{q}).\raether{G_i}{\lb{l}_i}{k} )}{\{i\}} ::k{:}\mywith{\lb{l}_i : \mathtt{Q} \lolli \gt{G_i}}{\{i\}}}{
\infer[\name{T$\rgt\with$}]{
c_\pt{p}: A_i,
u{:}U_i,
c_\pt{q}:U_i \otimes B_i,
\D
\vdash \mycase{k}{\lb{l}_i : \outp{c_\pt{q}}{v}.( \linkr{u}{v} \para k(\mathsf{q}).\raether{G_i}{\lb{l}_i}{k} )}{\{i\}}::k{:}\mywith{\lb{l}_i : \mathtt{Q} \lolli \gt{G_i}}{\{i\}}
}{
\infer[\name{T$\lft\lolli$}]{
c_\pt{p}: A_i,
u{:}U_i,
c_\pt{q}:U_i \otimes B_i,
\D
\vdash \outp{c_\pt{q}}{v}.( \linkr{u}{v} \para k(\mathsf{q}).\raether{G_i}{\lb{l}_i}{k} )::k{:}\mathtt{Q} \lolli \gt{G_i}}{
\infer{u{:}U_i \vdash \linkr{u}{v} :: v{:}U_i}{} & 
\infer[\name{T$\rgt\lolli$}]{
c_\pt{p}: A_i,
c_\pt{q}: B_i,
\D
\vdash k(\mathsf{q}).\raether{G_i}{\lb{l}_i}{k} ::k{:}\mathtt{Q} \lolli \gt{G_i}}{
\infer[\name{T$\lft\one$}]{
c_\pt{p}: A_i,
c_\pt{q}: B_i,
\mathsf{q}:\mathtt{Q},
\D
\vdash \raether{G_i}{\lb{l}_i}{k} ::k{:}\gt{G_i}}{
\infer{
c_\pt{p}: A_i,
c_\pt{q}: B_i,
\D
\vdash \raether{G_i}{\lb{l}_i}{k} ::k{:}\gt{G_i}}{
}
}
}
}
}
}
}
}
} 
} 
\]
}
\caption{Typing for Annotated Mediums: Case $G = \gto{p}{q}\{\lb{l}_i\langle U_i\rangle.G_i\}_{i \in I}$ \label{f:typann}}
\end{figure}

\begin{theorem}[From Well-Typed Annotated Mediums To WF Global Types]
Let $G \in \rgtypes$ be a 
global type. 
If the following judgment is well-typed
$$\G; c_1{:}A_1, \ldots, c_ n{:}A_ n \vdash \raether{G}{k}{k} :: k:A_0$$
then 
$\gt{G} \subts A_0$ and 
$\exists T_1, \ldots, T_ n$ 
s.t. 
$\proj{G}{\mathtt{r}_j} \subt T_j$ and 
$\clt{T_j}{\eta}{\mathcal{P}} = A_j$
for all $\pt{r}_j \in G$.
\end{theorem}

\begin{proof}
By structural induction on $G$.
\end{proof}

\subsection{Operational Correspondence via Annotated Mediums\label{app:opcorr}}
  The following property follows directly from the definition of annotated mediums 
(Definition~\ref{d:raether})
and systems (Definition~\ref{d:systems}) as well as from the properties of typed processes.

\begin{proposition}\label{p:opcorrt}
Let $G$ be a MWF global type, and let $\mathcal{S}^k(G)$ be the systems implementing $G$.
We have:
\begin{enumerate}[1.]
\item for all $P_j \in \mathcal{S}^k(G)$, we have $\G; \cdot \vdash P_j :: k{:} \gt{G}$, for some $\G$.
\item for all $P_j \in \mathcal{S}^k(G)$, we have that if $P_j \tra{\,\labelset\,} P'$ with
 $\labelset \neq \tau$ then $subj(\lambda) = k$.
\end{enumerate}
\end{proposition}

We repeat the statement given in Page~\pageref{th:opcorr}:

\begin{theorem}[Global Types and Mediums: Operational Correspondence]\label{app:opcorr}
Let $G$ be a MWF global type and $P$ any process in $\mathcal{S}^k(G)$. We have:
\begin{enumerate}[(1)]
\item If $G \tra{\,\gtlabelset\,} G'$ then 
there exist $\labelset, P'$ such that $P \wtra{\,\labelset\,} P'$,
$\labelset = \mlab{\gtlabelset}{k}$, 
 and $P' \in \mathcal{S}^k(G')$.
 \item If there is some $P_0$ such that $P \wtra{} P_0 \tra{\,\labelset\,} P'$ with
 $\labelset \neq \tau$ then there exist $\gtlabelset,  G'$ such that $G \tra{\,\gtlabelset\,} G'$, 
 $subj(\gtlabelset) = \pt{p}$,
 $ \gtlabelset = \glab{\labelset}{\pt{p}} $, 
 and $P' \in \mathcal{S}^k(G')$.
\end{enumerate}
\end{theorem}

\begin{proof}[Proof]
Part (2) is easy, relying on Proposition~\ref{p:opcorrt} and on the definitions of system and annotated mediums.
We notice that the (finite) reduction sequence leading to $P_0$ (i.e., preceding the observable action on $\labelset$) must necessarily involve at least one synchronization between the annotated medium and its closure. 
As for part (1), we consider only the case  in which the global transition is realized via rule
$$
\inferrule*[Left=\name{G1}]{}{\gto{p}{q}\{\lb{l}_i\langle U_i\rangle.G_i\}_{i \in I}  \tra{\,\ov{\mysel{\pt{p}}{~\lb{l}_j}}\,}  \igto{p}{q}\lb{l}_j\langle U_j\rangle.G_j}  \quad (j \in I) 
$$
since other cases are similar.
Writing $\pt{p}_1$ and $\pt{p}_2$ instead of $\pt{p}$ and $\pt{q}$, by Definition~\ref{d:systems} we have:
\begin{equation*}
P  =   (\nub c_{\pt{p_1}}, \ldots, c_{\pt{p_n}})(Q_{\pt{p_1}} \para Q^* \para \raether{\gto{p_1}{p_2}\{\lb{l}_i\langle U_i\rangle.G_i\}_{i \in I}}{k}{k}) 
\end{equation*}
where
$Q^* = Q_{\pt{p_2}} \para \cdots \para Q_{\pt{p_n}}$ 
and for all $i \in \{1, \ldots, n\}$ Definition~\ref{d:elrauxc} ensures that
\begin{equation*}
\G ; \cdot  \vdash   Q_{{\pt{p_i}}} :: c_{{\pt{p_i}}}{:}\lt{\proj{G}{\pt{p}_i}}
\end{equation*}
In particular, 
well-typedness of $\raether{\gto{p_1}{p_2}\{\lb{l}_i\langle U_i\rangle.G_i\}_{i \in I}}{k}{k}$
ensures that 
$$\G ; \cdot  \vdash   Q_{{\pt{p_1}}} :: c_{{\pt{p_1}}}{:}\myoplus{\lb{l}_i : U_i \otimes B_i}{i \in I}$$
We now have a weak transition (which is finite, due to Theorem~\ref{th:prop}(3)):
\begin{eqnarray*}
P  & \wtra{} & (\nub c_{\pt{p_1}}, \ldots, c_{\pt{p_n}})(Q'_{\pt{p_1}} \para \cdots \para Q^*_{1} \para \raether{\gto{p_1}{p_2}\{\lb{l}_i\langle U_i\rangle.G_i\}_{i \in I}}{k}{k}) = P_1
\end{eqnarray*}
where typing preservation (Theorem~\ref{th:prop}(1)) ensures that $Q'_{\pt{p_1}} \tra{\ov{\mysel{c_{\pt{p_1}}}{~\lb{l}_j}}} Q''_{\pt{p_1}}$. Then,
Definition~\ref{d:raether} ensures 
a synchronization on name $c_{\pt{p_1}}$; by expanding the definition of annotated medium we obtain:
\begin{eqnarray*}
P_1  & \tra{\,\tau\,} & (\nub c_{\pt{p_1}}, \ldots, c_{\pt{p_n}})(Q''_{\pt{p_1}} \para \cdots \para Q^*_{1} \para \mysel{k}{\lb{l}_j};\raether{ \igto{p_1}{p_2}\lb{l}_j\langle U_j\rangle.G_j }{\lb{l}}{k}) \\
& \tra{\ov{\mysel{k}{\, \lb{l}_j}}} & (\nub c_{\pt{p_1}}, \ldots, c_{\pt{p_n}})(Q''_{\pt{p_1}} \para \cdots \para Q^*_{1} \para \raether{ \igto{p_1}{p_2}\lb{l}_j\langle U_j\rangle.G_j }{\lb{l}}{k}) = P'
\end{eqnarray*}
and by Definition~\label{d:labels} we have that 
$\mlab{\,\ov{\mysel{\pt{p}}{\lb{l}}} \,}{k}  =  \ov{\mysel{k}{\lb{l}}}$.
Finally, 
we observe that by 
performing a selection action, we obtain that the offer of $Q''_{{\pt{p_1}}}$ evolves: 
$\G ; \cdot  \vdash   Q''_{{\pt{p_1}}} :: c_{{\pt{p_1}}}{:}U_j \otimes B_j$.
Therefore, by expanding the definition of system,  
we may infer $P' \in  \mathcal{S}^k(\igto{p_1}{p_2}\lb{l}_j\langle U_j\rangle.G_j)$.
\end{proof}

%
%
%


 

%
%

\end{document}